\documentclass[10pt,journal,twoside]{IEEEtran}

\usepackage{cite}
\usepackage{amsmath}
\usepackage{amssymb}
\usepackage{amsthm}
\usepackage{graphicx}

\newtheorem{prop}{Proposition}
\newtheorem{lemma}{Lemma}
\theoremstyle{definition}
\newtheorem*{replicatrick}{Replica Trick}
\newtheorem{defn}{Definition}
\newtheorem{example}{Example}
\newtheorem{remark}{Remark}

\newcommand{\tr}{\mathop{\rm tr\/}}
\newcommand{\diag}{\mathop{\rm diag\/}}
\newcommand{\e}{\mathrm{e}}
\newcommand{\E}{\mathsf{E}}
\newcommand{\C}{\mathbb{C}}
\newcommand{\R}{\mathbb{R}}
\newcommand{\Qfunc}{\mathcal{Q}}

\newcommand{\vm}[1]{\boldsymbol{#1}}

\newcommand{\trans}{\mathsf{T}}
\newcommand{\im}{\mathrm{i}}
\newcommand{\dx}{\mathrm{d}}
\newcommand{\Dx}{\mathrm{D}}
\newcommand{\extr}{\mathop{\rm extr\/}}

\newcommand{\mse}{\mathsf{mse}}
\newcommand{\argmin}{\mathop{\rm arg\hspace{2pt}min\/}}

\newcommand{\PM}{p}
\newcommand{\QM}{q}
\newcommand{\Gpdf}[1]{g_{#1}}
\newcommand{\I}{\vm{I}}
\newcommand{\tmpfuncA}{a}

\newcommand{\Dproj}{\vm{P}}

\newcommand{\nuAvg}{\nu}
\newcommand{\Mx}{\hat{\M}}
\newcommand{\M}{M}
\newcommand{\N}{N}
\newcommand{\Nidx}{j}
\newcommand{\Midx}{i}

\newcommand{\T}{T}
\newcommand{\Tx}{T}
\newcommand{\tx}{t}
\newcommand{\Freg}{c}
\newcommand{\sigmanorm}{a(\kappa)}
\newcommand{\sigmanormC}{A(\kappa)}

\newcommand{\NR}{n}
\newcommand{\nr}{a}
\newcommand{\nrb}{b}

\newcommand{\Qhat}{\hat{Q}}
\newcommand{\mhat}{\hat{m}}
\newcommand{\chihat}{\hat{\chi}}
\newcommand{\Qmathat}{\hat{\vm{Q}}}
\newcommand{\xhat}{\hat{x}}
\newcommand{\xvechat}{\hat{\vm{x}}}

\begin{document}

\title{Analysis of Regularized LS Reconstruction and Random Matrix Ensembles in Compressed Sensing}

\author{\IEEEauthorblockN{Mikko Vehkaper{\"a}~\IEEEmembership{Member,~IEEE}, 
		Yoshiyuki Kabashima, and Saikat Chatterjee~\IEEEmembership{Member,~IEEE}}%
	\thanks{
		Manuscript received December 1, 2013; revised November 6, 2015; 
		accepted January 25, 2016.
		The editor coordinating the review of this manuscript and approving 
		it for publication was Prof.\ Venkatesh Saligrama.
		The research was funded in part by
		Swedish Research Council under VR Grant 621-2011-1024 (MV)
		and MEXT KAKENHI Grant No. 25120013 (YK).
		M. Vehkaper{\"a}'s visit to Tokyo Institute of Technology was 
		funded by the MEXT KAKENHI Grant No. 24106008.	
		This paper was presented in part at the 2014 
		IEEE International Symposium on Information Theory.			
		}%
	\thanks{
		M. Vehkaper\"a was with the KTH Royal Institute of Technology, 
		Sweden and Aalto University, Finland.  
		He is now with the Department of Electronic and Electrical Engineering,  
		University  of  Sheffield, Sheffield S1 3JD, UK.
		(e-mail: m.vehkapera@sheffield.ac.uk)}%
	\thanks{
		Y. Kabashima is with the Department of Computational Intelligence and Systems Science,  
		Tokyo Institute of Technology, Yokohama 226-8502, Japan.
		(e-mail: kaba@dis.titech.ac.jp)}%
	\thanks{  
		S. Chatterjee is with the School of Electrical Engineering and the ACCESS Linnaeus Center, 
		KTH Royal Institute of Technology, SE-100 44 Stockholm, Sweden.
		(e-mail: sach@kth.se)}
}

\IEEEpubid{\begin{minipage}{\textwidth}\ \\[12pt]
		\centering
		Copyright (c) 2014 IEEE. Personal use of this material is permitted. 
	However, permission to use this material for any other purposes \\
	must be obtained from the IEEE by sending a request to pubs-permissions@ieee.org. 
	\end{minipage}}
	
\markboth{IEEE Transactions on Information Theory}{Vehkaper{\"a}~\MakeLowercase{\textit{et al.}}: 
	Analysis of Regularized LS Reconstruction and Random Matrix Ensembles in Compressed Sensing}

\maketitle


\begin{abstract}
	Performance of regularized least-squares estimation in noisy compressed 
	sensing is analyzed in the limit when the dimensions of 
	the measurement matrix grow large. The sensing matrix is considered to be from 
	a class of random ensembles that encloses as special cases 
	standard Gaussian, row-orthogonal, geometric and so-called $T$-orthogonal 
	constructions.  Source vectors that have non-uniform sparsity 
	are included in the system model.  Regularization based on $\ell_{1}$-norm
	and leading to LASSO estimation, or basis pursuit denoising,
	is given the main emphasis in the analysis.  Extensions to 
	$\ell_{2}$-norm and ``zero-norm'' regularization are also 
	briefly discussed. 
	The analysis is carried out using the replica method in conjunction
	with some novel matrix integration results.  Numerical experiments
	for LASSO are provided to verify the accuracy of the analytical results.
	
	The numerical experiments show that for 
	noisy compressed sensing, the standard Gaussian ensemble is a 
	suboptimal choice for the measurement matrix. 
	Orthogonal constructions provide a superior
	performance in all considered scenarios and are easier to implement in
	practical applications.  It is also discovered 
	that for non-uniform sparsity patterns the $T$-orthogonal 
	matrices can further improve the mean square error behavior of the 
	reconstruction when the noise level is not too high.
	However, as the additive noise becomes more prominent in the system,
	the simple row-orthogonal measurement matrix appears to be the best choice 
	out of the considered ensembles.  
\end{abstract}



\section{Introduction}

\IEEEPARstart{C}{onsider}
the standard compressed sensing (CS)
\cite{Donoho_2006_Compressed_sensing, 
Candes_Romberg_Tao_2006, Candes_Tao_NearOptimal_2006} setup
where the sparse vector $\vm{x}^{0} \in \mathbb{R}^{\N}$ of 
interest is observed via noisy linear measurements
\begin{equation}
	\vm{y} = \vm{A} \vm{x}^{0} + \sigma\vm{w},
	\label{eq:Sparse_Representation_with_Noise}
\end{equation}
where $\vm{A} \in \R^{\M \times \N}$ represents the compressive
$(\M \leq \N)$ sampling system.  Measurement errors are 
captured by the vector $\vm{w} \in \R^{\M}$ and parameter 
$\sigma$ controls the magnitude of the distortions.
The task is then to infer $\vm{x}^{0}$ from $\vm{y}$, given 
the measurement matrix $\vm{A}$.
Depending on the chosen performance metric, the level of knowledge about the 
statistics of the source and error vectors, or computational complexity 
constraints, multiple choices are available for achieving this task.
One possible solution that does not require detailed information about
$\sigma$ or statistics of 
$\{\vm{x}^{0}, \vm{w}\}$ is regularized least-squares (LS) based 
 reconstruction
\begin{equation}
	\xvechat= 
	\argmin_{\vm{x} \in \mathbb{R}^{N}} \hspace{3pt} 
	\left\{\frac{1}{2 \lambda}\| \vm{y} -  \vm{A} \vm{x} \|^{2}
	+ \Freg(\vm{x})
	\right\},
	\label{eq:CS_Standard_Problem_LS}
\end{equation}
where $\| \cdot \|$ is the standard Euclidean norm,
$\lambda$ a non-negative design parameter and 
$\Freg: \R^{\N} \to \R$ a fixed non-negative valued (cost) function.
If we  interpret \eqref{eq:CS_Standard_Problem_LS} as a 
maximum a posteriori probability (MAP) estimator, the implicit assumption
would be that: 1) the additive noise can be 
modeled by a zero-mean Gaussian random vector with covariance 
$\lambda \I_{\M}$, and 2) the distribution of the source 
is proportional to $\e^{-\Freg(\vm{x})}$.
Neither is in general true for the model
\eqref{eq:Sparse_Representation_with_Noise} and, therefore,
reconstruction based on \eqref{eq:CS_Standard_Problem_LS} 
is clearly suboptimal.

\IEEEpubidadjcol

In the sparse estimation framework, the  purpose of the cost 
function $\Freg$ is to penalize the trial $\vm{x}$ so that some desired 
property of the source is carried over to the solution
$\xvechat$.  In the special case when the measurements are \emph{noise-free}, 
that is, $\sigma = 0$, the choice $\lambda\to 0$ reduces 
\eqref{eq:CS_Standard_Problem_LS}
to solving a constrained optimization problem
\begin{equation}
	\min_{\xvechat \in \mathbb{R}^{N}} \hspace{3pt} \Freg(\xvechat)
	\quad \text{s.t.}
	\quad \vm{y} = \vm{A} \xvechat.
	\label{eq:CS_Standard_Problem_LS_noisefree}
\end{equation}
It is well-known that in the noise-free case
the $\ell_{1}$-cost $\Freg(\vm{x}) = \|\vm{x}\|_{1} = \sum_{\Nidx} |x_{\Nidx}|$
leads to sparse solutions that can be found using linear programming.
For the noisy case the resulting scheme is called
LASSO \cite{Tibshirani_1996_lasso} or basis pursuit denoising
\cite{chen1998atomic} 
\begin{equation}
	\xvechat_{\ell_1} = 
	\argmin_{\vm{x} \in \mathbb{R}^{N}} \hspace{3pt} 
	\left\{\frac{1}{2 \lambda}\| \vm{y} -  \vm{A} \vm{x} \|^{2}
	+ \| \vm{x} \|_{1}
	\right\}.
	\label{eq:CS_Standard_Problem_LASSO}
\end{equation}
Just like its noise-free counterpart, it is of particular 
importance in CS since \eqref{eq:CS_Standard_Problem_LASSO} can be solved by using 
standard convex optimization tools such as 
\texttt{cvx} \cite{cvx}.  
Due to the prevalence of 
reconstruction methods based on $\ell_{1}$-norm
regularization in CS, we shall keep the special case
of $\ell_{1}$-cost $\Freg(\vm{x}) = \|\vm{x}\|_{1}$  as the main 
example of the paper, although it is known to be a suboptimal 
choice in general.

\subsection{Brief Literature Review}

In the literature, compressed sensing has 
a strong connotation of \emph{sparse representations}.  We shall 
next provide a brief review of the CS literature while keeping this in mind.
The theoretical works in CS can be roughly divided  
into two principle directions: 1) \emph{worst case} analysis, and 
2) \emph{average / typical case} analysis. 
In the former approach, analytical tools that examine the algebraic properties of 
the sensing matrix $\vm{A}$, such as, mutual coherence, spark or 
restricted isometry property (RIP) are used.  The goal is then to find
sufficient conditions for the chosen property of $\vm{A}$ that guarantee
perfect reconstruction --- at least with high probability.
The latter case usually strives for sharp conditions when the 
reconstruction is possible when $\vm{A}$ is sampled from some 
random distribution.  Analytical tools vary from combinatorial geometry
to statistical physics methods.  Both, worst case and average case 
analysis have their merits and flaws as we shall discuss below.

For mutual coherence, several works have considered 
the case of noise-free observations ($\sigma=0$) and
$\ell_1$-norm minimization based reconstruction.
The main objective is usually to find the conditions that need
to be satisfied between the allowed sparsity level of $\vm{x}$ and 
the mutual coherence property of $\vm{A}$ so that exact reconstruction
is possible.
In particular, the authors of~\cite{Donoho_Huo_2001} established such conditions for 
the special case when $\vm{A}$ is constructed by concatenating a pair of orthonormal bases.  
These conditions were further refined in \cite{Elad_Bruckstein_2002} and the extension
to general matrices was reported in \cite{Donoho_Elad_2003} using the concept 
of \emph{spark}. 

Another direction in the worst case analysis was taken in \cite{Candes_Tao_2005}, where 
the basic setup \eqref{eq:Sparse_Representation_with_Noise} with 
\emph{sparse additive noise} was considered. 
The threshold for exact reconstruction under these conditions was derived
using  RIP.  By establishing a connection between the Johnson-Lindenstrauss lemma 
and RIP, the authors of \cite{Baraniuk_asimple_proof_RIP_2007} proved later that
RIP holds with high probability when $\M$ grows large for a certain 
class of random matrices.  Special cases of this ensemble are, for example,
matrices whose components are 
independent identically distributed (IID) Gaussian or Bernoulli random variables.
This translates roughly to a statement that such matrices are ``good'' for CS problems
when $\ell_{1}$-norm based penalty is used if the system size is sufficiently large.

In addition to the basic problem stated above, mutual coherence and RIP based worst 
case analysis are prominent also in the study of greedy CS algorithms and fusion strategies.
Some examples are analysis of orthogonal matching pursuit 
\cite{Tropp_2007_OMP,Davenport_2010_Orthogonal_Matching_pursuit, Cai_OMP_TIT_2011}, 
subspace pursuit \cite{Dai_2009_Subspace_pursuit}, 
CoSaMP \cite{Needell_Tropp_2009_CoSaMP}, group LASSO \cite{Lv_Bi_Wan_2011_Group_Lasso} 
and Fusion strategy \cite{Ambat_Chatterjee_Hari_FACS_2013}.  
The general weakness of these approaches is, however, that if one is interested in 
\emph{typical} or \emph{average} case performance, the results provided by the worst 
case analysis are often very pessimistic and loose.  This consideration is tackled
by the second class of analytical results we mentioned at the beginning of the 
review.

In a series of papers, the 
authors of \cite{Donoho-Tanner-2009, Donoho-Tanner-2010,Donoho-Tanner-PrIEEE2010} 
used tools from combinatorial geometry to show that in the limit of increasing
system size, the $\ell_{1}$-reconstruction has a sharp \emph{phase transition} when the measurements
are noise-free.  A completely different approach based on approximate message passing (AMP)
algorithm \cite{Donoho-Maleki-Montanari-2009,Montanari-graphical-2012} 
was introduced in \cite{Bayati-Montanari-2011} and shown to match the combinatorial
results perfectly. Both of the above methods are mathematically rigorous and 
the AMP approach has the additional benefit that 
it provides also a low-complexity computational algorithm 
that matches the threshold behavior.
The downside is that extending these analysis 
for more general ensembles, both for the measurement matrix and 
the source vector, seems to be quite difficult. 
Alternative route is to use statistical mechanics inspired tools
like the \emph{replica method} 
\cite{mezard1987spin, Dotsenko-2001,Nishimori-2001}.

By now the replica method has been accepted in the  
information theory society as a mathematical tool that can 
tackle problems that are very difficult, or impossible,
to solve using other (rigorous) approaches. 
Although the outcomes of the replica analysis have 
received considerable success (see, e.g.,
\cite{Rangan-Fletcher-Goyal-2012, Guo-Baron-Shamai-2009,
Tulino-etal-cs2013, Vehkapera_Kabashima_Chatterjee_et_all_2012_ITW,
Kabashima-Vehkapera-Chatterjee-2012,
Tanaka_Raymond_2010, Kabashima-Wadayama-Tanaka-2009} for some 
results related to the present paper), one should 
keep in mind that mathematical rigor is still lacking 
in parts of the method \cite{Talagrand-2003}.
However, recent results in mathematical physics
have provided at least circumstantial evidence that the 
main problem of the replica method is most likely in the assumed 
structure of the solution \cite{Guerra-Toninelli-2002, 
Guerra-Toninelli-2002-2, Guerra-2003, Talagrand-2003, 
Talagrand-2006} and not in the parts such as replica continuity
that lack mathematical proof.
In particular, the mistake in the original solution of the 
Sherrington-Kirkpatrick spin glass has now been traced to 
the assumption of \emph{replica symmetric} (RS) ansatz in 
the saddle-point evaluation of the free energy.  Indeed, 
the end result of the Parisi's full \emph{replica symmetry breaking} 
(RSB) solution (see, e.g., \cite{mezard1987spin}) has been proved to 
be correct \cite{Guerra-2003, Talagrand-2006} in this case.   
Similar rigorous methods have also been applied in 
wireless communications \cite{Korada-Macris-2010} and 
error correction coding \cite{Montanari-2005, Kudekar-Macris-2009}, 
to name just a few examples%
\footnote{To avoid the misconception that these methods have 
made non-rigorous approaches obsolete, some comments are in place.
Firstly, the scope of the rigorous methods tend
to be much more limited than that of the non-rigorous ones.  Secondly, the 
analysis typically give bounds for the quantities of interest 
rather than sharp predictions.  Thirdly, it is often helpful to know the end-result 
obtained through some non-rigorous way, like the replica method, before 
applying the mathematically exact tools on the problem.}.

\subsection{Related Prior Work}

The authors of \cite{Rangan-Fletcher-Goyal-2012} analyzed 
the asymptotic performance of LASSO and ``zero-norm'' 
regularized LS by extending the minimum mean square error 
(MMSE) estimation problem in code division multiple access (CDMA) 
to MAP detection in linear vector models.
More specifically, the MMSE formulas obtained with the replica method 
\cite{Tanaka-2002,Guo-Verdu-2005Jun} were first assumed to
be valid and then transformed to the case of MAP decoding
through ``hardening''.  Unfortunately, this approach was 
limited to the cases where the appropriate MMSE formulas 
already existed and the end result of the analysis still required quite a 
lot of numerical computations.  
The scope of the analysis was extended to a more general 
class of random matrices
by employing the Harish-Chandra-Itzykson-Zuber (HCIZ) integral formula 
\cite{Harish-Chandra-1957, Itzykson-Zuber-1980} 
in \cite{Tulino-etal-cs2013}.
Although the emphasis there was in the 
\emph{support recovery}, also the MSE could be inferred 
from the given results.
A slightly different scenario when the additive noise is sparse
was analyzed in \cite{Wright-Ma-2010, Vehkapera-Kabashima-Chatterjee-2013}.
For such a measurement model, 
if one replaces the squared $\ell_{2}$-norm distance in 
\eqref{eq:CS_Standard_Problem_LS} by $\ell_{1}$-norm and uses
also $\ell_{1}$-regularization, perfect
reconstruction becomes sometimes feasible
\cite{Wright-Ma-2010,Vehkapera-Kabashima-Chatterjee-2013}. 
It is also possible to characterize
the MSE of reconstruction outside of this region 
using the replica method \cite{Vehkapera-Kabashima-Chatterjee-2013}.

The references above left the question open how the 
choice of measurement matrix affects the fidelity of the
reconstruction in the noisy setup.  In \cite{Wu-Verdu-2012} a partial answer
was obtained through information theoretic analysis.  The 
authors showed that standard Gaussian sensing matrices
incurred no loss in the noise sensitivity threshold 
if optimal encoding and decoding 
were used.  Similar result was obtained
earlier using the replica method in 
\cite{Guo-Baron-Shamai-2009}, and extended 
to more general matrix ensembles
in the aforementioned paper \cite{Tulino-etal-cs2013}.
On the other hand, generalization of the Lindeberg principle
was used by the authors of \cite{Korada-Montanari-2011}
to show that the average cost in LASSO was universal
for a class of matrices of \emph{standard type}.

Based on the above results and the knowledge that for the noise-free 
case the perfect reconstruction threshold is quite universal
\cite{Donoho-Tanner-2010, Donoho-Tanner-PrIEEE2010,
Kabashima-Wadayama-Tanaka-2009}, one might be tempted to 
conclude that using sensing matrices that are sampled from the standard 
Gaussian ensemble is the optimal choice 
also in the noisy case when practical algorithms such as LASSO are used.
However, there is also some counter-evidence in other settings, such 
as the noise-free case with non-uniform sparsity 
\cite{Vehkapera_Kabashima_Chatterjee_et_all_2012_ITW,Kabashima-Vehkapera-Chatterjee-2012}
and spreading sequence design in CDMA 
\cite{viswanath1999optimal, Kitagawa-Tanaka-2010} that leave 
the problem still interesting to investigate in more detail%
\footnote{After the initial submission of the present paper, parallel studies
using completely different mathematical methods and
arguing for the superiority of the orthogonal constructions
have been presented in \cite{Oymak-Hassibi-isit2014} and 
\cite{Thrampoulidis-Hassibi-isit2015}.  Since then,
an extension to the present paper has been proposed in 
\cite{Wen-etal-arxiv2014} and
iterative algorithms approximating Bayesian optimal estimation
for structured matrices
have been devised, see for example,
\cite{Kabashima-Vehkapera-isit2014, Cakmak-Winther-Fleury-itw2014, 
	Ma-Yuan-Ping-2015, Opper-Cakmak-Winther-arxiv2015}.}.
	
Albeit from a slightly different motivational point-of-view,
similar endeavor was taken earlier in 
\cite{Kudekar-Pfister-Allerton-2010,
Donoho-Javanmard-Montanari-isit2012,
Krzakala-etal-CS-2012-2, Krzakala-etal-CS-2012}, where it was discovered 
that measurement matrices with specific structure are beneficial
for  message passing decoding in noise-free settings.  These spatially 
coupled, or seeded, measurement matrices helped the 
iterative algorithm to get past local extrema and hence improved
the perfect reconstruction threshold of $\ell_{1}$-recovery significantly.  
Such constructions,
however, turned out to be detrimental for convex relaxation based
methods when compared to the standard Gaussian ensemble.

Finally we remark that the uniform sparsity model studied in 
\cite{Kabashima-Wadayama-Tanaka-2009}
was extended to a non-uniform noise-free setting in \cite{Tanaka_Raymond_2010}.   
The goal there was to optimize
the recovery performance using weighted $\ell_{1}$-minimization when
the sparsity pattern is known.  We deviate from those 
goals by considering a noisy setup with a more general 
matrix ensemble for measurements.  On the other hand, we 
do not try to optimize the reconstruction with block-wise
adaptive weights and leave such extensions as future research topics.

\subsection{Contribution and Summary of Results}

The main goal of the present paper is to extend the scope of 
\cite{Rangan-Fletcher-Goyal-2012} and
\cite{Tulino-etal-cs2013} to a wider range of matrix ensembles 
and to non-uniform sparsities of the vector of interest.
We deviate from the approach of 
\cite{Rangan-Fletcher-Goyal-2012,Tulino-etal-cs2013}
by evaluating the performance directly using the replica method
as in \cite{Kabashima-Wadayama-Tanaka-2009,
Tanaka_Raymond_2010,
Vehkapera_Kabashima_Chatterjee_et_all_2012_ITW,
Kabashima-Vehkapera-Chatterjee-2012}.
The derivations are also akin to some earlier works on linear models
\cite{Takeda-Uda-Kabashima-2006,Kabashima-confser2008}.
After obtaining the results for 
$\ell_{1}$-regularization, we sketch how they can be 
generalized to other cases like $l_{2}$-norm and ``zero-norm'' based regularization.

The analysis show that under the assumption of 
RS ansatz (for details, see Section~\ref{sec:extensions_and_sketch}), 
the average MSE of reconstruction is obtained via a system of coupled 
fixed point equations that can be solved numerically.  For the $T$-orthogonal case, 
we find that the solution depends on the sparsity pattern (how the non-zero 
components are located block-wise in the vector) of the 
source --- even when such knowledge is not used in the 
reconstruction.  In the case of \emph{rotationally invariant} ensemble,
the results are obtained as a function of the Stieltjes transform 
of the eigenvalue spectrum that describes the measurement matrix.
For this case only the total sparsity of the source vector has
influence on the reconstruction performance.
The end results for the rotationally invariant case are also shown to 
be equivalent to those in \cite{Tulino-etal-cs2013},
bridging the gap between two different approaches to replica analysis.

Finally, solving the MSE of the replica analysis for some
practical settings reveals that the standard Gaussian ensemble
is suboptimal as a sensing matrix when the system is 
corrupted by additive noise.  For example, a random 
row-orthogonal measurement matrix provides uniformly
better reconstructions compared to the Gaussian one.  This is in
contrast to the noise-free case where it is well known that
the perfect reconstruction threshold is the same for 
the whole rotationally invariant ensemble 
(see, e.g., \cite{Kabashima-Wadayama-Tanaka-2009}).
On the other hand, albeit $T$-orthogonal measurement  
matrices are able to offer lower MSE than any other 
ensemble we tested when the sparsity of the source is 
not uniform, the effect diminishes as the noise level 
in the system increases.  This may be intuitively 
explained by the fact that the additive noise in the
system makes it more difficult to differentiate between
blocks of different sparsities when we have no prior
information about it.

\subsection{Notation and Paper Outline}
\label{sec:notation}

Boldface symbols denote (column) vectors 
and matrices. Identity matrix of size $\M \times \M$ is written
$\I_{\M}$ and the transpose of matrix $\vm{A}$ as $\vm{A}^{\trans}$.
Given a variable $x_{k}$ with a countable index 
set $\mathcal{K}$, we abbreviate
$\{x_{k}\} = \{x_{k} : k\in\mathcal{K}\}$.
We write $\im = \sqrt{-1}$ and
for some (complex) function $f(z)$, denote
$f(z_{0}) = \extr_{z} f(z)$ where $z_{0}$ is an
extremum of the function $f$, that is,
satisfies $\frac{\dx f}{\dx z}\big|_{z_{0}} = 0$.
Analogous definition holds for functions of multiple
variables.  The indicator function satisfies $1(A) = 1$ if $A$ is true
and is zero otherwise.  Dirac's delta function is written
$\delta(x)$ and the Kronecker symbol $\delta_{ij}$.

Throughout the paper we assume for simplicity that given any
continuous (discrete) random variable, the respective
probability density (probability mass) function exists.
Same notation is used for both cases, and given a general
continuous / discrete random variable (RV), we often refer to probability 
density function (PDF) for brevity.
The \emph{true} and \emph{postulated} PDF of a 
random variable is denoted $\PM$ and $\QM$,
respectively.  If $\vm{x}$ is a real-valued Gaussian RV 
with mean $\vm{\mu}$ and covariance 
$\vm{\Sigma}$, we write the density of 
$\vm{x}$ as $p(\vm{x})=\Gpdf{\vm{x}}(\vm{\mu};\,\vm{\Sigma})$.

The rest of the paper is organized as follows.  The problem formulation and 
brief introduction to the replica trick is given in Section~\ref{sec:formulation}.
Section~\ref{sec:lasso_results} provides the end-results 
of replica analysis for LASSO estimation.
This case is also used in the detailed replica analysis provided in
Appendices~\ref{sec:Tortho_replica}~and~\ref{sec:eigen_replica}. 
Sketch of the main steps involved in the replica analysis and 
comparison to existing results are given in Section~\ref{sec:extensions_and_sketch}
for the rotationally invariant setup. 
Conclusions are provided in 
Section~\ref{sec:conclusions} and two matrix integral results
used as a part of the replica analysis are proved in 
Appendix~\ref{sec:matrix_integrals}.  Finally, 
Appendix~\ref{sec:geometric_ensemble} provides the details of the
geometric ensemble.


\section{Problem Formulation and Methods}
\label{sec:formulation}

Consider the CS setup \eqref{eq:Sparse_Representation_with_Noise} 
and assume that the elements of $\vm{w}$ are 
IID standard Gaussian random variables, so that
\begin{equation}
	\label{eq:true_conditional_measurement_pdf}
	p(\vm{y} \mid \vm{A}, \vm{x}^{0}) 
	= \Gpdf{\vm{y}}(\vm{A} \vm{x}^{0};\,\sigma^{2}\I_{\M})
\end{equation}
is the conditional PDF of the observations.
Recall that the notation $\vm{x}^{0}$ means here that 
the observation \eqref{eq:Sparse_Representation_with_Noise} 
was generated as 
$\vm{y} = \vm{A} \vm{x}^{0} + \sigma\vm{w}$, that is,
$\vm{x}^{0}$ is the true vector generated by the source.
Note that in this setting the additive noise is dense 
and, therefore, perfect reconstruction is in general 
not possible \cite{Wright-Ma-2010,Vehkapera-Kabashima-Chatterjee-2013}.
Let the sparse vector of interest $\vm{x}^{0}$ be 
partitioned into $T$ equal length parts 
$\{\vm{x}^{0}_{t}\}_{t=1}^{T}$ that are statistically
independent.  
The components in each of the blocks $t=1,\ldots,T$ are drawn
IID according to the mixture distribution 
\begin{equation}
	p_{\tx}(x) = 
	(1-\rho_{\tx}) \delta(x) +
	\rho_{\tx} \pi(x),
	\quad \tx = 1,\ldots,\Tx,
	\label{eq:true_sourcesym_pdf_k}
\end{equation}
where $\rho_{t} \in [0,1]$ is the expected fraction of 
non-zero elements in $\vm{x}^{0}_{t}$ that are drawn 
independently according to $\pi(x)$.
The expected density, or sparsity, of the whole signal is
thus $\rho = \Tx^{-1}\sum_{\tx} \rho_{\tx}$.
We denote the true prior according to which the 
data is generated by $\PM(\vm{x}^{0};\, \{\rho_{\tx}\})$
and call $\{\rho_{\tx}\}$ the \emph{sparsity pattern} of the 
source. 
For future reference, we define the following  nomenclature.

\begin{defn}
	When the system size grows without bound, namely,
	$\M,\N\to\infty$ with fixed
	and finite \emph{compression rate} $\alpha = \M/\N$ and 
	sparsity levels $\{\rho_{t}\}$, 
	we say the CS setup approaches the 
	large system limit (LSL).
\end{defn}

\begin{defn} 
	\label{defn:matrix_ensembles}
	Let $\vm{A}\in\R^{\M \times \N}$ be a sensing matrix
	with compression rate $\alpha = \M / \N \leq 1$.
	We say that the recovery problem 
	\eqref{eq:CS_Standard_Problem_LS} is:
	\begin{enumerate}
	\item $\T$-\emph{orthogonal setup},
	if $\N = \T \M$ and the sensing matrix is constructed as
	\begin{equation}
		\label{eq:Tortho_sensing}
		\vm{A} = 
		\begin{bmatrix} \vm{O}_{1} & \cdots & \vm{O}_{\T} \end{bmatrix},
	\end{equation}
	where $\{\vm{O}_{t}\}$ are independent and distributed 
	uniformly on the group of orthogonal $\M\times\M$ matrices according 
	to the Haar measure%
	\footnote{In the following, a matrix $\vm{O}$ that has 
	this distribution is said to be simply a \emph{Haar matrix}.};
	\item \emph{Standard Gaussian setup}, if the elements of $\vm{A}$ are
	IID drawn according to $\Gpdf{a}(0;\,1/\M)$;
	\item \emph{Row-orthogonal setup}, if $\vm{O}$
	is an $\N\times\N$ Haar matrix and the sensing matrix is constructed as
	$\vm{A} = \alpha^{-1/2} \Dproj \vm{O}$, where $\Dproj =
	[\vm{I}_{M} \; \vm{0}_{M \times (N-M)}]$
	picks the first $M$ rows of $\vm{O}$.  Clearly 
	$\vm{A} \vm{A}^{\trans} = \alpha^{-1}\I_{\M}$ and $\vm{A}$ has
	orthogonal rows.
	\item 
	\emph{Geometric setup}, if 
	$\vm{A} = \vm{U}\vm{\Sigma}\vm{V}^{\trans}$ where $\vm{U},\vm{V}$ are independent 
	Haar matrices and $\vm{\Sigma}\in\R^{\M\times\N}$ is a diagonal matrix whose 
	$(m,m)th$ entry is given by $\sigma_{m} \propto \tau^{m-1}$ for $m = 1,\ldots,M$.  
	The parameter $\tau\in(0,1]$ is chosen so that given value of peak-to-average
	eigenvalue ratio 
	\begin{equation}
		\label{eq:peaktoaverage}
		\kappa = \frac{\sigma^2_{1}}{\frac{1}{M}\sum_{m=1}^{M}\sigma^2_{m}}
	\end{equation}
	is met and the singular values are scaled to satisfy the power constraint
	$N^{-1}\sum_{m=1}^{M}\sigma^2_{m} = 1$.  For 
	details, see Appendix~\ref{sec:geometric_ensemble}.
	\item 
	\emph{General rotationally invariant setup}, if the decomposition
	$\vm{R} = \vm{A}^{\trans}\vm{A} = \vm{O}^{\trans} \vm{D} \vm{O}$ 
	exists, so that $\vm{O}$ is an
	$\N\times\N$ Haar matrix
	and $\vm{D}$ is a diagonal matrix containing the eigenvalues of $\vm{R}$.
	We also assume that the empirical distribution of the eigenvalues
	\begin{equation}
		\label{eq:eed}
		F^{\M}_{\vm{R}}(x) 
		= \frac{1}{\M} \sum_{\Midx=1}^{\M} 1(\lambda_{\Midx}(\vm{R}) \leq x),
	\end{equation}
	where 
	$1( \cdot )$ is the indicator function and
	$\lambda_{\Midx}(\vm{R})$ denotes the $\Midx$th 
	eigenvalue of $\vm{R}$, 
	converges to some non-random limit
	in the LSL and satisfies
	$\E\tr(\vm{A} \vm{A}^{\trans})/N = \E \tr(\vm{D})/N  = 1$.
	The setups 2)~--~4) are all special cases of this ensemble.
	\end{enumerate}
	To make comparison fair between different setups, 
	all cases above are defined so that 
	$\E\tr(\vm{A} \vm{A}^{\trans})/N = 1$.
	In addition, both of the orthogonal setups satisfy 
	the condition 
	$\alpha \vm{A} \vm{A}^{\trans} = \I_{\M}$.
\end{defn}

\begin{remark}
The $T$-orthogonal sensing matrix was considered in
\cite{Vehkapera_Kabashima_Chatterjee_et_all_2012_ITW,
Kabashima-Vehkapera-Chatterjee-2012} 
under the
assumption of noise-free measurements.  There it was shown
to improve the perfect recovery threshold when the source had
non-uniform sparsity.  On the other hand, the row-orthogonal
setup is the same matrix ensemble that was
studied in the context of CDMA in \cite{viswanath1999optimal,Kitagawa-Tanaka-2010}.
There it was called Welch bound equality (WBE) spreading sequence 
ensemble and shown to provide maximum spectral efficiency 
both for Gaussian \cite{viswanath1999optimal}
and non-Gaussian \cite{Kitagawa-Tanaka-2010} inputs
given optimal MMSE decoding.
The geometric setup is inspired by 
\cite{Rangan-Fletcher-Schniter-Kamilov-arxiv2015}, where similar 
sensing matrix was used to examine the robustness of AMP algorithm 
and its variants via Monte Carlo simulations.
It reduces to the row-orthogonal ensemble when $\kappa \to 1$.
\end{remark}

\subsection{Bayesian Framework}

To enable the use of statistical mechanics tools, we 
reformulate the original optimization problem 
\eqref{eq:CS_Standard_Problem_LS} 
in a probabilistic framework.  
For simplicity%
\footnote{This assumption is in fact not necessary for the replica analysis.  
However, if the source vector has independent elements and 
the regularization function decouples element-wise, the numerical 
evaluation of the saddle-point equations is a particularly simple task.}, 
we also make the additional restriction that 
the cost function separates as 
\begin{equation}
	\Freg(\vm{x}) = \sum_{\Nidx=1}^{\N} \Freg(x_{\Nidx}),
	\label{eq:Fdecoupling}
\end{equation}
where $\Freg$ is a function whose actual form depends 
on the type of the argument (scalar or vector).
Then, the postulated model prior (recall our notational convention 
from Section~\ref{sec:notation})
of the source is defined as
\begin{equation}
	q_{\beta} (\vm{x}) = \frac{1}{z_{\beta}}\e^{-\beta\Freg(\vm{x})},
	\label{eq:post_pdf_xi}
\end{equation}
where $z_{\beta} = \int \e^{-\beta\Freg(\vm{x})} \dx \vm{x} < \infty$ is a normalization constant.
Hence, $\Freg(\vm{x})$ needs to be such that the above integral is convergent for given
finite $\M,\N$ and $\beta>0$.
The purpose of the 
non-negative parameter $\beta$ (inverse temperature)
is to enable MAP detection as will become clear later.  
Note that \eqref{eq:post_pdf_xi} 
encodes no information about the sparsity pattern of the source 
and is mismatched from the true prior 
$\PM(\vm{x}^{0};\, \{\rho_{\tx}\})$.  From an algorithmic point-of-view, 
this means that the system operator has no specific knowledge about the 
underlying sparsity structure or does not want to utilize it due 
to increased computational complexity.
We also define a postulated PDF for the measurement process
\begin{equation}
	q_{\beta} (\vm{y} \mid \vm{A}, \vm{x})
	= \Gpdf{\vm{y}}\bigg(\vm{A} \vm{x};\, 
	\frac{\lambda}{\beta}\I_{\M}\bigg),
	\label{eq:post_conditional_measurement_pdf}	
\end{equation}
so that unless $\lambda/\beta = \sigma^{2}$, the observations
are generated according to a different model
than what the reconstruction algorithm assumes.
Note that $\lambda$ is the same parameter as in the original 
problem \eqref{eq:CS_Standard_Problem_LS}.

Due to Bayes' theorem, the (mismatched) posterior density of $\vm{x}$ 
based on the postulated distributions reads
\begin{IEEEeqnarray}{l}
	\label{eq:mismatched_posterior_of_x}
	q_{\beta} (\vm{x} \mid \vm{y}, \vm{A}) \IEEEnonumber\\
	\quad =
	\frac{1}{Z_{\beta}(\vm{y},\vm{A})}
	\exp \bigg[ -\beta \bigg(\frac{1}{2\lambda} 
	\| \vm{y} - \vm{A} \vm{x}\|^{2} 
	+ \Freg(\vm{x}) \bigg)\bigg],
	\IEEEeqnarraynumspace
\end{IEEEeqnarray}
where $Z_{\beta}(\vm{y},\vm{A})$ is the normalization factor
or \emph{partition function} of the above PDF.  We could
now estimate $\vm{x}$ based on 
\eqref{eq:mismatched_posterior_of_x},
for example by computing the posterior mean 
$\langle \vm{x} \rangle_{\beta}$, where we
used the notation
\begin{equation}
	\langle h (\vm{x}) \rangle_{\beta} = 
	\int h (\vm{x}) q_{\beta} (\vm{x} \mid \vm{y}, \vm{A}) \dx \vm{x}
	\label{eq:mmse_estimate}
\end{equation}
for some given $\beta>0$ and trial function $h$ of $\vm{x}$.
The specific case that maximizes the a posteriori probability
for given $\lambda$ (and $\sigma^{2})$ is 
the \emph{zero-temperature configuration}, obtained by 
letting $\beta\to\infty$. 
In this limit \eqref{eq:mismatched_posterior_of_x}
reduces to a uniform distribution over $\vm{x}$ that 
provides the global minimum
of $\| \vm{y} - \vm{A} \vm{x}\|^{2}/ (2\lambda)
+ \Freg(\vm{x})$.  If the problem has a unique solution, 
we have $\langle \vm{x} \rangle_{\beta\to\infty} = \xvechat$,
where $\xvechat$ is the solution of \eqref{eq:CS_Standard_Problem_LS}.
Thus, the behavior of regularized LS reconstruction
can be obtained by studying the 
density \eqref{eq:mismatched_posterior_of_x}. 
This is a standard problem in statistical mechanics if we 
interpret $q_{\beta} (\vm{x} \mid \vm{y}, \vm{A})$ as the Boltzmann 
distribution of a spin glass, as described next.

\subsection{Free Energy, The Replica Trick and Mean Square Error}
\label{sec:freeE_replicas}

The key for finding the statistical properties 
of the reconstruction \eqref{eq:CS_Standard_Problem_LS}
is the normalization factor or partition function
$Z_{\beta}(\vm{y},\vm{A})$.
Based on the statistical mechanics approach, our goal is to 
assess the (normalized) \emph{free energy}
\begin{equation}
	f_{\beta}(\vm{y},\vm{A}) 
	= -\frac{1}{\beta \N} \ln Z_{\beta}(\vm{y},\vm{A})
	\label{eq:freeE_fixed}	
\end{equation}
in the LSL where $\M,\N\to\infty$ with $\alpha = \M / \N$ fixed,
and obtain the desired statistical properties from it.
However, the formulation above is problematic since 
$f_{\beta}$ depends on the observations
$\vm{y}$ and the measurement process $\vm{A}$.  
One way to circumvent this difficulty is to notice that
the law of large numbers guarantees that for 
$\forall \epsilon > 0$, the probability that 
$|f_{\beta}(\vm{y},\vm{A}) - 
\E \{f_{\beta}(\vm{y},\vm{A})\}| > \epsilon$ 
tends to vanish in the LSL for any 
finite and positive $\lambda,\sigma^{2}$.  
This leads to computation
of the average free energy
$f_{\beta} = \E \{ f_{\beta}(\vm{y},\vm{A})\}$
instead of \eqref{eq:freeE_fixed} and is called 
\emph{self-averaging} in statistical mechanics.

Concentrating on the average free energy $f_{\beta}$ avoids 
the explicit dependence on $\{\vm{y}, \vm{A} \}$.
Unfortunately, assessing the necessary expectations 
is still difficult and we need some further manipulations 
to turn the problem into a tractable one.  The first step is 
to rewrite the average free energy in the zero-temperature limit
as
\begin{equation}
	f =  -\lim_{\beta,\N\to\infty} \frac{1}{\beta \N}
	\lim_{\NR\to 0^{+}} \frac{\partial}{\partial \NR}\ln \E \{[Z_{\beta}(\vm{y},\vm{A})]^{\NR}\}.
	\label{eq:freeE_replica_real}	
\end{equation}
So-far the development has been rigorous
if $\NR\in\R$ and the limits are unique and exist%
\footnote{In principle, the existence of a unique
thermodynamic limit can be checked using the techniques 
introduced in \cite{Guerra-Toninelli-2002-2}. 
However, since the replica method itself is 
already non-rigorous we have opted to verify the results 
in the end using numerical simulations.}.  
The next step is to employ the \emph{replica trick} 
to overcome the apparent road block of evaluating
the necessary expectations as a function of real-valued
parameter $\NR$.

\begin{replicatrick}
	Consider the free energy in
	\eqref{eq:freeE_replica_real} and 
	let $\vm{y} = \vm{A} \vm{x}^{0} + \vm{w}$ be a fixed
	observation vector.
	 Assume that the
	limits commute, which in conjunction with the expression
	\begin{IEEEeqnarray}{l}
		[Z_{\beta}(\vm{y},\vm{A};  \lambda)]^{\NR} 
		\IEEEnonumber\\
		= \int \exp 
		\bigg(-\frac{\beta}{2 \lambda}\sum_{a=1}^{\NR} \| \vm{y} - \vm{A} \vm{x}^{a}\|^{2}
		- \beta \Freg(\vm{x}^{a})\bigg)\prod_{a=1}^{\NR}\mathrm{d} \vm{x}^{a}
		\IEEEeqnarraynumspace
		\label{eq:Z_replicated_1}
	\end{IEEEeqnarray}
	for $\NR =1,2,\ldots$ allows the evaluation of the expectation in
	\eqref{eq:freeE_replica_real}
	as a function of $\NR \in \R$. The obtained 
	functional expression is then utilized in 
	taking the limit of $\NR \to 0^{+}$. 
\end{replicatrick}

It is important to note that as written above, the validity of
the analytical continuation remains an open question and the replica 
trick is for this part still lacking mathematical validation.
However, as remarked in Introduction, the most serious problem in 
practice seems to arise from the simplifying 
assumptions one makes about how the correlations between the 
variables $\{\vm{x}^{a}\}$ behave in the LSL.  
The simplest case is the \emph{RS ansatz}, that is described 
by the \emph{overlap matrix} $\vm{Q}\in \R^{(\NR+1)\times(\NR+1)}$
of the form
\begin{equation}
	\vm{Q} = [Q^{[a,b]}]_{a,b=0}^{\NR} = \begin{bmatrix}
	r & m & \cdots& m\\
	m & Q & q \\
	\vdots & q & \ddots\\
	m & & & Q
	\end{bmatrix}
	\label{eq:Q_extended}
\end{equation}
with slightly non-standard indexing that is
common in literature related to replica analysis.
The elements of $\vm{Q}$ are defined as overlaps, or 
empirical correlations,
$Q^{[a,b]} = \N^{-1} \vm{x}^{a} \cdot \vm{x}^{b}$.
The implication of RS ansatz is that the
replica indexes $a = 1,2,\ldots,\NR$ can be arbitrarily permuted without changing the 
end result when $\M = \alpha \N \to \infty$.  
This seems a priori 
reasonable since the replicas introduced in \eqref{eq:Z_replicated_1}
were identical and in no specific order.  But as also
mentioned in Introduction, the RS assumption is not always correct.  
Sometimes the discrepancy 
is easy to fix, for example as in the case of random energy model 
\cite{Mezard-Montanari-2009}, while much more intricate methods
like Parisi's full RSB solution are 
needed for other cases \cite{mezard1987spin}.
For the purposes of the present paper, we restrict ourselves to the 
RS case and check accuracy of the end result w.r.t.\ simulations.
Although this might seem mathematically somewhat unsatisfying approach, we believe 
that the RS results are useful for practical purposes due to their simple form 
and can server as a stepping stone for possible extensions to the RSB cases.

Finally, let us consider the problem of finding the MSE of reconstruction
\eqref{eq:CS_Standard_Problem_LS}.   Using the notation introduced 
earlier, we may write
\begin{IEEEeqnarray}{rCl}
	\mse &=&
	\N^{-1} \E\big\{\big\langle \|\vm{x}^{0} - \vm{x} \|^{2}
	\big \rangle_{\beta\to\infty}\big\} \nonumber\\
	&=& \rho - 2 
	\E\big\{ \N^{-1} 
	\langle \vm{x} \rangle^{\trans}_{\beta\to\infty}\vm{x}^{0}\big\}
	+ \E\big\{\N^{-1}\big\langle\|  \vm{x} \|^{2}\big\rangle_{\beta\to\infty}\big\} 
	\quad \IEEEnonumber\\
	 &=& \rho - 2 m + Q,
	 \label{eq:mse_rho_m_Q}
\end{IEEEeqnarray}
where $\langle \,\cdots \rangle_{\beta}$ was defined in \eqref{eq:mmse_estimate} and
$\E \{\, \cdots\,\}$ denotes the expectation w.r.t.\ 
variables in \eqref{eq:Sparse_Representation_with_Noise}.
Thus, if we can compute $m$ and $Q$ using the replica method, 
the MSE of reconstruction follows immediately from 
\eqref{eq:mse_rho_m_Q}.  As shown above,
this amounts to computing the overlap matrix \eqref{eq:Q_extended}.


\section{Results for LASSO Reconstruction}
\label{sec:lasso_results}

In this section we provide the results of the replica
analysis for LASSO reconstruction \eqref{eq:CS_Standard_Problem_LASSO}
for the ensembles introduced in
Definition~\ref{defn:matrix_ensembles}.  For simplicity, 
we let the non-zero elements of the source be standard Gaussian, that is,
$\pi(x) = \Gpdf{x}(0;\,1)$ in \eqref{eq:true_sourcesym_pdf_k}.
Recall also that LASSO is the 
special case of regularization $\Freg(\vm{x}) = \|\vm{x}\|_{1}$ 
in the general problem \eqref{eq:CS_Standard_Problem_LS}.
Replica symmetric ansatz is assumed in the derivations 
given in Appendices~\ref{sec:Tortho_replica}~and~\ref{sec:eigen_replica}.
Casual reader finds a sketch of replica analysis 
along with some generalizations for different choices 
of the cost function $\Freg(\vm{x})$ in Section~\ref{sec:extensions_and_sketch}.  
Further interpretation of the result and connections to 
the earlier work in \cite{Rangan-Fletcher-Goyal-2012} are also 
discussed there.
After the analytical results, we provide some 
numerical examples in the following subsection.

\subsection{Analytical Results}

The first result shows that when the measurement matrix is of the 
$T$-\emph{orthogonal form}, the MSE over the whole vector 
may depend, not just on the 
average sparsity $\rho$ but also on the 
block-wise sparsities $\{\rho_{t}\}$.

\begin{prop}
	\label{prop:Tortho_general}
	Consider the $T$-orthogonal setup described in 
	Definition~\ref{defn:matrix_ensembles} and 
	let $\mse_{t}$ denote the MSE of the 
	LASSO reconstruction in block 
	$t=1,\ldots,T$.   The average MSE over the entire vector 
	of interest reads 
	\begin{equation}
		\label{eq:Tortho_total_mse}
		\mse = \frac{1}{T}\sum_{t=1}^{T}\mse_{t}
		=\frac{1}{T}\sum_{t=1}^{T} (\rho_{t}- 2 m_{t} + Q_{t}).
	\end{equation}
	Then, under RS ansatz,
	\begin{IEEEeqnarray}{rCl}
		\label{eq:m_Torthogonal}
		m_{t}  &=&  
		2 \rho_{t}\Qfunc\bigg(\frac{1}{\sqrt{\chihat_{t}+\mhat^{2}_{t}}}\bigg),
		\IEEEeqnarraynumspace \\
		\label{eq:Q_Torthogonal}
		Q_{t}  &=&    
		-\frac{2(1-\rho_{t})}{\mhat_{t}^{2}} r(\chihat_{t}) - \frac{2\rho_{t}}{\mhat_{t}^{2}}
		r(\chihat_{t}+\mhat^{2}_{t}),
		\IEEEeqnarraynumspace
	\end{IEEEeqnarray}
	where $\Qfunc(x) = \int_{x}^{\infty} \dx z \:\! \e^{-z^{2}/2}/\sqrt{2 \pi}$ is the
	standard $Q$-function and we denoted
	\begin{IEEEeqnarray}{rCl}
		r(h) &\triangleq& \sqrt{\frac{h}{2\pi}} \e^{-\frac{1}{2 h}}
		-(1+h)\Qfunc\bigg(\frac{1}{\sqrt{h}}\bigg),
		\label{eq:rfunc_prop} 
		\IEEEeqnarraynumspace\\
		\mhat_{t} &\triangleq& \frac{1}{\lambda + \sum_{k\neq t}\Lambda_{k}^{-1}},
	\end{IEEEeqnarray}
	for notational convenience.
	The parameters $\{\Lambda_{t}\}$ and $\{\chihat_{t}\}$ are the 
	solutions to the set of coupled equations
	\begin{IEEEeqnarray}{rCl}
		\Lambda_{t} &=& \bigg(\frac{1}{R_{t}} - 1\bigg)\mhat_{t}  \\
		\chihat_{t} &=& 
		\frac{(\rho_{t} - 2 m_{t} + Q_{t}) \Lambda_{t}^{2}}{(1-R_{t})^{2}} \IEEEnonumber\\
		&& + \sum_{s=1}^{T} \Delta_{s,t} (\rho_{s} - 2 m_{s} + Q_{s} - \sigma^{2} R^{2}_{s}),
		\IEEEeqnarraynumspace
	\end{IEEEeqnarray}
	where we also used the auxiliary variables $R_{t} = \chi_{t}\mhat_{t}$ with
	\begin{IEEEeqnarray}{rCl}
		\label{eq:Rt_prop}
		\chi_{t} 
		&\triangleq&
		\frac{2(1-\rho_{t})}{\mhat_{t}}
		\Qfunc\bigg(\frac{1}{\sqrt{\chihat_{t}}}\bigg) + \frac{2\rho_{t}}{\mhat_{t}}
		\Qfunc\bigg(\frac{1}{\sqrt{\chihat_{t}+\mhat^{2}_{t}}}\bigg), \IEEEeqnarraynumspace\\
		\Delta_{s,t}
		&\triangleq& 
		\frac{R_{s}R_{t}\Lambda_{s}\Lambda_{t}}{(1-2 R_{s})(1-2 R_{t})}
		\bigg(1+\sum_{k=1}^{T}\frac{R_{k}^{2}}{1-2 R_{k}}\bigg)^{-1}
		\IEEEnonumber\\ 
		&& -\frac{\Lambda_{t}^{2}}{1-2 R_{t}} \delta_{st}, \IEEEeqnarraynumspace
	\end{IEEEeqnarray}
	and Kronecker delta symbol $\delta_{ij}$ to simplify the notation.
\end{prop}

\begin{proof}
See Appendix~\ref{sec:Tortho_replica}.
\end{proof}

The connection of the MSE provided in
Proposition~\ref{prop:Tortho_general} and the formulation given in
\eqref{eq:mse_rho_m_Q} is as follows.  Here $m = T^{-1}\sum_{t} m_{t}$
and $Q = T^{-1}\sum_{t} Q_{t}$ due to the assumption of block-wise sparsity,
as described in Section~\ref{sec:formulation}.  
The parameters $\{\Lambda_{t}, \chihat_{t} \}$ have to be solved for all 
$t = 1,\ldots,T$, i.e., we have a set of $2 T$ non-linear 
equations of $2 T$ variables.  Note that for the purpose of solving 
these equations, the parameters
$\{\mhat, R_{t},  \Delta_{s,t}, m_{t}, Q_{t}\}$ are just 
notation and do not act as additional variables in the problem.
Except for $m_{t}$ and $Q_{t}$, the rest of the 
variables can in fact be considered to be ``auxiliary''.  They 
arise in the replica analysis when we assess the expectations 
w.r.t.\ randomness of the measurement process, additive noise and
vector of interest.  Hence, one may think that the 
replica trick transformed the task of computing difficult expectations
to a problem of finding solutions to a set of coupled fixed point 
equations defined by some auxiliary variables.  In terms of 
computational complexity, this is a very fair trade indeed.

The implication of Proposition~\ref{prop:Tortho_general}
is that the performance of the $T$-orthogonal ensemble is in general 
dependent on the details of the sparsity pattern $\{\rho_{t}\}$
--- in a rather complicated way.
Similar result was
reported for the noise-free case in 
\cite{Kabashima-Vehkapera-Chatterjee-2012}, where the 
$T$-orthogonal ensemble was shown 
to provide superior reconstruction
threshold compared to rotationally invariant cases when the 
source vector had non-uniform sparsity.

For future convenience, we next present the special case of uniform sparsity
as an example.
Since now $\rho = \rho_{t}$ for all $t = 1,\ldots,T$, we have only two 
coupled fixed point equations to solve.  The  auxiliary parameters 
also simplify significantly and we have the following result.

\begin{example}
\label{example:Tortho_unif}
	Consider the $T$-orthogonal setup and 
	assume uniform sparsity $\rho_{t} = \rho$ for all $t=1,\ldots,T$.  
	The per-component MSE of reconstruction 
	can be obtained by solving the set of equations
	\begin{IEEEeqnarray}{rCl}
		\Lambda &=& \frac{1}{\lambda}
		\bigg[2(1-\rho)\Qfunc\bigg(\frac{1}{\sqrt{\chihat}}\bigg)
		+2\rho \Qfunc\bigg(\frac{1}{\chihat+\mhat^{2}}\bigg)\bigg]^{-1} 
		- \frac{T}{\lambda},
		\IEEEeqnarraynumspace \\
		\chihat &=& 
		\Lambda^{2}\bigg(\frac{ \rho - 2 m + Q}{(1-R)^{2}} 
		- \frac{ \rho - 2 m + Q - \sigma^{2} R^{2}}{1-2 R + T R^{2}}\bigg),
		\IEEEeqnarraynumspace
	\end{IEEEeqnarray}
	where we introduced the definitions
	\begin{IEEEeqnarray}{rCl}
		R &\triangleq& \frac{1}{T + \lambda \Lambda}, \IEEEeqnarraynumspace\\
		\mhat &\triangleq& \frac{\Lambda}{T+\lambda\Lambda -1}, \IEEEeqnarraynumspace
	\end{IEEEeqnarray}
	for notational simplicity.
\end{example}

For completeness, we provide next a result similar to 
Proposition~\ref{prop:Tortho_general} 
for the rotationally invariant case.  It shows that 
the performance of this matrix ensemble 	
does not depend on the specific values of $\{\rho_{t}\}$ but
only on the expected sparsity level $\rho = \Tx^{-1}\sum_{\tx} \rho_{\tx}$
of the vector of interest.
The MSE is given as a function of the Stieltjes
transform and its first order derivative of the asymptotic eigenvalue 
distribution $F_{\vm{A} \vm{A}^{\trans}}$ of the measurement matrix.
As shown in Section~\ref{sec:equivalence}, the end result is
essentially the same as the HCIZ integral formula  based
approach in \cite{Tulino-etal-cs2013}, but the derivation 
and form of the proposition are chosen here
to match the previous analysis.  It should be 
remarked, however, that Proposition~\ref{prop:Tortho_general} 
cannot be obtained from \cite{Tulino-etal-cs2013} 
and the $T$-orthogonal ensemble requires a special treatment.

\begin{prop}
	\label{prop:eigensetup}
	Recall the rotationally invariant setup given in 
	Definition~\ref{defn:matrix_ensembles}. The average MSE 
	of reconstruction for LASSO in this case is given by
	$\mse = \rho - 2 m + Q$, where the parameters $m$ and $Q$ 
	are as in \eqref{eq:m_Torthogonal}~and~\eqref{eq:Q_Torthogonal}
	with the block index $t$ omitted.
	To obtain the MSE, the following set of equations 
	\begin{IEEEeqnarray}{rCl}
		\label{eq:chi_rotational}
		\chi  &=&
		\frac{2(1-\rho)}{\mhat}\Qfunc\bigg(\frac{1}{\sqrt{\chihat}}\bigg)
		+\frac{2\rho}{\mhat}\Qfunc\bigg(\frac{1}{\sqrt{\chihat+\mhat^{2}}}\bigg),\\
		\label{eq:hatchi_rotational}
		\chihat 
		&=& (\rho - 2 m + Q)
		\bigg(\frac{1}{\chi^{2}} + \Lambda' \bigg) \IEEEnonumber\\ 
		&&  - \alpha \sigma^{2}
		\big[G_{\vm{A} \vm{A}^{\trans}}(-\lambda \Lambda)  
		- (\lambda \Lambda) \cdot  G'_{\vm{A} \vm{A}^{\trans}}(-\lambda \Lambda)\big] \Lambda',
		\IEEEeqnarraynumspace\\
		\label{eq:Lambda_rotational}
		\Lambda
		&=&\frac{1}{\chi}\bigg(1  - \alpha\big[1- (\lambda\Lambda)
		\cdot G_{\vm{A} \vm{A}^{\trans}}(-\Lambda\lambda) \big]\bigg),
	\end{IEEEeqnarray}
	need to be solved given the definitions
	\begin{IEEEeqnarray}{rCl}
		 \label{eq:LambdaP_rotational}
		\Lambda' &\triangleq& \frac{\partial \Lambda}{\partial \chi} 
		= -\bigg[\frac{1-\alpha}{\Lambda^{2}} + \alpha\lambda^{2} 
		G'_{\vm{A} \vm{A}^{\trans}}(-\Lambda\lambda)\bigg]^{-1},
		\IEEEeqnarraynumspace\\
		\label{eq:hatm_rotational}
		\mhat  &\triangleq& \frac{1}{\chi} - \Lambda
		= \frac{\alpha}{\chi}
		\big[1- (\lambda\Lambda) \cdot G_{\vm{A} \vm{A}^{\trans}}(-\Lambda\lambda) \big].
	\end{IEEEeqnarray}
	The function $G_{\vm{A} \vm{A}^{\trans}}(s)$ is the Stieltjes transform 
	of the asymptotic distribution $F_{\vm{A} \vm{A}^{\trans}}(x)$  and 
	$G_{\vm{A} \vm{A}^{\trans}}'(s)$ is the derivative of 
	$G_{\vm{A} \vm{A}^{\trans}}(s)$ w.r.t.\ the argument.
\end{prop}

\begin{proof}
See Appendix~\ref{sec:eigen_replica}.  
\end{proof}

\begin{remark}
	\label{remark:same_variables_ensembles}
	Due to \eqref{eq:Rt_prop}~and~\eqref{eq:chi_rotational}, 
	the rotationally invariant and $T$-orthogonal setups have exactly 
	the same form for variables $\{m,Q,\chi\}$.  Hence, the choice of 
	random matrix ensemble does not affect these variables, except for 
	adding the indexes $t=1,\ldots,T$ to them in the case of $T$-orthogonal setup.
\end{remark}

Notice that in contrast to Definition~\ref{defn:matrix_ensembles},
the above proposition uses the eigenvalue distribution of 
$\vm{A} \vm{A}^{\trans}$ instead of $\vm{A}^{\trans}\vm{A}$.  This is more
convenient in the present setup since $\M \leq \N$, so 
we do not have to deal with the zero eigenvalues.
Compared to the $T$-orthogonal case, here we have a fixed set of 
three equations and unknowns to solve, regardless of $\{\rho_{t}\}$ and 
the partition of the source vector.  
Thus, if one knows the Stieltjes transform $G_{\vm{A} \vm{A}^{\trans}}(s)$
and it is (once) differentiable with respect to the argument, 
the required parameters can be solved numerically from 
the coupled equations given in Proposition~\ref{prop:eigensetup}.  
For some $G_{\vm{A} \vm{A}^{\trans}}(s)$, however, the 
equations can be reduced analytically to simpler forms that allow
for efficient numerical evaluation of the MSE, as seen 
in the following two examples.

\begin{example}
	\label{example:row-orthogonal}
	Recall the \emph{row-orthogonal setup}.  For this case
	$\vm{D} = \alpha^{-1}\I_{\M}$ and, thus,  the Stieltjes 
	transform of $F_{\vm{A}\vm{A}^{\trans}}$ reads
	\begin{IEEEeqnarray}{rCl}
		\label{eq:GAA_rowortho}
		G_{\vm{A}\vm{A}^{\trans}}(s) &=& \frac{1}{\alpha^{-1}-s}, \IEEEeqnarraynumspace
	\end{IEEEeqnarray}
	so that $G'_{\vm{A}\vm{A}^{\trans}}(s) = (\alpha^{-1} - s)^{-2}$.
	Plugging the above in Proposition~\ref{prop:eigensetup} provides
	\begin{IEEEeqnarray}{rCl}
		\Lambda &=& \frac{\lambda -\alpha^{-1}\chi
			+\sqrt{-4 \lambda \chi +(\lambda  + \alpha^{-1}\chi)^{2}}}{2 \lambda \chi},
		\IEEEeqnarraynumspace \\
		\mhat &=& \frac{\lambda + \alpha^{-1} \chi - \sqrt{-4 \lambda \chi + 
				(\lambda + \alpha^{-1} \chi)^2}}{2 \lambda  \chi},
		\IEEEeqnarraynumspace\\
		\Lambda' &=& -\bigg[\frac{1-\alpha}{\Lambda^{2}} +  
			\frac{\alpha\lambda^{2}}{(\alpha^{-1}+\lambda \Lambda)^2}\bigg]^{-1}, \\
		\chihat 
		&=& (\rho - 2 m + Q) \bigg(\frac{1}{\chi^{2}} + \Lambda' \bigg)
		- \frac{\sigma^{2}}{(\alpha^{-1}+\lambda \Lambda)^2} \Lambda', 
		\IEEEeqnarraynumspace
	\end{IEEEeqnarray}
	so that the MSE can be obtained by solving the set $\{\chi, \chihat\}$
	of equations given by \eqref{eq:chi_rotational}~and~\eqref{eq:hatchi_rotational}.
	As expected, for the special case of uniform sparsity 
	$\rho_{t} = \rho$ and 
	$\alpha = 1/T$, the row-orthogonal and $T$-orthogonal setups give
	always the same MSE (see
	Remark~\ref{rem:Tortho_vs_rowortho} in Appendix~\ref{sec:eigen_replica}).
\end{example}

In the above example, we first used \eqref{eq:GAA_rowortho} in 
\eqref{eq:Lambda_rotational}~and~\eqref{eq:hatm_rotational} 
to solve $\Lambda$ and $\mhat$ as functions of $\chi$.  
Plugging then $G'$ into \eqref{eq:LambdaP_rotational} provides immediately 
$\Lambda'$, and $\chihat$ follows similarly from \eqref{eq:hatchi_rotational}.
However, if the form of $G$ is more cumbersome, it may be 
more convenient to compute the parameters in slightly 
different order as demonstrated in the next example 
that considers a generalized version of the standard Gaussian 
CS setup.

\begin{example}
	\label{example:Gauss}
	Consider the rotationally invariant setup where the
	elements of $\vm{A}$ are zero-mean IID with 
	variance $1/\M$.
	The eigenvalue spectrum of 
	$\vm{A}\vm{A}^{\trans}$ is then given by 
	the Mar{\v c}enko-Pastur law 
	\begin{equation}
		\label{eq:MP_G}
		G_{\vm{A}\vm{A}^{\trans}}(s) =
		\frac{-1 + \alpha^{-1} - s - \sqrt{-4 s + (-1 + \alpha^{-1} - s)^{2}}}{2 s}, 
	\end{equation}
	that provides together with Proposition~\ref{prop:eigensetup}
	\begin{IEEEeqnarray}{rCl}
		\Lambda &=& \frac{1}{\chi} - \frac{1}{\lambda + \alpha^{-1}\chi}, \\
		\mhat &=& \frac{1}{\lambda + \alpha^{-1}\chi}, \\
		\Lambda' &=& -\frac{1}{\chi^{2}} + \frac{\alpha^{-1}}{(\lambda + \alpha^{-1}\chi)^{2}}, \\
		G'_{\vm{A}\vm{A}^{\trans}}(-\lambda \Lambda) 
		&=& -\frac{1}{\alpha \lambda^{2}} \bigg(\frac{1-\alpha}{\Lambda^{2}}
		+ \frac{1}{\Lambda'} \bigg).
		\label{eq:Stieltjes_diff_MP}
	\end{IEEEeqnarray}
	Plugging the above to Proposition~\ref{prop:eigensetup} and solving 
	$\{\chi,\chihat\}$ yields the MSE of reconstruction.  
\end{example}

In this case, $G'_{\vm{A}\vm{A}^{\trans}}(s)$ is of more complex
form than in Example~\ref{example:row-orthogonal}, and the approach used 
there does not provide as simple solution as before.  However, since now
$\Lambda$ has a particularly convenient form, 
we may use the definition
$\Lambda' = \frac{\partial \Lambda}{\partial \chi}$  
to write $\Lambda'$ as a function of $\chi$. 
This can be then used in \eqref{eq:LambdaP_rotational}
to obtain $G'$ indirectly.

\begin{remark}
	As shown in Section~\ref{sec:equivalence},
	Examples~\ref{example:row-orthogonal}~and~\ref{example:Gauss} 
	provide the same average reconstruction error as 
	reported in \cite{Tulino-etal-cs2013}, which also implies 
	that the IID case matches \cite{Rangan-Fletcher-Goyal-2012}.
	 The benefit of Example~\ref{example:Gauss}
	compared to \cite{Rangan-Fletcher-Goyal-2012} is that there 
	are no integrals or expectations left to solve.
	Example~\ref{example:row-orthogonal}, on the other hand,
	proved directly that for the special case of uniform sparsity 
	$\rho_{t} = \rho$ and $\alpha = 1/T$, the row-orthogonal and 
	$T$-orthogonal setups give always the same MSE.
\end{remark}

The final example demonstrates the capabilities of the analytical
framework for a ``non-standard'' ensemble that has singular values
defined by geometric progression.

\begin{example}
	\label{example:geometric}
	Recall the geometric setup where the singular values $\sigma_{m} \propto \tau^{m-1}, 
	m = 1,\ldots,M$ satisfy peak-to-average condition \eqref{eq:peaktoaverage}.
	Equivalent limiting eigenvalue spectrum of $\vm{A}\vm{A}^{\trans}$ 
	in the large system limit is described by the Stieltjes transform
	\begin{IEEEeqnarray}{rCl}
		G_{\vm{A}\vm{A}^{T}} (s) &=&
		\frac{1}{s\gamma}\ln\bigg( \frac{A - s}{A e^{-\gamma}- s} \bigg) - \frac{1}{s}, \\
		G_{\vm{A}\vm{A}^{T}}'(s) &=& - \frac{1}{s}G_{\vm{A}\vm{A}^{T}} (s)
		- \frac{1}{s}\frac{A( e^{-\gamma} - 1)}{\gamma(A e^{-\gamma}-s)(A-s)},
		\IEEEeqnarraynumspace 
	\end{IEEEeqnarray}
	where $A = \kappa / \alpha$ and $\gamma$ satisfies 
	\begin{equation}
		\kappa = \frac{\gamma}{1-\e^{\gamma}},
	\end{equation}
	for some given $\kappa$.
	For details on how to generate the geometric ensemble for simulations and 
	how the Stieltjes transform arises for this setup, see 
	Appendix~\ref{sec:geometric_ensemble}.
\end{example}

\subsection{Numerical Examples}

Having obtained the theoretical performance of various matrix ensembles,
we now examine the behavior of the MSE in some chosen setups numerically.
First we consider the case of uniform density and the MSE of reconstruction
as a function of the tuning parameter $\lambda$, as shown 
in Fig.~\ref{fig:example1}.  The solid lines depict the performances
of  row-orthogonal and standard Gaussian setups, as given by 
Examples~\ref{example:row-orthogonal}~and~\ref{example:Gauss}.  The markers,
on the other hand, correspond to the result obtained by 
Rangan~et~al.~\cite[Section~V-B]{Rangan-Fletcher-Goyal-2012} and 
 Example~\ref{example:Tortho_unif} given in this paper.  As expected, 
the solid lines and markers match perfectly although 
the analytical results are represented in a completely different form.
It is important to notice, however, that we plot here the MSE 
(in decibels) while normalized mean square error (in decibels) 
$\mse/\rho$ is used in \cite{Rangan-Fletcher-Goyal-2012}.  
Also, the definition of signal-to-noise ratio $\mathsf{SNR}_{0}$ 
there would correspond to value $\rho /\sigma^{2}$ in this paper.
Comparing the two curves in Fig.~\ref{fig:example1} makes it 
clear that the orthogonal 
constructions provide superior MSE performance compared to the 
standard Gaussian setup.   It is also worth pointing out that 
the optimal value of $\lambda$ depends on the choice of 
the matrix ensemble.
\begin{figure}[tb]
	\centering
	\includegraphics[width=0.49\textwidth]{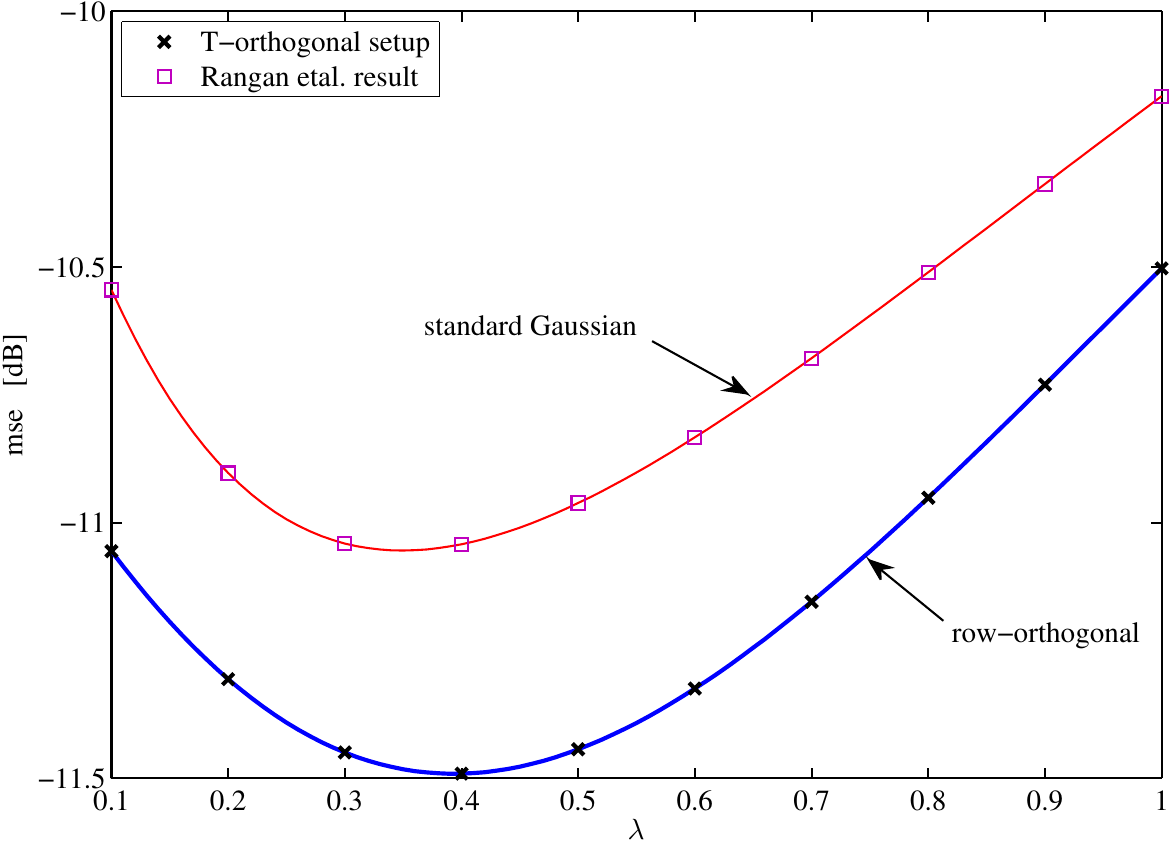}
	\caption{Average MSE in decibels $\mse \to 10 \log_{10} (\mse)$~dB 
	as a function of the tuning parameter $\lambda$.  
	Solid lines given by 
	Examples~\ref{example:row-orthogonal}~and~\ref{example:Gauss}.
	Markers for the standard Gaussian setup are obtained from
	the analytical results provided 
	in \cite[Section~V]{Rangan-Fletcher-Goyal-2012}  and
	by Example~\ref{example:Tortho_unif} for the $T$-orthogonal case.  All 
	results should thus be considered to be in the LSL.
	As predicted by the analysis, the lines and markers match perfectly.
	Parameter values: $\alpha = 1/3, \rho = 0.15, \sigma^{2} = 0.1$.}
	\label{fig:example1}
\end{figure}
\begin{figure*}[tb]
	\centering
	\includegraphics[width=0.6\textwidth]{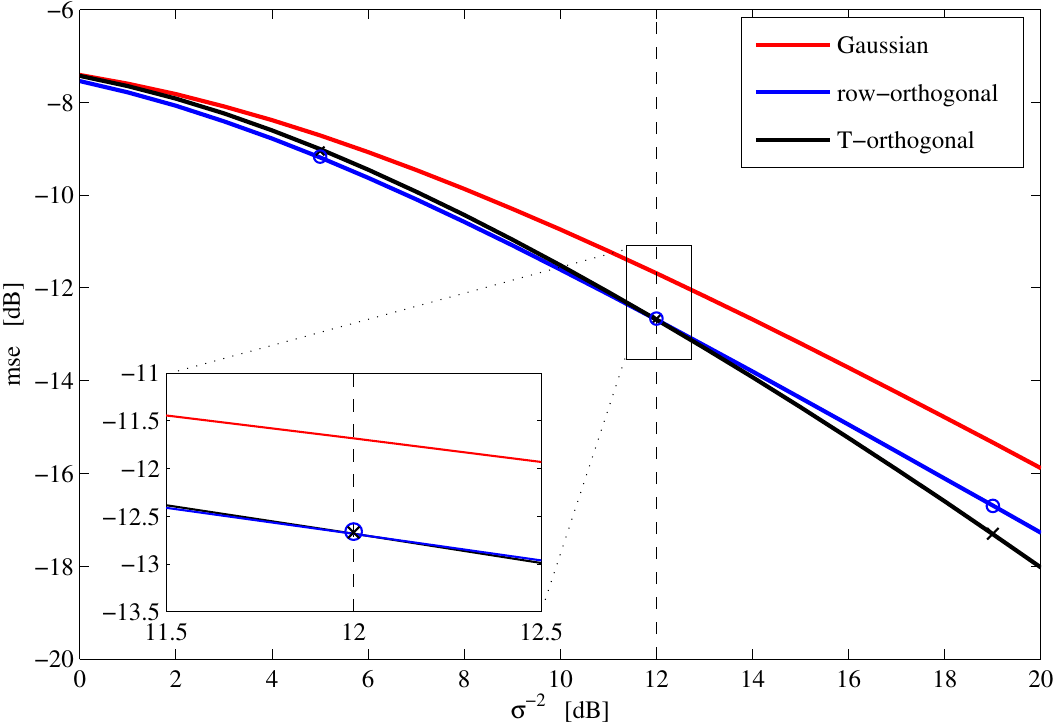}
	\caption{Average MSE vs.\ inverse noise variance $1/\sigma^{2}$, where both 
	quantities are presented in decibels, that is, $x \to 10 \log_{10} (x)$~dB.  
	For each point, the value of $\lambda$ that minimizes the MSE is chosen 
	numerically. The parameter values are $\alpha = 1/2, \rho = 0.2$ and 
	localized sparsity is considered, that is, either 
	$\rho_{1} = 0.4$ and $\rho_{2} = 0$ or vice versa.  The dashed line at 
	$\sigma^{-2} = 12$~dB represents the point where the 
	MSE performance of $T$-orthogonal and row-orthogonal ensembles approximately 
	cross each other.
	Markers at $\sigma^{-2} = 5,12,19$~dB are obtained by using \texttt{cvx}
	for reconstruction and averaging over $100\, 000$ realizations of the problem.}
	\label{fig:example2}
\end{figure*}
\begin{figure}[tb]
	\centering
	\includegraphics[width=0.49\textwidth]{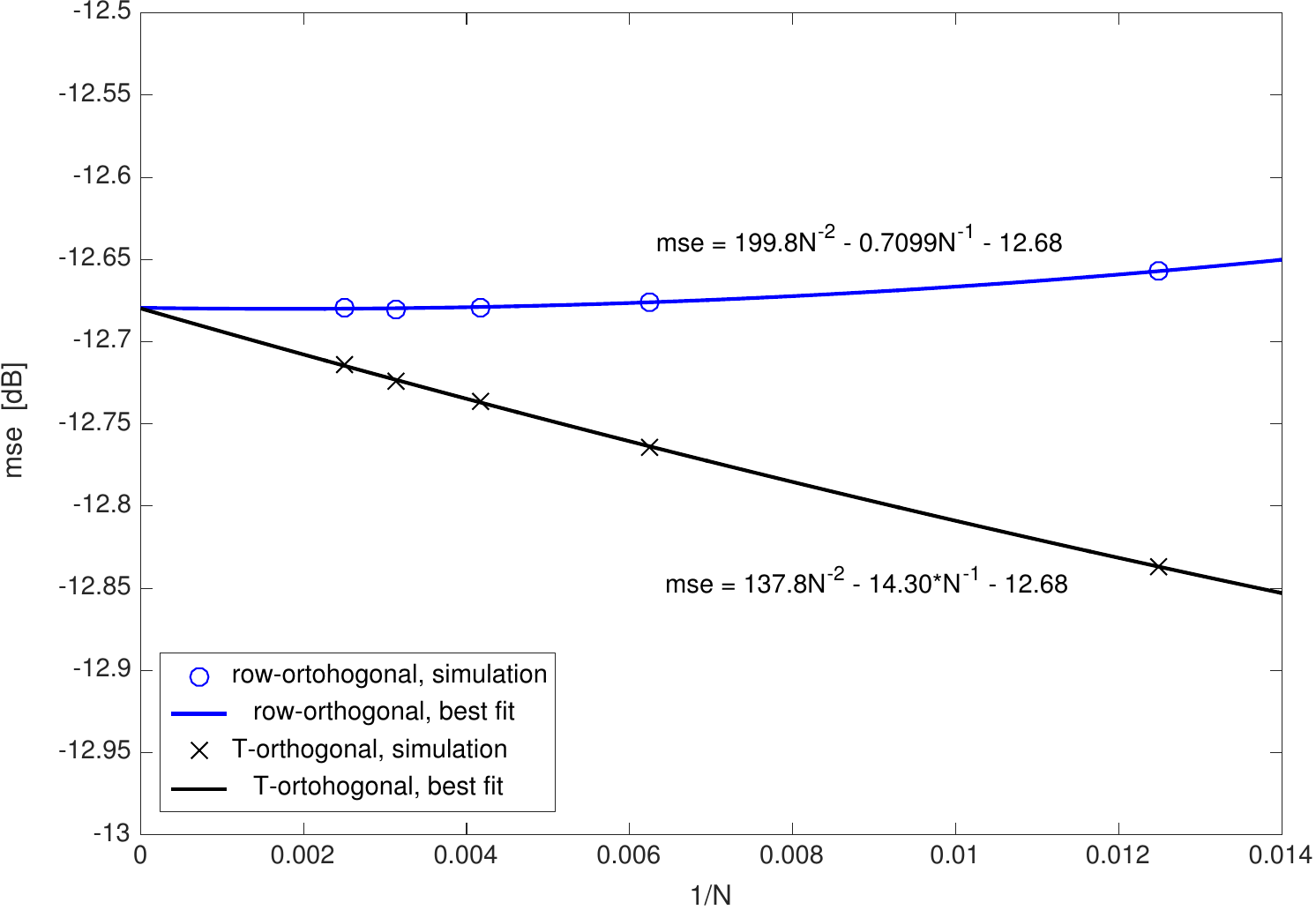}
	\caption{Experimental assessment of the MSE for the point $\sigma^{-2} = 12$~dB in
	Fig.~\ref{fig:example2}.  Markers correspond to simulated values using \texttt{SpaSM} 
	with path-following optimization over $\lambda$.
	Problem sizes $\N = 80, 160, \ldots, 400$, each averaged over $10^6$ realizations.  The
	experimental data were fitted with a quadratic function
	of $1/N$ and plotted with solid lines. Extrapolation for $\M =  \alpha \N \to \infty$ 
	provides the estimates for the asymptotic MSE that agrees 
	well with the replica prediction for the RS ansatz $\mse = -12.685$~dB.}
	\label{fig:example3}
\end{figure}

In the next experiment we consider the case of ``localized sparsity''
where all non-zero elements are concentrated in one subvector $\vm{x}_{t}$,
namely, $\rho_{t} = \rho T$ for some $t \in \{1,\ldots,T\}$ and 
$\rho_{s} = 0 \; \forall s \neq t$.  For simplicity, we take the simplest 
case of $T=2$ and choose the overall sparsity to be $\rho = 0.2$.
The average mean square error vs.\ inverse noise variance $1/\sigma^{2}$ of this 
case is depicted in Fig.~\ref{fig:example2}. For clarity of presentation,
the variables related to both axes are given in a decibel form, 
that is, $x \to 10 \log_{10} (x)$~dB, where $x \in \{\mse, \sigma^{-2}\}$.
The tuning parameter $\lambda$ is 
chosen for each point so that the lowest possible MSE is obtained for all 
considered methods.  Due to the simple form of the analytical results, this 
is easy to do numerically.  
Examining Fig.~\ref{fig:example2} reveals a surprising 
phenomenon, namely, for small noise variance the $T$-orthogonal setup
gives the lowest average MSE while for more noisy setups it is 
the row-orthogonal ensemble that achieves the best reconstruction performance.
The universality of this behavior for other parameter values
is left as an open question for future research.
The point where these two ensembles give approximately 
the same MSE for the given setup is located at 
$\sigma^{-2} = -12$~dB.  Hence, for optimal performance, if one is given 
the choice of these three matrix ensembles, 
the row-orthogonal setup should be chosen when on the left hand side of the dashed line 
and $T$-orthogonal setup otherwise.  At any point in the figure, however, 
the standard Gaussian setup gives the worst reconstruction in the MSE 
sense and should never been chosen if such freedom is presented.

To illustrate how the experimental points in Fig.~\ref{fig:example2} were
obtained, we have plotted in Fig.~\ref{fig:example3} the average MSE 
of the point $\sigma^{-2} = 12$~dB, for which the asymptotic prediction 
given by the replica method is $\mse = -12.685$~dB. 
To obtain more accurate results, the experimental data is averaged now 
over $10^6$ realizations and estimates are obtained by using \texttt{SpaSM} \cite{spasm}
that provides an efficient MATLAB implementation of the
least angle regression (LARS) algorithm \cite{Efron-Hastie-Johnstone-Tibshirani-2004}.
The simulated data is fitted with a quadratic function
of $N^{-1}$ and the estimates for the asymptotic MSEs are obtained by
extrapolating $\M =  \alpha \N \to \infty$.  The end result for 
for both the row-orthogonal and $T$-orthogonal ensembles is $\mse = -12.68$~dB,
showing that the replica analysis provides a good approximation of the MSE 
for large systems.  One can also observe that the simulations approach the 
asymptotic result relatively fast and for realistic system sizes the match is 
already very good.  Albeit the convergence behavior depends somewhat 
on the noise variance $\sigma^{2}$, the present figure is a typical 
example of what was observed in our simulations.

Finally, we use the geometric setup to examine the robustness of LASSO
against varying peak-to-average ratio $\kappa$ introduced in \eqref{eq:peaktoaverage}.
Substituting the formulas given in Example~\ref{example:geometric} to 
Proposition~\ref{prop:eigensetup} provides the MSE of reconstruction for LASSO 
as given in Fig.~\ref{fig:example4}.  The system parameters are $\alpha = 1/2, \rho = 0.2$
and $\sigma^2$ is chosen to match the markers in Fig.~\ref{fig:example2}.
As mentioned earlier, the geometric setup reduces to the row-orthogonal one when 
$\kappa\to 1$, which can verified by comparing the $\kappa=1$ values to the 
markers of the row-orthogonal curve in Fig.~\ref{fig:example2}.  We have also 
included additional simulation points at $\kappa = 5$, obtained by using 
\texttt{SpaSM} as in Fig.~\ref{fig:example3}.  One can observe that the performance
degradation with increasing $\kappa$ for the LASSO problem is relatively graceful 
compared to the abrupt transition of GAMP observed in 
\cite{Rangan-Fletcher-Schniter-Kamilov-arxiv2015}.  Hence, the algorithmic considerations
are indeed very important, as studied therein.  We also observe that as expected, the MSE is a 
monotonic increasing function of $\kappa$ for all $\sigma^{2}$, highlighting the 
fact that sensing matrices with ``flat'' eigenvalue distributions are in general 
good for reconstruction of noisy sparse signals.
\begin{figure}[tb]
	\centering
	\includegraphics[width=0.49\textwidth]{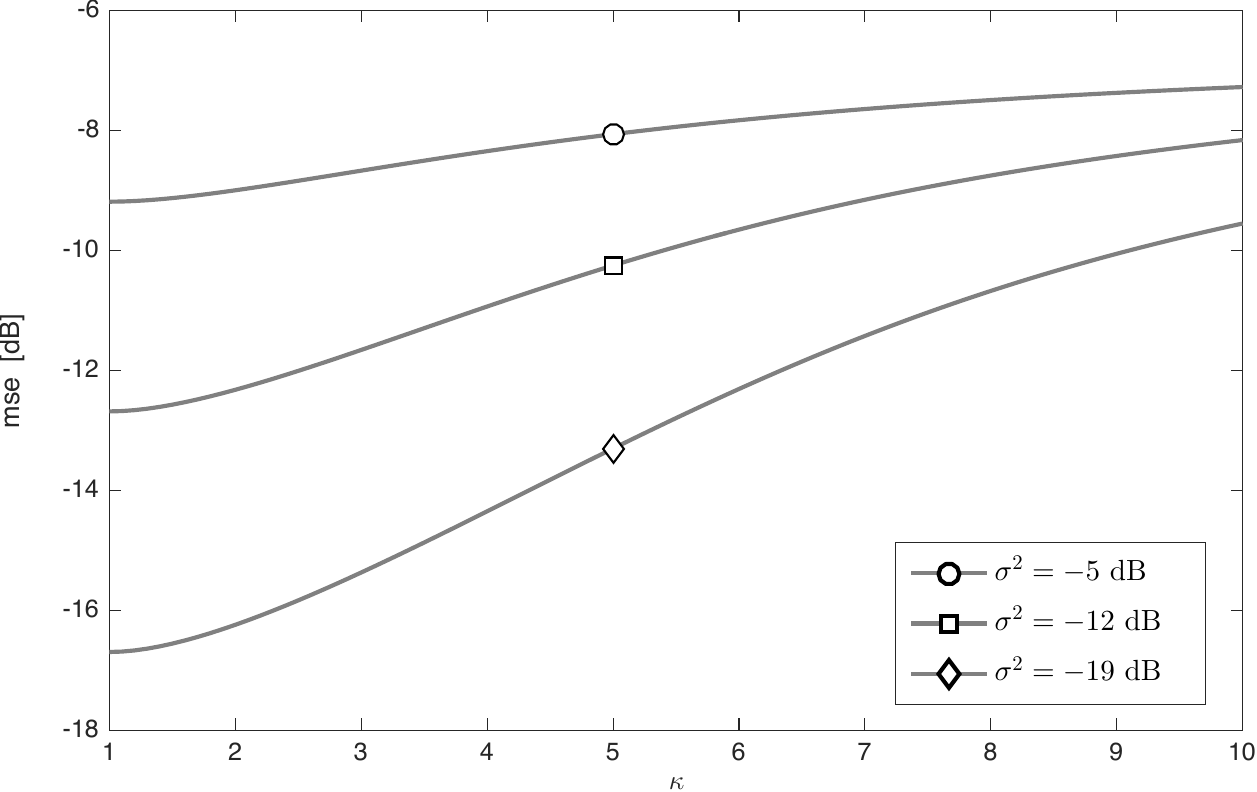}
	\caption{Average MSE vs.\ the peak-to-average ratio $\kappa$ for the geometric setup.
	Markers are obtained by extrapolating simulated values as in Fig.~\ref{fig:example3}
	and optimal $\lambda$ is used for each point.  Parameter values are $\alpha = 1/2, \rho = 0.2$, 
	and the points at $\kappa = 1$ match the markers of
	the row-orthogonal setup in Fig.~\ref{fig:example2}.}
	\label{fig:example4}
\end{figure}


\section{Extensions and a Sketch of a Derivation with the Replica Method}
\label{sec:extensions_and_sketch}

Although the LASSO reconstruction examined in the previous section 
is one of the most important special cases of regularized LS problem
\eqref{eq:CS_Standard_Problem_LS},
one might be interested in expanding the scope of analysis to 
some other cases.  In fact, up to a certain point, the evaluation of
free energy \eqref{eq:freeE_replica_real} is independent 
of the choice of regularization $\Freg(\vm{x})$ as well as the 
marginal distribution of the 
non-zero components of the source vector.
For this end, let us assume in the following that the source vector has 
elements drawn independently according to 
\eqref{eq:true_sourcesym_pdf_k} 
where $\pi(x)$ is a suitable PDF with zero-mean, unit variance and finite moments.
Note, however, that in order to obtain the final saddle-point conditions, 
one needs to make a choice about both $\Freg$ and $\pi$ in the end.

\subsection{Sketch of a Replica Analysis}

To provide a brief sketch how the results in previous section
were obtained and elucidate where the choice of regularization \eqref{eq:Fdecoupling}
affects the analysis, let us consider for simplicity
the rotationally invariant 
ensemble with source vector that has uniform sparsity, i.e.,
we set $T=1$ and $\rho = \rho_{1}$.
By Appendices~\ref{sec:Tortho_replica}~and~\ref{sec:eigen_replica} we know that 
\eqref{eq:freeE_replica_real} can be written under RS ansatz as
\begin{equation}
	\label{eq:freeE_replica_extended}
	f = -\lim_{\NR\to 0^{+}}
	\frac{\partial}{\partial \NR}
	\lim_{\beta,\N\to\infty} \frac{1}{\beta \N}
	\ln \Xi_{\beta,\N}(\NR),
\end{equation}
where $\NR$ is the number of replicas in the system.
To characterize $\Xi_{\beta,\N}(\NR)$, we 
consider the simplest case of RS overlap matrix \eqref{eq:Q_extended}
and remind the reader that albeit this choice may seem intuitively 
reasonable, it is known to be incorrect in some cases
(see Introduction for further discussion).

For given replica symmetric $\vm{Q}$, we may compute its probability 
weight%
\footnote{We have omitted a term vanishing multiplicative terms 
and $-\NR^{2}\beta\chihat q/2$ in the 
exponent since it does not affect the free energy, see 
\eqref{eq:Tortho-p_bM} in Appendix~\ref{sec:Tortho_replica}.}
\begin{equation}
	p_{\beta,\N}(\vm{Q};\, \NR) \propto
	\int [\mathcal{V}_{\beta}(\Qmathat;\, \NR)]^{\N}
	\e^{\N \beta\NR (\Qhat Q - \chihat \chi - 2 \mhat m ) / 2} \dx \Qmathat,
\end{equation}
where $\chi = \beta (Q-q)$.  With some abuse of notation,
we used above $\vm{Q}$ as a shorthand for 
the set $\{\chi,Q,m\}$ and similarly
$\Qmathat$ for the auxiliary parameters $\{\chihat,\Qhat,\mhat\}$.
Given that the prior of $\vm{x}^{0}$ factorizes
and the regularization that separates as \eqref{eq:Fdecoupling}
the auxiliary term $\mathcal{V}_{\beta}(\Qmathat;\, \NR)$ is given in
\begin{figure*}[tb]
	\normalsize  
	\begin{IEEEeqnarray}{rCl}
		\mathcal{V}_{\beta}(\Qmathat;\, \NR) &=&
		\int\! p(x^{0})\bigg(\prod_{a=0}^{\NR}  \dx x^{a}\bigg) 
		\exp\bigg(\!\!-\beta \sum_{a=1}^{\NR} \Freg(x^{a})\bigg)
		\exp\bigg[\!-\frac{\beta\Qhat}{2} \sum_{a=1}^{\NR} (x^{a})^{2}
		+ \beta \mhat x^{0} \sum_{a=1}^{\NR} x^{a} 
		+ \frac{1}{2}\bigg(\beta\sqrt{\chihat}\sum_{a=1}^{\NR}x^{a} \bigg)^{2}\bigg]  
		\IEEEnonumber\\
		&=&
		\iint p(x^{0})\Bigg\{ \int \exp \Bigg(\beta  \bigg[-\frac{\Qhat}{2} x^{2}
		+\big(z\sqrt{\chihat} + \mhat x^{0} \big)x 
		- \Freg(x) \bigg] \Bigg)  \dx x \Bigg\}^{\NR} \dx x^{0} \Dx z,
		\IEEEeqnarraynumspace
		\label{eq:V_extended}
	\end{IEEEeqnarray}
	\hrulefill 
\end{figure*}
\eqref{eq:V_extended} at the top of the next page
where $\Dx z = \dx z\, \e^{z^2/2}/\sqrt{2\pi}$
and $p(x^{0})$ is assumed to be of the form 
\eqref{eq:true_sourcesym_pdf_k} with $\rho = \rho_{1}$ and $T=1$.
It is important to realize that the replica analysis could, in principle,
be done with general forms of $\Freg(\vm{x})$ and $p(\vm{x})$ but then 
\eqref{eq:V_extended} would have vectors in place of scalars.
This creates computational problems as explained in short.
Note also that we have taken here a slightly different route (on purpose)
compared to the analysis carried out in the Appendices.  The approach
there is more straightforward mathematically but the methods presented
here give some additional insight to the solution and provide
connection to the results given in \cite{Rangan-Fletcher-Goyal-2012}
and \cite{Tulino-etal-cs2013}.

Next, define a function
\begin{IEEEeqnarray}{l}
	H_{\beta,\lambda}(\sigma^{2},\nu)  
	= - \frac{1 + \ln (\beta \nu)- \alpha\ln \lambda}{2} \IEEEnonumber\\
	+\frac{1}{2} \extr_{\Lambda}\{\Lambda (\beta\nu) - (1-\alpha) \ln \Lambda 
	- \alpha \E \ln (\Lambda\beta\sigma^{2} +\Lambda\lambda+ x) \}, \IEEEnonumber\\
	\label{eq:hatH_extended}
\end{IEEEeqnarray}
where $\alpha = \M / \N$, $\nu > 0$ and the expectation of the $\ln$-term is w.r.t.\ 
the asymptotic eigenvalue distribution $F_{\vm{A}\vm{A}^{\trans}}(x)$.  Then
\begin{IEEEeqnarray}{l}
	\Xi_{\beta,\N}(\NR) = \int \dx \vm{Q} \,p_{\beta,\N}(\vm{Q};\, \NR)
	\IEEEnonumber\\ 
	\qquad \times \e^{\N H_{\beta,\lambda} (\NR\sigma^{2},\NR (r - 2 m + q) + Q - q) 
	+ \N (\NR-1) H_{\beta,\lambda} (0,Q - q)}, \IEEEeqnarraynumspace
	\label{eq:I_Q_sigma2_extended}
\end{IEEEeqnarray}
where the exponential term involving function $H$ arises from the expectation
of \eqref{eq:Z_replicated_1} w.r.t.\ noise $\vm{w}$
and sensing matrix $\vm{A}$ given $\{\vm{x}^{a}\}$ and $\vm{Q}$.  The 
relevant matrix integral is proved in Appendix~\ref{sec:Ffunc_rotational} 
and more details can be found in the derivations following 
Lemma~\ref{lemma:Ffunc_Tortho} in Appendix~\ref{sec:Tortho_replica}.
In the limit of vanishing temperature and increasing system size
$\beta,\N\to\infty$, saddle point method%
\footnote{For more information, see for example \cite[Ch.~6]{Orszag-Bender-1978}
and \cite[Ch.~12.7]{Arfken-Weber-Harris-2013},
or for a gentle introduction \cite{Goutis-Casella-1999}. 
Note that in our case when $\beta,N\to\infty$, the 
correction terms in front of the exponentials after saddle point
approximation vanish due to the logarithm and division by $\beta N$ 
in \eqref{eq:freeE_replica_extended}.  With some abuse of 
notation, we have simply omitted them in the paper 
to avoid unnecessary distractions.}
may be employed to assess the integrals 
over $\vm{Q}$ and $\Qmathat$.  Taking also the partial 
derivative w.r.t.\ $\NR$ and then letting $\NR\to 0$ provides
\begin{IEEEeqnarray}{rCl}
	f &=& \extr_{\chi,Q,m,\chihat,\Qhat,\mhat} 
	\bigg\{\mhat m  -  \frac{ \Qhat Q}{2} +\frac{\chihat \chi}{2}
	\IEEEnonumber\\
	&& + \iint p(x^{0}) \phi\big( z\sqrt{\chihat} + \mhat x^{0} ;\, \Qhat\big)  \dx x^{0}\Dx z
	\label{eq:freeE_3_extended} \\
	&& - \lim_{\beta\to \infty} \frac{1}{\beta} \lim_{\NR\to 0} \frac{\partial}{\partial \NR}
	H_{\beta,\lambda}(\NR\sigma^{2}, \NR (\rho - 2 m + q) + Q - q) \bigg\}, \IEEEnonumber
\end{IEEEeqnarray}
where we defined a scalar function
\begin{equation}
	\phi(y;\, \Qhat) = \min_{x} \bigg\{\frac{\Qhat}{2} x^{2} - y x + \Freg (x) \bigg\}.
	\label{eq:func_phi_extended}
\end{equation}
Note that due to the limit $\NR \to 0$ in \eqref{eq:freeE_3_extended}, the 
function $H$ has the arguments $\sigma^{2}=0$ and
$\nu =  Q - q$ when it comes to solving the extremization over $\Lambda$ 
in \eqref{eq:hatH_extended}.  The value of $\Lambda$ 
which provides a solution for this is denoted $\Lambda^{*}$ in the following.  
We also remark that \eqref{eq:func_phi_extended} is the only part of free energy
that directly depends on the choice of cost 
function $\Freg$.  If we would have chosen $p(\vm{x})$ that is not a 
product distribution and $\Freg$ that did not separate as in 
\eqref{eq:Fdecoupling}, we would need to consider here a multivariate 
optimization over $\vm{x}$.  In fact, the dimensions should 
grow without bound by the assumptions of the analysis and, hence, the
problem would be infeasible in general.  In practice one could consider
some large but finite setting and use it to approximate the infinite 
dimensional limit, or concentrate on forms of $\Freg(\vm{x})$ and $p(\vm{x})$ that
have only a few local dependencies.  For simplicity we have chosen to restrict ourselves
to the case where the problem fully decouples into a single dimensional one.

We can get an interpretation of the remaining parameters as follows.  First, let 
\begin{IEEEeqnarray}{rCl}
	\langle h(x); \, y\rangle_{\beta}  
	&=& \frac{1}{Z_{\beta} (y)} \int  h(x) \e^{-\beta [\Qhat x^{2} / 2 - y x + \Freg(x)]} \dx x 
	\IEEEeqnarraynumspace
\end{IEEEeqnarray}
be a posterior mean of some function $h(x)$.  Note that the structure 
of this scalar estimator is essentially the same as the vector valued 
counterpart given in \eqref{eq:mismatched_posterior_of_x}~and~\eqref{eq:mmse_estimate}.
If we denote the $x$ that minimizes \eqref{eq:func_phi_extended} by 
$\xhat(y;\,\Qhat)$ then clearly 
\begin{equation}
	\langle h(x); \, y\rangle_{\beta\to\infty} = h \big(\xhat(y;\,\Qhat) \big),
\end{equation}
and also
\begin{equation}
	\label{eq:xhat_asderivative}
	\xhat(y;\,\Qhat) =  - \frac{\partial}{\partial y} \phi(y;\, \Qhat).
\end{equation}
We thus obtain from \eqref{eq:V_extended} with a little bit calculus that
as $\N,\beta\to \infty$ and $\NR\to 0$, the terms that yield the 
mean square error $\mse = \rho - 2 m + Q$ are given by
\begin{IEEEeqnarray}{rCl}
	\label{eq:m_interp}
	m &=&  \iint x^{0} \xhat\big(z\sqrt{\chihat} + \mhat x^{0};\,\Qhat\big) p(x^{0}) 
	\dx x^{0} \Dx z, \IEEEeqnarraynumspace\\
	\label{eq:Q_interp}
	Q &=&  \iint \big[\xhat\big(z\sqrt{\chihat} + \mhat x^{0};\,\Qhat\big)\big]^{2}
	p(x^{0}) \dx x^{0} \Dx z, \IEEEeqnarraynumspace
\end{IEEEeqnarray}
where we now have $y = z\sqrt{\chihat} + \mhat x^{0}$ and $z$ is a standard Gaussian RV. 
Similarly, the extremum condition for the variable $\chi$ reads
\begin{IEEEeqnarray}{rCl}
	\chi &=& \frac{1}{\sqrt{\chihat}}\int z \int \xhat\big(z\sqrt{\chihat} + \mhat x^{0};\,\Qhat\big)
	p(x^{0}) \dx x^{0} \Dx z \IEEEeqnarraynumspace \IEEEnonumber\\
	&=&  \iint \frac{\partial}{\partial(z\sqrt{\chihat})}
	\xhat\big(z\sqrt{\chihat} + \mhat x^{0};\,\Qhat\big) p(x^{0}) \dx x^{0} \Dx z, 
	\IEEEeqnarraynumspace  
	\label{eq:chi_interp}
\end{IEEEeqnarray}
where the latter equality is obtained using integration by parts formula.
For many cases of interest, these equations can be evaluated 
analytically, or at least numerically, and they provide the single 
body representation of the variables in \eqref{eq:mse_rho_m_Q}.
If one substitutes \eqref{eq:true_sourcesym_pdf_k} with
$\pi(x) = \Gpdf{x}(0;\,1)$ in the above formulas
along with $\Freg(x) = |x|$, some (tedious) calculus shows that the results of 
the previous section are recovered.  An alternative way of obtaining 
the parameters is provided in 
Appendices~\ref{sec:Tortho_replica}~and~\ref{sec:eigen_replica}.

Given $\{\mhat, \Qhat, \chihat\}$, one may now obtain the MSE of the original 
reconstruction described in Section~\ref{sec:formulation} 
by considering an equivalent scalar problem, namely, 
\begin{IEEEeqnarray}{rCl}
	\label{eq:mse_decoupled}
	\mse &=& \rho - 2 m + Q  \IEEEnonumber\\
	 &=& \iint \big|x^{0} 
	 - \xhat\big(z\sqrt{\chihat} + \mhat x^{0};\,\Qhat\big)\big|^{2} 
	 p(x^{0}) \dx x^{0} \Dx z.
	 \IEEEeqnarraynumspace
\end{IEEEeqnarray}
We may thus conclude that the minimizing $x$ in \eqref{eq:func_phi_extended}
is the $\xhat$ above, which can be interpreted as the output of a regularized 
LS estimator that postulates $y = \Qhat x + \Qhat^{1/2}z$, or equivalently,
\begin{equation}
	\label{eq:scalar_post_ymod}
	y = x + \Qhat^{-1/2} z,
\end{equation}
while the true model is $y = \mhat x^{0} + z\sqrt{\chihat}$, or equivalently,
\begin{equation}
	\label{eq:scalar_true_ymod}
	y = x^{0} + z \frac{\sqrt{\chihat}}{\mhat}.
\end{equation}
In the notation of \cite{Tulino-etal-cs2013} (resp. 
\cite{Rangan-Fletcher-Goyal-2012}), 
we thus have the relations $\xi \leftrightarrow \Qhat \; (\leftrightarrow \lambda_{p})$ 
and $\eta \leftrightarrow \mhat^{2}/\chihat \; (\leftrightarrow \mu^{-1})$
between the parameters.  The above also 
implies $\mhat = \Qhat$ as is indeed verified later in
\eqref{eq:hatQ-hatm_extended}.  We shall expand 
on the connection between this paper and  
\cite{Rangan-Fletcher-Goyal-2012,Tulino-etal-cs2013} 
in Section~\ref{sec:equivalence}.

Let us denote $\Lambda^{*}$
for the solution of extremization in 
\eqref{eq:hatH_extended} under condition $\sigma^{2}=0$, namely,
\begin{IEEEeqnarray}{l}
\Lambda^{*}-\frac{1}{\chi} = - \frac{\alpha}{\chi} \big[1- (\lambda\Lambda^{*}) 
\cdot G_{\vm{A}\vm{A}^{\trans}}(-\lambda \Lambda^{*}) \big],
\IEEEeqnarraynumspace
\label{eq:Fextr_eigs2_extended}
\end{IEEEeqnarray}
which is the same condition as \eqref{eq:Lambda_rotational} in 
Proposition~\ref{prop:eigensetup}.
Then we can plug in \eqref{eq:freeE_3_extended} the identity
\begin{IEEEeqnarray}{l}
	\lim_{\beta\to \infty} \frac{1}{\beta} \lim_{\NR\to 0} 
	\frac{\partial}{\partial \NR} H_{\beta,\lambda}(\NR\sigma^{2},\nu(\NR)) \IEEEnonumber\\
	=  - \frac{\alpha\sigma^{2}\Lambda^{*}}{2} G_{\vm{A}\vm{A}^{\trans}}(-\lambda \Lambda^{*})
	+ \frac{\rho - 2 m + Q}{2} \bigg( \Lambda^{*} - \frac{1}{\chi} \bigg),
	\IEEEeqnarraynumspace
	\label{eq:chainrule2_extended}
\end{IEEEeqnarray}
where we used the fact that in the LSL, $r \to \rho$ in \eqref{eq:Q_extended} by 
the weak law of large numbers.  The RS free energy thus becomes 
\begin{IEEEeqnarray}{rCl}
	f &=& \extr_{\chi,Q,m,\chihat,\Qhat,\mhat} \bigg\{
	\mhat m  - \frac{ \Qhat Q}{2} + \frac{\chihat \chi}{2}
	\IEEEnonumber\\
	&& + \iint p(x^{0}) \phi\big( z\sqrt{\chihat} + \mhat x^{0} ;\, \Qhat\big)  \dx x^{0}\Dx z 
	\label{eq:freeE_4_extended}\\
	&& + \frac{\alpha\sigma^{2}\Lambda^{*}}{2} G_{\vm{A}\vm{A}^{\trans}}(-\lambda \Lambda^{*})
	- \frac{\rho - 2 m + Q}{2} \bigg( \Lambda^{*} - \frac{1}{\chi} \bigg)\bigg\}, \IEEEnonumber
\end{IEEEeqnarray}
where the parameters $\{m,Q,\chi\}$ satisfy \eqref{eq:m_interp}--\eqref{eq:chi_interp}
Note that so-far we have not made any assumptions
about the details of the function $\Freg(x)$, which means that 
\eqref{eq:freeE_4_extended} is valid for any 
type of regularization that separates as given in \eqref{eq:Fdecoupling}.  
In fact, we can go even further and solve the partial derivatives
w.r.t.\ variables $m$ and $Q$, which reveals that
\begin{equation}
	\label{eq:hatQ-hatm_extended}
	\Qhat = \mhat = \frac{1}{\chi} - \Lambda^{*}.
\end{equation}
With some additional effort, one also finds that
\begin{IEEEeqnarray}{l}
	\chihat = (\rho - 2 m + Q) \bigg(\frac{1}{\chi^{2}} 
	+ \frac{\partial\Lambda^{*}}{\partial \chi} \bigg)
	\IEEEnonumber\\
	\quad - \alpha \sigma^{2} 
	\big[G_{\vm{A}\vm{A}^{\trans}}(-\lambda \Lambda^{*})  - (\lambda \Lambda^{*}) 
	\cdot G'_{\vm{A}\vm{A}^{\trans}}(-\lambda \Lambda^{*}) \big] 
	\frac{\partial\Lambda^{*}}{\partial \chi} 
	\IEEEeqnarraynumspace
\end{IEEEeqnarray}
holds for the RS free energy with
\begin{equation}
	\label{eq:chi_derivative_extended}
	\frac{\partial\Lambda^{*}}{\partial \chi}
	= -\bigg[\frac{1-\alpha}{(\Lambda^{*})^{2}} + (\alpha\lambda^{2})
	\cdot G'_{\vm{A}\vm{A}^{\trans}}(-\lambda \Lambda^{*})\bigg]^{-1}.
\end{equation}

It is now easy to see that for the rotationally invariant 
setup, our initial assumption of uniform sparsity, i.e., 
$T=1$ and $\rho = \rho_{1}$ is not necessary and the same 
set of equations is obtained for arbitrary sparsity 
pattern $\{\rho_{t}\}$ that satisfies 
$\rho = \Tx^{-1}\sum_{\tx} \rho_{\tx}$.  Similarly, we may 
obtain the equivalent representation for the $T$-orthogonal
setup considered in Proposition~\ref{prop:Tortho_general}.

\begin{remark}
\label{remark:equal_parameters}
Comparing \eqref{eq:Fextr_eigs2_extended} together with
\eqref{eq:hatQ-hatm_extended}--\eqref{eq:chi_derivative_extended} 
to the saddle-point conditions 
\eqref{eq:hatchi_rotational}--\eqref{eq:hatm_rotational}
given in Proposition~\ref{prop:eigensetup} shows that the choice of
$\Freg$ or the marginal PDF of the non-zero elements in the source 
vector in \eqref{eq:true_sourcesym_pdf_k}  has no direct impact
on the expressions that provide the variables $\{\mhat,\Qhat,\chihat\}$.
Hence, the form of these conditions is the same for all setups where the 
sensing matrix is from the same ensemble.
The parameters $\{m,Q,\chi\}$ on the other hand are affected by the 
choice of regularization and source distribution.  As stated in 
Remark~\ref{remark:same_variables_ensembles}, the effect of 
sensing matrix ensemble is the reverse, namely, for fixed $\Freg$ 
and $p(\vm{x}^{0})$, the form of $\{m,Q,\chi\}$ is always the same
while $\{\mhat,\Qhat,\chihat\}$ can have different form depending 
on the choice of the measurement matrix.
\end{remark}

\subsection{Alternative Representation of Rotationally 
	Invariant Case and Comparison to Existing Results}
\label{sec:equivalence}

The saddle-point condition for the rotationally invariant case is described 
in Proposition~\ref{prop:eigensetup} in terms of the Stieltjes transform 
of $F_{\vm{A}\vm{A}^{\trans}}$.
In \cite{Tulino-etal-cs2013} the HCIZ-formula is used, which 
makes it natural to express the results in terms of the R-transform 
of $F_{\vm{A}^{\trans}\vm{A}}$.  Furthermore, different sets 
of auxiliary variables are used in these two papers.  
In this section we sketch an alternative representation 
of Proposition~\ref{prop:eigensetup} 
that is equivalent to \cite{Tulino-etal-cs2013}
up to some minor scaling factors.
This also implies that apart from minor differences 
in scalings, our results for the IID 
case are also equivalent to \cite{Rangan-Fletcher-Goyal-2012} 
as explained in \cite[Sec.~IV-C]{Tulino-etal-cs2013} and shown in 
Fig.~\ref{fig:example1}.

Let us first consider the conditions enforced by the matrix integration formula 
\eqref{eq:hatH_extended} through $\Lambda$.  By the remark following \eqref{eq:func_phi_extended},
we know that the relevant terms can also be obtained from 
\eqref{eq:H_func_appendix1} by setting $\sigma^2 = 0$ and $\chi = \beta(Q - q)$, namely
\begin{IEEEeqnarray}{l}
	H_{\beta,\lambda}(\sigma^{2}=0,\nu = Q - q) \IEEEnonumber\\
	\quad \simeq  \frac{1}{2} \extr_{\Lambda}\bigg\{\Lambda \chi - \int \ln  (x +\lambda \Lambda) 
	\dx F_{\vm{A}^{\trans}\vm{A}}(x) \bigg\}, \IEEEeqnarraynumspace
\end{IEEEeqnarray}
where we have omitted the terms that do not depend on $\Lambda$.
The solution to the extremization then provides the condition 
(we write $\Lambda = \Lambda^{*}$ here for simplicity)
\begin{IEEEeqnarray}{rCl}
	\chi &=& \lambda \int \frac{1}{x +\lambda \Lambda}
	\dx F_{\vm{A}^{\trans}\vm{A}}(x) 
	= \lambda G_{\vm{A}^{\trans}\vm{A}}(-\lambda \Lambda) 
	\IEEEeqnarraynumspace\\
	\iff \Lambda &=& -\frac{1}{\lambda} G_{\vm{A}^{\trans}\vm{A}}^{-1}\bigg(\frac{\chi}{\lambda}\bigg),
	\label{eq:altLambda1}
\end{IEEEeqnarray}
where $G_{\vm{A}^{\trans}\vm{A}}^{-1}\big( G_{\vm{A}^{\trans}\vm{A}} (s)\big) = s$ 
is the functional inverse of the Stieltjes transform.  Note that this also implies
\begin{equation}
	\label{eq:GLambda_simple}
	G_{\vm{A}^{\trans}\vm{A}}(-\lambda \Lambda) 
	= G_{\vm{A}^{\trans}\vm{A}}
	\bigg(G_{\vm{A}^{\trans}\vm{A}}^{-1}\bigg(\frac{\chi}{\lambda}\bigg)\bigg) 
	= \frac{\chi}{\lambda}.
\end{equation}
Using the definition $\mathsf{R}_{X}(z) = G_{X}^{-1}(-z) - z^{-1}$
of the R-transform in \eqref{eq:altLambda1} yields
\begin{equation}
	\label{eq:Lambda_R}
	\frac{1}{\lambda}\mathsf{R}_{\vm{A}^{\trans}\vm{A}}\bigg(\!\!\!-\frac{\chi}{\lambda}\bigg)
	=  \frac{1}{\chi} -\Lambda.
\end{equation}
On the other hand, we know from \eqref{eq:hatQ-hatm_extended}
that a solution to \eqref{eq:freeE_4_extended} satisfies
$\Qhat = \mhat = \chi^{-1} - \Lambda$, so that
\begin{equation}
	\label{eq:hatQ-hatm_tulino}
	\Qhat = \mhat = \frac{1}{\lambda}
	\mathsf{R}_{\vm{A}^{\trans}\vm{A}}\bigg(\!\!\!-\frac{\chi}{\lambda}\bigg)
\end{equation}
is the saddle-point condition for $\Qhat$ and $\mhat$ in terms of the R-transform.
Note that compared to the Stieltjes-transform  that is related to 
$\vm{A}\vm{A}^{\trans}$ at the saddle-point solution, the R-transform describes the 
eigenvalue spectrum of $\vm{A}^{\trans}\vm{A}$.
The condition \eqref{eq:hatQ-hatm_tulino} matches \cite[(131)]{Tulino-etal-cs2013}
apart from a slightly different placements of regularization parameters so that
$\mhat \leftrightarrow \xi$ as already remarked earlier.  The above also suggests 
that apart from scalings by the regularization parameter 
$\chi \leftrightarrow \E[\sigma^{2}(Y;\,\xi)]$, which can also be inferred from
\cite[Lemma~9]{Rangan-Fletcher-Goyal-2012}.

Finally, we know from the above developments and 
\cite[Appendix~B]{Tulino-etal-cs2013} that 
$\chihat \leftrightarrow f^{\star}$ should hold if the results 
are equal.  To this end, let us examine the last line of 
\eqref{eq:freeE_4_extended}, namely,
\begin{equation}
	\frac{\alpha\sigma^{2}\Lambda}{2}
	G_{\vm{A}\vm{A}^{\trans}}(-\lambda \Lambda) 
	- \frac{\rho - 2 m + Q}{2}\bigg( \Lambda - \frac{1}{\chi} \bigg),
\end{equation}
and substitute \eqref{eq:altLambda1}--\eqref{eq:Lambda_R} there.  Considering the end 
result as a function of $\chi$, we obtain
\begin{IEEEeqnarray}{rCl}
	\varphi(\chi) &=& \frac{\alpha\sigma^{2}}{2}
	\bigg[\frac{1}{\chi}-\frac{1}{\lambda}\mathsf{R}_{\vm{A}^{\trans}\vm{A}}
	\bigg(\!\!\!-\frac{\chi}{\lambda}\bigg)\bigg]\frac{\chi}{\lambda} \IEEEnonumber\\
	&& + \,\frac{\rho - 2 m + Q}{2 \lambda}\:\!
	\mathsf{R}_{\vm{A}^{\trans}\vm{A}}\bigg(\!\!\!-\frac{\chi}{\lambda}\bigg) \IEEEeqnarraynumspace \\
	&\simeq&
	\frac{1}{2}\bigg(\frac{\rho - 2 m + Q}{\lambda} 
	- \frac{\alpha\sigma^{2}\chi}{\lambda^{2}}\bigg)
	\mathsf{R}_{\vm{A}^{\trans}\vm{A}}\bigg(\!\!\!-\frac{\chi}{\lambda}\bigg), \IEEEeqnarraynumspace 	
	\label{eq:freeE_R1}
\end{IEEEeqnarray}
where \eqref{eq:freeE_R1} is obtained by omitting the terms that 
do not depend on $\chi$.  We are  interested in the point where the partial 
derivative in \eqref{eq:freeE_4_extended} w.r.t.\ $\chi$ vanishes, that is,
\begin{IEEEeqnarray}{rCl}
	\chihat &=& -2\frac{\partial}{\partial \chi}\varphi(\chi) \IEEEnonumber\\
	&=&
	\bigg(\frac{\rho - 2 m + Q}{\lambda^{2}}\bigg)
	\mathsf{R}'_{\vm{A}^{\trans}\vm{A}}\bigg(\!\!\!-\frac{\chi}{\lambda}\bigg)	
	\IEEEnonumber\\
	&& + \, \bigg(\frac{\alpha\sigma^{2}}{\lambda^{2}}\bigg)
	\frac{\partial}{\partial \chi}
	\bigg[\chi\mathsf{R}_{\vm{A}^{\trans}\vm{A}}\bigg(\!\!\!-\frac{\chi}{\lambda}\bigg)\bigg] 
	\IEEEnonumber \\
	&=& \bigg(\frac{\alpha\sigma^{2}}{\lambda^{2}}\bigg)
	\mathsf{R}_{\vm{A}^{\trans}\vm{A}}\bigg(\!\!\!-\frac{\chi}{\lambda}\bigg)
	\IEEEnonumber\\
	&& + \,	
	\frac{1}{\lambda^{2}}
	\mathsf{R}'_{\vm{A}^{\trans}\vm{A}}\bigg(\!\!\!-\frac{\chi}{\lambda}\bigg)
	\bigg(\rho - 2 m + Q - \frac{\chi\alpha\sigma^{2}}{\lambda}\bigg).	
	\IEEEeqnarraynumspace 			
	\label{eq:chihat_R1}
\end{IEEEeqnarray}
Thus we have a formula for $\chihat$ as a function of $\chi$
and $\mse$, expressed in terms of the R-transform (and its derivative
$\mathsf{R}'$).  Comparing to
\cite[(195)]{Tulino-etal-cs2013} we see that the expressions 
are the same apart from minor scalings.  
The differences can be explained by noticing that
in \cite{Tulino-etal-cs2013}: 1) the 
noise variance is $\sigma^{2} = 1$, 2) $\lambda = \gamma^{-1}$ by definition,
and 3) the matrices are square so that $\alpha = 1$.
Thus, we conclude that $\chihat \leftrightarrow f^{\star}$
and Proposition~2 is indeed identical to \cite{Tulino-etal-cs2013},
just expressed differently.

\subsection{Regularization with $\ell_2$-norm and ``zero-norm''} 

As a first example of regularization other than $\ell_{1}$-norm, 
consider the case when
\begin{equation}
	\label{eq:l2-norm}
	\Freg(\vm{x}) = \frac{1}{2}\|\vm{x}\|^{2} = \sum_{\Nidx=1}^{\N} \frac{x_{\Nidx}^{2}}{2}.
\end{equation}
This regularization implies that in the MAP-framework, the desired signal is 
postulated to have a standard Gaussian distribution.  It is thus not surprising 
that such an assumption reduces the estimate \eqref{eq:mmse_estimate}
to the standard linear form
\begin{equation}
	\label{eq:LMMSE}
	\xvechat = \vm{A}^{\trans}(\vm{A}\vm{A}^{\trans} + \lambda\I_{\M})^{-1} \vm{y},
\end{equation}
which is independent of the parameter $\beta$.  For the 
replica analysis one obtains from 
\eqref{eq:scalar_post_ymod} and \eqref{eq:scalar_true_ymod}
\begin{IEEEeqnarray}{l}
	\xhat(y;\,\Qhat) = \frac{\Qhat}{1+\Qhat} y
	\,\overset{\Qhat = \mhat}{=}\, \frac{\mhat x^{0} + z \sqrt{\chihat}}{1+\mhat},
\end{IEEEeqnarray}
which can be interpreted as mismatched linear MMSE estimation
of $x^{0}$ from observation \eqref{eq:scalar_true_ymod}.   
From \eqref{eq:m_interp}--\eqref{eq:chi_interp} we obtain
the following simple result.

\begin{example}
\label{example:l2-regularization}
Let the distribution of the non-zero elements $\pi(x)$ 
have zero-mean, unit variance and finite moments.  For 
rotationally invariant setup and general problem 
\eqref{eq:CS_Standard_Problem_LS} with the $\ell_{2}$-regularization
we have
\begin{IEEEeqnarray}{rCl}
	\mse &=& \frac{\rho+\chihat}{(1+\mhat)^{2}},
	\IEEEeqnarraynumspace\\
	\chi &=& \frac{1}{1+\mhat}.
\end{IEEEeqnarray}
Comparing to \cite[(135)--(138)]{Tulino-etal-cs2013},
we see that the results indeed match as discussed in 
Section~\ref{sec:equivalence}.
The value of the parameter $\lambda$ that minimizes the MSE 
is $\lambda^{*} = \sigma^{2}/\rho$.  The marginal density 
$\pi(x)$ of the non-zero elements \eqref{eq:true_sourcesym_pdf_k} 
has no impact on the MSE.  The choice of the sensing matrix, 
on the other hand, does affect the MSE.
In this special case though it is the same for the 
$T$-orthogonal and row-orthogonal setups --- also for non-uniform sparsities.
\end{example}

The benefit of $\ell_{1}$ and $\ell_{2}$-norm regularizations is that 
both are of polynomial complexity.  Implementation of 
\eqref{eq:LMMSE} is trivial and for solving 
\eqref{eq:CS_Standard_Problem_LASSO} one may use standard convex optimization
tools like \texttt{cvx} \cite{cvx}.  However, one may wonder if there are 
better choices for regularization when the goal is to reconstruct a sparse vector.  
If we take the noise-free case as the guide, instead of say $\ell_{1}$-norm,
we should have a direct penalty on the number of 
non-zero elements in the source vector.  We may achieve 
this by so-called ``zero-norm'' based regularization.  
One way to write the corresponding cost function
is%
\footnote{As remarked also in \cite[Section~V-C]{Rangan-Fletcher-Goyal-2012}, 
``zero-norm'' regularization does not satisfy 
the requirement that the normalization constant in
\eqref{eq:post_pdf_xi} is well-defined for any finite 
$\M,\N$ and $\beta>0$.  Hence, appropriate limits should be 
considered for mathematically rigorous treatment.  
However, using similar analysis as given in 
\cite[Section~3.5]{Kabashima-Vehkapera-Chatterjee-2012}, it is possible
to show that the RS solution for the ``zero-norm'' regularization is 
in fact always unstable due to the discontinuous nature 
of \eqref{eq:hard_Thold}.  For this reason, we skip the formalities of 
taking appropriate limits and report the results as they arise by 
directly using the form given in \eqref{eq:l0-norm}.}
\begin{IEEEeqnarray}{rCl}
	\label{eq:l0-norm}
	\Freg(\vm{x}) &=& \sum_{\Nidx=1}^{\N} 1 \big(x_{\Nidx} \in \R\setminus \{0\}\big)
	\IEEEeqnarraynumspace\\
	&=& \text{ number of non-zero elements in $\vm{x}$}. \IEEEeqnarraynumspace
\end{IEEEeqnarray}
If the postulated and true scalar outputs are given by
\eqref{eq:scalar_post_ymod}~and~\eqref{eq:scalar_true_ymod},
respectively, we obtain
\begin{equation}
	\xhat(y;\,\Qhat) = y \cdot 1\big(|y|  > \sqrt{\smash[b]{2 \Qhat}}\big),
	\label{eq:hard_Thold}
\end{equation}
which is just hard thresholding estimator of scalar input 
\eqref{eq:scalar_true_ymod}, given mismatched 
model \eqref{eq:scalar_post_ymod}.
To proceed further, we need to fix the marginal distribution 
$\pi(x)$ of the non-zero components in \eqref{eq:true_sourcesym_pdf_k}.  
For the special case of Gaussian distribution, some algebra provides the 
following result.

\begin{example}
\label{example:l0-regularization}
Let $\pi(x) = \Gpdf{x}(0;\,1)$, that is, consider the case of Gaussian
marginals $\eqref{eq:true_sourcesym_pdf_k}$ with
 the rotationally invariant setup and general problem 
\eqref{eq:CS_Standard_Problem_LS} with the ``zero-norm'' regularization given by
\eqref{eq:l0-norm}.  Define a function
\begin{equation}
	r_{0}(h) = \e^{-h} \sqrt{\frac{h}{\pi}} + \Qfunc( \sqrt{2 h}).
\end{equation}
Then, the average MSE for the rotationally invariant setup
$\mse = \rho - 2 m + Q$ is obtained from 
\begin{IEEEeqnarray}{rCl}
	\label{eq:l0_m}
	m &=& 2 \rho r_{0} \bigg(\frac{\mhat}{\mhat^{2}+\chihat}\bigg),
	\IEEEeqnarraynumspace\\
	Q &=& 2 (1-\rho) \bigg(\frac{\chihat}{\mhat^{2}}\bigg)
	r_{0}\bigg( \frac{\mhat}{\chihat} \bigg) \IEEEnonumber\\
	&& +\, 2\rho  \bigg(\frac{\mhat^{2}+\chihat}{\mhat^{2}}\bigg)
	r_{0}\bigg( \frac{\mhat}{\mhat^{2}+\chihat} \bigg), \IEEEeqnarraynumspace 
	\label{eq:l0_Q}
\end{IEEEeqnarray}
using the condition
\begin{IEEEeqnarray}{rCl}
	\label{eq:l0_chi}
	\chi &=& \frac{2(1-\rho)}{\mhat} r_{0} \bigg(\frac{\mhat}{\chihat}\bigg) 
	+ \frac{2\rho}{\mhat} r_{0} \bigg(\frac{\mhat}{\mhat^{2}+\chihat}\bigg)
	\IEEEeqnarraynumspace 
\end{IEEEeqnarray}
and equations \eqref{eq:hatchi_rotational}--\eqref{eq:hatm_rotational}
in Proposition~\ref{prop:eigensetup}.
Furthermore, by Remark~\ref{remark:equal_parameters}, the $T$-orthogonal
case can also be obtained easily by first adding the block 
index $t$ to all variables and then 
replacing \eqref{eq:m_Torthogonal}~and~\eqref{eq:Q_Torthogonal} 
by the formulas given here.  The last modification is to
 write \eqref{eq:Rt_prop} simply as $R_{t} = \chi_{t}\mhat_{t}$ with
 the $\chi_{t}$ given in \eqref{eq:l0_chi}.
\end{example}

To illustrate the above analytical results, we have plotted 
the \emph{normalized mean square error} 
$10 \log_{10} (\mse / \rho)$~dB 
as a function of \emph{inverse compression rate} $1/\alpha$ in
Fig.~\ref{fig:example5}.  The axes are chosen so that 
the curves can be directly compared to 
\cite[Fig.~3]{Rangan-Fletcher-Goyal-2012}.  Note, however,
that we plot only the region $1/3 \leq \alpha \leq 1$ (in contrast to 
$1/3 \leq \alpha \leq 2$ there) since
this corresponds to the assumption of CS setup and one cannot
construct a row-orthogonal matrix for $\alpha>1$.  It is 
clear that for all estimators, using row-orthogonal ensemble
for measurements is beneficial compared to the standard Gaussian 
setup in this region.  Furthermore, if the source has non-uniform
sparsity and $\ell_{0}$ or $\ell_{1}$ regularization is used, 
the $T$-orthogonal setup (the black markers in the figure) 
provides an additional gain 
in average MSE for the compression ratios $\alpha = 1/2$ and $\alpha=1/3$.

\begin{figure}[tb]
\centering
\includegraphics[width=0.48\textwidth]{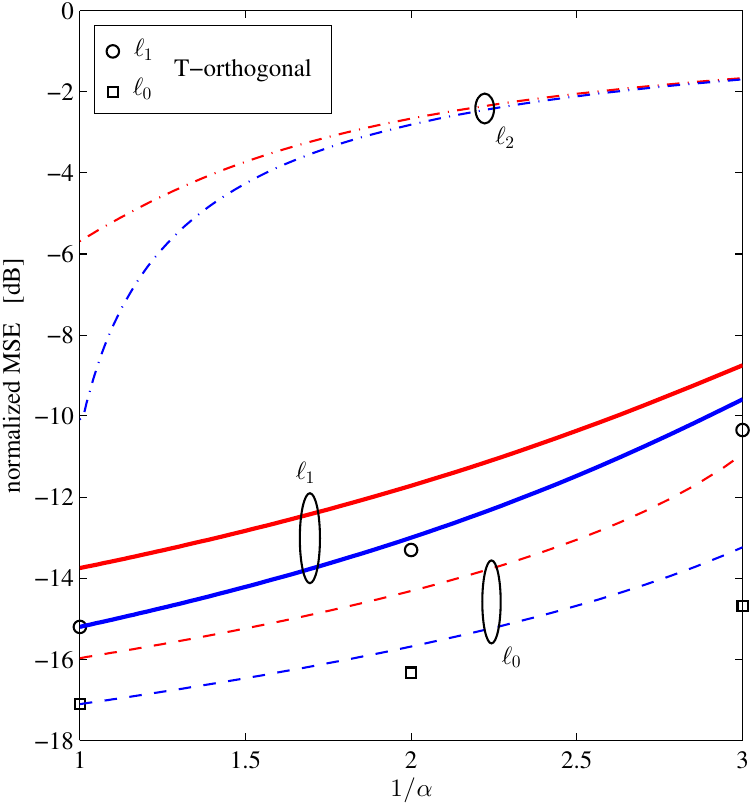}
\caption{Normalized mean square error $\mse / \rho$ in decibels vs.\ inverse compression
rate $1/\alpha$ as given by 
Examples~\ref{example:l2-regularization}~and~\ref{example:l0-regularization}.  
Rotationally invariant ensemble with arbitrary sparsity pattern and 
$T$-orthogonal case with localized sparsity, namely,
$\rho_{t} = \rho T$ for some $t \in \{1,\ldots,T\}$ and 
$\rho_{s} = 0 \; \forall s \neq t$.
Thin red lines for IID sensing matrix, thick blue lines for 
row-orthogonal case. Noise variance $\sigma^{2}=0.01$ and 
the parameter $\lambda$ is chosen so that the MSE is minimized.
The red lines match the corresponding 
curves in \cite[Fig.~3]{Rangan-Fletcher-Goyal-2012}.}
\label{fig:example5}
\end{figure}


\section{Conclusions}
\label{sec:conclusions}

The main emphasis of the present paper was in the analysis of 
$\ell_{1}$-regularized least-squares estimation, also known as 
LASSO or basis pursuit denoising, in the noisy compressed  
sensing setting.  Extensions to $\ell_{2}$-norm and 
``zero-norm'' regularization were briefly discussed.
Using the replica method from statistical physics, 
the mean square error behavior of reconstruction was derived
in the limit when the system size grows very large.
By introducing some novel results for taking an expectation of 
a matrix in an exponential form, the previous results concerning
standard Gaussian measurement matrix was extended to more general
ensembles.  As specific examples, row-orthogonal, geometric and $T$-orthogonal 
random matrices were considered in addition to the Gaussian one.
The assumption about uniform sparsity of the source was also 
relaxed and blockwise sparsity levels were allowed.

The analytical results show that while the MSE of reconstruction 
depends only on the average sparsity level of the source for 
rotationally invariant cases, the performance 
of $T$-orthogonal setup depends on the individual sparsities 
of the sub-blocks.  In case of uniform 
sparsity, row-orthogonal and $T$-orthogonal setups have provably 
the same performance.  It was also found that 
while the row-orthogonal, geometric and Gaussian setups each fall under the 
category of rotationally invariant ensemble, that is known to 
have a unique perfect recovery threshold in a noise-free setting,
with additive noise the MSE performance of these ensembles 
can be very different.

The numerical experiments revealed the fact that under all considered
settings, the standard Gaussian ensemble performed always worse than
the orthogonal constructions.  The MSE for the geometric ensemble was 
found to be an increasing function of the peak-to-average ratio of the 
eigenvalues of the sensing matrix, suggesting that spectrally uniform
sensing matrices are beneficial for recovery. When the sparsity 
was non-uniform, the ranking of the orthogonal constructions depended on the 
noise level.  For highly noisy measurements, the row-orthogonal 
measurement matrix was found to provide the best overall 
MSE, while relatively clean measurements benefited from the
$T$-orthogonal sensing matrices.
These findings show that the choice of random measurement matrix 
does have an impact in the MSE of the reconstruction when noise 
is present. Furthermore, if the source does not have a uniform 
sparsity, the effect becomes even more varied and complex.

A natural extension of the current work is to consider Bayesian 
optimal recovery and its message passing approximation for the 
matrix ensembles that differ from the standard ones.  
Indeed, since the initial submission of the present paper, 
such algorithms have been developed and shown
to benefit of matrices with structure, see for example,
\cite{Kabashima-Vehkapera-isit2014, Cakmak-Winther-Fleury-itw2014, 
	Ma-Yuan-Ping-2015, Opper-Cakmak-Winther-arxiv2015}.

\appendices


\section{Replica Analysis of $T$-Orthogonal Setup}
\label{sec:Tortho_replica}

\subsection{Free Energy}

Recall the $T$-orthogonal setup from Definition~\ref{defn:matrix_ensembles} and
let the partition of $\vm{A}$ be one that matches that of $\vm{x}$, i.e.,
\begin{equation}
	\vm{A} \vm{x} = \sum_{t=1}^{T} \vm{O}_{t} \vm{x}_{t}.
\end{equation}
We then recall \eqref{eq:freeE_replica_real}, 
invoke the replica trick introduced in Section~\ref{sec:freeE_replicas},
and assume that the limits commute.  The normalized free energy of the system reads thus 
\begin{equation}
	\label{eq:freeE_replica_appendix_1}
	f =  -\frac{1}{\T} \lim_{\NR\to 0^{+}}
	\frac{\partial}{\partial \NR} \lim_{\beta,\M\to\infty} \frac{1}{\beta \M}
	\ln \Xi_{\beta,\M}(\NR),
\end{equation}
where we denoted 
\begin{IEEEeqnarray}{l}
	\Xi_{\beta,\M}(\NR) =
	\E_{\vm{w},\{\vm{O}_{t}\}} \int \PM(\vm{x}^{0};\, \{\rho_{\tx}\}) 
	\exp\bigg(\!\!-\beta\sum_{a=1}^{\NR}\Freg(\vm{x}^{a})\bigg)
	\IEEEnonumber\\
	\quad \times \exp \bigg( - 
	\frac{\beta}{2\lambda} \sum_{a=1}^{\NR} \bigg\| \sigma\vm{w} - \sum_{t=1}^{T}\vm{O}_{t} 
	\Delta\vm{x}_{t}^{a}\bigg\|^{2} \bigg) \prod_{a = 0}^{\NR} \dx \vm{x}^{a}
	\IEEEeqnarraynumspace
	\label{eq:AVGpartFu}
\end{IEEEeqnarray}
for notational convenience.
The form of \eqref{eq:AVGpartFu} implies 
that $\{\vm{x}^{a}\}$ are independently drawn according to
\eqref{eq:post_pdf_xi},
$\vm{x}_{t}^{0}$ has the same distribution as $\vm{x}_{t}$, i.e.,
elements drawn according to \eqref{eq:true_sourcesym_pdf_k} for 
each block $t=1,\ldots,T$,  and 
$\Delta\vm{x}_{t}^{a} = \vm{x}_{t}^{0}-\vm{x}_{t}^{a}$ for 
$a=1,\ldots,n$.   The outer expectation is w.r.t.\ the 
additive noise and measurement matrices.

For each set of random vectors $\{\Delta\vm{x}_{t}^{a}\}_{a=1}^{\NR}$, let 
us now construct a matrix $\vm{S}_{t} \in \R^{n \times n}$
for all $t=1,\ldots,T$ whose $(a,b)$th element represents the empirical covariance
\begin{IEEEeqnarray}{rCl}
	S^{[a,b]}_{t} &=& \frac{1}{\M}\Delta\vm{x}_{t}^{a} \cdot \Delta\vm{x}_{t}^{b} 
	\IEEEnonumber\\
	\label{eq:Smtx_1}
	&=& \frac{\|\vm{x}_{t}^{0}\|^{2}}{\M} - \frac{\vm{x}_{t}^{0}\cdot \vm{x}_{t}^{b}}{\M}
	- \frac{\vm{x}_{t}^{a} \cdot \vm{x}_{t}^{0}}{\M} 
	+ \frac{\vm{x}_{t}^{a} \cdot \vm{x}_{t}^{b}}{\M} \IEEEeqnarraynumspace\\
	&=& Q^{[0,0]}_{t} - Q^{[0,b]}_{t} - Q^{[a,0]}_{t} + Q^{[a,b]}_{t}.
	\label{eq:Smtx_2}
\end{IEEEeqnarray}
We also construct a similar set of matrices $\{\vm{Q}_{t}\}_{t=1}^{T}$ whose 
elements $\{Q^{[a,b]}_{t}\}$ have the obvious definitions. 
The rotational symmetry of distributions $\vm{w}_t$ and $\{\vm{Q}_t\}$ 
indicates that 
\begin{equation}
	\E_{\vm{w},\{\vm{O}_{t}\}}
	\exp \bigg(-\frac{\beta}{2\lambda} 
	\sum_{a=1}^{\NR} \bigg\| \sigma\vm{w} - \sum_{t=1}^{T}\vm{O}_{t} 
	\Delta\vm{x}_{t}^{a}\bigg\|^{2} \bigg)	
\end{equation}
becomes a function of $\{Q_t^{[a,b]}\}$ for any fixed set of $\{\Delta\vm{x}_{t}^{a}\}_{a=1}^{\NR}$.
In addition, inserting a set of trivial identities
\begin{equation}
	1 = \M^{\NR(\NR+1)/2} \int 
	\prod_{0 \leq a\leq b \leq \NR}\delta(\M Q^{[a,b]}_{t} - \vm{x}_{t}^{a} \cdot \vm{x}_{t}^{b})
	\dx Q^{[a,b]}_{t}
	\label{eq:identity_trivial}
\end{equation}
for $t=1,2,\ldots,T$ into  \eqref{eq:AVGpartFu} and performing the 
integration over $\{\vm{x}^a\}$ yields an expression that allows for 
saddle point evaluation of $\Xi_{\beta,\M}(\NR)$ with respect to $\{Q_t^{[a,b]}\}$.

To proceed with the analysis, we make the RS assumption which states that
at the dominant saddle point, 
the overlap matrices $\{\vm{Q}_t\}$ are invariant under the rotation 
of the replica indexes $a = 1,2,\ldots,\NR$ (see also \eqref{eq:Q_extended})
\begin{IEEEeqnarray}{rll}
	\label{eq:RS_ass_first}
	Q_{t}^{[0,0]} &\;= r_{t}, \\
	Q_{t}^{[0,\nrb]} = Q_{t}^{[\nr,0]} &\;=  m_{t} \qquad 
	&\forall a,b\geq 1, \IEEEeqnarraynumspace\\
	Q_{t}^{[\nr,\nr]} &\;=  Q_{t} \qquad &\forall a\geq 1, \\
	Q_{t}^{[\nr,\nrb]} &\;= q_{t} \qquad &\forall a \neq b\geq 1.
	\label{eq:RS_ass_last}
\end{IEEEeqnarray}
Under the RS assumption, the matrices $\vm{S}_{t}$ introduced above 
have also a simple form
\begin{equation}
	\label{eq:Smtx_i_RS}
	\vm{S}_{t} = S_{t}^{[1,2]}\vm{1}_{\NR}\vm{1}_{\NR}^{\trans}
	+ (S_{t}^{[1,1]}-S_{t}^{[1,2]})\vm{I}_{\NR},
\end{equation}
where $\vm{1}_{\NR} = [1\; \cdots\;\, 1]^{\trans} \in \R^{\NR}$ is an all-ones vector and
\begin{IEEEeqnarray}{rCl}
	S_{t}^{[1,1]} &=& r_{t} - 2 m_{t} + Q_{t} \IEEEeqnarraynumspace\\
	S_{t}^{[1,2]} &=& r_{t} - 2 m_{t} + q_{t}.
\end{IEEEeqnarray}

When $\vm{O}_t$ are independent Haar matrices, the vectors
$\vm{O}_t \Delta \vm{x}_t^a$ are distributed 
on the $M$-dimensional hyper spheres of radius 
$||\vm{O}_t \Delta \vm{x}_t^a ||= ||\Delta \vm{x}_t^a||=\sqrt{M Q_t}$
that are centered at the origin (see  
Appendix~\ref{sec:matrix_integrals}). 
As $\vm{O}_t$ are sampled independently among $t=1,2,\ldots,T$, 
the cross-correlations for all $t_1 \neq t_2$ reduce to
\begin{IEEEeqnarray}{l}
	\E_{\{\vm{O}_{t}\}}\big\{( \vm{O}_{t_1} \Delta 
	\vm{x}_{t_1}^a )^{\trans} 
	(\vm{O}_{t_2} \Delta \vm{x}_{t_2}^b)\big\} \IEEEnonumber\\
	\quad =
	(\Delta \vm{x}_{t_1}^a )^{\trans} 
	\E_{\{\vm{O}_{t}\}} \{\vm{O}_{t_1}^{\trans} \vm{O}_{t_2}\}
	\Delta \vm{x}_{t_2}^b  \IEEEeqnarraynumspace\IEEEnonumber\\
	\quad = (\Delta \vm{x}_{t_1}^a )^{\trans}  \vm{0}_M \Delta \vm{x}_{t_2}^b=0, 
	\IEEEeqnarraynumspace
\end{IEEEeqnarray}
where $\vm{0}_M$ is the $M \times M$ zero matrix. 
On the other hand, given $t_1=t_2=t$ we obtain
\begin{IEEEeqnarray}{l}
	M^{-1} \E_{\{\vm{O}_{t}\}}\big\{( \vm{O}_t \Delta {\vm{x}}_t^a )^{\trans}
	( \vm{O}_t \Delta {\vm{x}}_{t}^b) \big\} \IEEEnonumber\\
	\quad = M^{-1} (\Delta {\vm{x}}_t^a )^{\trans}  
	\E_{\{\vm{O}_{t}\}} \{\vm{O}_t^{\trans} \vm{O}_t \} \Delta {\vm{x}}_{t}^b \IEEEnonumber\\
	\quad = M^{-1} (\Delta {\vm{x}}_t^a )^{\trans} 
	\vm{I}_M \Delta {\vm{x}}_{t}^b = S_t^{[a,b]}.
	\IEEEeqnarraynumspace
	\label{eq:Ox_correlation}
\end{IEEEeqnarray}
Since $S_t^{[a,b]}$ are in general non-zero for any pairs of replica 
indexes $a,b=1,2,\ldots,\NR$, the expectation in \eqref{eq:AVGpartFu} 
is nontrivial to compute. 
However, these correlations can be decoupled by linearly transforming the 
variables using a matrix 
\begin{equation}
	\label{eq:Emtx}
	\vm{E} = 
	\begin{bmatrix}
	\vm{e}_{1} & \vm{e}_{2} & \cdots & \vm{e}_{\NR}
	\end{bmatrix} \in \R^{\NR\times\NR},
\end{equation}
that satisfies 
$\vm{E}^{\trans}\vm{E} = \vm{E}\vm{E}^{\trans} = \vm{I}_{\NR},$
and $\vm{e}_{1} = \vm{1}_{\NR}/\sqrt{\NR}$  by definition.
More precisely, if we let 
$\begin{bmatrix} \Delta \tilde{\vm{x}}_t^1 & \cdots &  \Delta \tilde{\vm{x}}_t^{\NR} \end{bmatrix}
= \begin{bmatrix} \Delta \vm{x}_t^1 & \cdots &  \Delta \vm{x}_t^{\NR} \end{bmatrix}\vm E$
be the transformed vectors,
\begin{IEEEeqnarray}{l}
	\frac{1}{M} \E_{\{\vm{O}_{t}\}}\Big\{
	\big( \begin{bmatrix} \vm{O}_t \Delta \tilde{\vm{x}}_{t}^1 & \cdots &  \vm{O}_t \Delta \tilde{\vm{x}}_t^{\NR} \end{bmatrix}\big)^{\trans}  
	\IEEEnonumber\\
	\qquad\qquad\quad \times 
	\begin{bmatrix}  \vm{O}_t\Delta \tilde{\vm{x}}_{t}^1 & \cdots &  \vm{O}_t\Delta \tilde{\vm{x}}_t^{\NR} \end{bmatrix}\Big\} \IEEEnonumber\\
	\quad =\frac{1}{M} \E_{\{\vm{O}_{t}\}} \Big\{
	\big(\vm{O}_{t}\begin{bmatrix} \Delta \vm{x}_{t}^1 & \cdots & \Delta \vm{x}_t^{\NR} \end{bmatrix} \vm{E}
	\big)^{\trans} \IEEEnonumber\\
	\qquad\qquad\qquad\quad \times 
	\big( \vm{O}_{t}\begin{bmatrix} \Delta \vm{x}_{t}^1 & \cdots & \Delta \vm{x}_t^{\NR} \end{bmatrix}  \vm{E} \big) \Big\} \IEEEnonumber\\
	\quad  =\frac{1}{M} 
	\big( \begin{bmatrix}   \Delta \vm{x}_{t}^1 & \cdots &  \Delta \vm{x}_t^{\NR} \end{bmatrix}
	\vm{E} \big)^{\trans}  \IEEEnonumber\\
	\qquad\qquad \times \E_{\{\vm{O}_{t}\}} \big\{\vm{O}_t^{\trans} \vm{O}_t\big\}
	\begin{bmatrix}  \Delta \vm{x}_{t}^1 & \cdots &   \Delta \vm{x}_t^{\NR} \end{bmatrix} \vm{E} 
	\IEEEnonumber\\
	\quad = \frac{1}{M} 
	\big( \begin{bmatrix}   \Delta \vm{x}_{t}^1 & \cdots &  \Delta \vm{x}_t^{\NR} \end{bmatrix}
	\vm{E} \big)^{\trans}  
	\begin{bmatrix}  \Delta \vm{x}_{t}^1 & \cdots &   \Delta \vm{x}_t^{\NR} \end{bmatrix} \vm{E} 
	\IEEEnonumber\\
	\quad \triangleq \tilde{\vm{S}}_{t}. 
\end{IEEEeqnarray}
But $\tilde{\vm{S}}_{t}\in \R^{\NR\times\NR}$ is just
\begin{IEEEeqnarray}{rCl}
	\tilde{\vm{S}}_{t}
	&=& \vm{E}^{\trans} \vm{S}_{t} \vm{E} \IEEEnonumber\\
	&=&\diag(\NR S_{t}^{[1,2]} + S_{t}^{[1,1]}-S_{t}^{[1,2]},
	\IEEEnonumber\\&& \qquad\quad 
	\underbrace{S_{t}^{[1,1]}-S_{t}^{[1,2]},\, \ldots,\,
		S_{t}^{[1,1]}-S_{t}^{[1,2]}}_{\NR-1 \; \text{times}}),
	\IEEEeqnarraynumspace
	\label{eq:StPrime_diag}
\end{IEEEeqnarray}
since $\vm{e}_{a}^{\trans} \vm{1}_{\NR} 
= \vm{e}_{a}^{\trans} \vm{e}_{1}
= 0$ for all $a=2,\ldots,\NR$ and, thus,
\begin{IEEEeqnarray}{l}
	\frac{1}{M} \E_{\{\vm{O}_t\}} 
	\big\{(\vm{O}_t \Delta \tilde{\vm{x}}_{t}^a)^{\trans}
	(\vm{O}_t \Delta \tilde{\vm{x}}_{t}^b)\big\} \IEEEnonumber\\
	\quad =
	\begin{cases}
		0 & \text{if } a \neq b, \\
		\NR (r_{t} - 2 m_{t} + q_{t}) 
		+ Q_{t}-q_{t} \quad & \text{if } a = b = 1, \\
		Q_{t}-q_{t} & \text{if } a = b = 2,\ldots,\NR,
	\end{cases} \IEEEnonumber\\
	\label{eq:StildeAB}
\end{IEEEeqnarray}
holds.  The above shows that for given $\{\Delta\vm{x}_{t}^{a}\}_{a=1}^{\NR}$, 
the set of vectors $\{\vm{O}_t\Delta\tilde{\vm{x}}_{t}^{a}\}$ are
independent among $t = 1,\ldots,T$ and
also uncorrelated in the space of replicas, as indicated by \eqref{eq:StildeAB}.
This is in contrast to the original set $\{\vm{O}_t\Delta\vm{x}_{t}^{a}\}$
whose replica space correlation structure \eqref{eq:Ox_correlation} is much more 
cumbersome to deal with.

To proceed with the analysis, 
we first notice that since $\vm{E}\vm{E}^{\trans} = \vm{I}_{\NR}$, 
the quadratic term in \eqref{eq:AVGpartFu} can be expressed as
\begin{equation}
	\sum_{a=1}^{\NR} \bigg\| \sum_{t=1}^{T} \vm{O}_{t} \Delta\vm{x}_{t}^{a}\bigg\|^{2}
	=\sum_{a=1}^{\NR} \bigg\| \sum_{t=1}^{T} \vm{O}_t \Delta\tilde{\vm{x}}_{t}^{a}\bigg\|^{2}  
\end{equation}
using the uncorrelated random vectors 
$\{\vm{O}_t \Delta \tilde{\vm{x}}_t^{a} \}$.
For notational convenience, we 
define next an auxiliary matrix 
$\vm{E}' = \begin{bmatrix}
\vm{e}'_{1} & \cdots & \vm{e}'_{\NR} 
\end{bmatrix} = \vm{E}^{\trans}$
so that 
\begin{equation}
	\Delta\vm{x}_{t}^{a} = 
	\begin{bmatrix}
	\Delta\tilde{\vm{x}}_{t}^{1} & \cdots & \Delta\tilde{\vm{x}}_{t}^{\NR}
	\end{bmatrix}\vm{e}'_{a}, \qquad a = 1,\ldots,\NR,
	\label{eq:transformed_deltax_1}
\end{equation}
where $\{\vm{e}'_{a}\}$ again forms an orthonormal set that is 
independent of $t$.  Then, after the transformation \eqref{eq:transformed_deltax_1}, 
the linear term in \eqref{eq:AVGpartFu} becomes 
\begin{IEEEeqnarray}{rCl}
	\sum_{t=1}^{T} \sum_{a=1}^{\NR}
	\vm{w}^{\trans} \vm{O}_{t} \Delta\vm{x}_{t}^{a}  
	&=& \sum_{t=1}^{T} \sum_{a=1}^{\NR} \vm{w}^{\trans} \vm{O}_{t}
	\begin{bmatrix}
		\Delta\tilde{\vm{x}}_{t}^{1} & \cdots & \Delta\tilde{\vm{x}}_{t}^{\NR}
	\end{bmatrix}\vm{e}'_{a} 
	\IEEEnonumber\\
	&=& \sqrt{\NR} \vm{w}^{\trans}\sum_{t=1}^{T}\vm{O}_{t}\Delta\tilde{\vm{x}}_{t}^{1},
	\label{eq:transF_linterm_1}
\end{IEEEeqnarray}
where the we used the fact that
\begin{equation}
\sum_{a=1}^{\NR} \vm{e}'_{a} = \vm{E}^{\trans} \vm{1}_{\NR}
= 
\begin{bmatrix}
\sqrt{\NR} & 0 & \cdots & 0
\end{bmatrix}^{\trans}.
\end{equation}
Combining the above findings and re-arranging
implies that \eqref{eq:AVGpartFu} can be equivalently expressed as
\begin{IEEEeqnarray}{l}
	\!\Xi_{\beta,\M}(\NR) \IEEEnonumber\\
	\quad =\E_{\vm{w},\{\vm{O}_{t}\}}
	\int \exp \bigg(-\frac{\beta}{2\lambda} \bigg\|\sqrt{\NR\sigma^{2}} \vm{w} -
	\sum_{t=1}^{T} \vm{O}_t\Delta\tilde{\vm{x}}_{t}^{1}\bigg\|^{2}\bigg)
	\IEEEnonumber\quad\\
	\qquad \times 
	\exp \bigg(-\frac{\beta}{2\lambda} \sum_{a=2}^{\NR}
	\bigg\| \sum_{t=1}^{T}  \vm{O}_t \Delta\tilde{\vm{x}}_{t}^{a}\bigg\|^{2} \bigg)
	\IEEEeqnarraynumspace \IEEEnonumber\\  
	\qquad \times
	\PM(\vm{x}^{0};\, \{\rho_{\tx}\}) 
	\exp\bigg(\!\!-\beta\sum_{a=1}^{\NR}\Freg(\vm{x}^{a})\bigg)
	\prod_{a = 0}^{\NR} \dx \vm{x}^{a}.
	\IEEEeqnarraynumspace
	\label{eq:Iu_2}
\end{IEEEeqnarray}

Next, recall the definition of the matrix $\vm{Q}_{t}$ whose
elements are as given in \eqref{eq:Smtx_2}.
From the identity of \eqref{eq:identity_trivial}
we obtain the probability weight for $\vm{Q}_{t}, t = 1,\ldots, T$ as
\begin{IEEEeqnarray}{rCl}
	p_{\beta,\M}(\vm{Q}_{t};\, \NR)
	&=& 
	\frac{1}{z^{\NR}_{\beta,\M}}
	\int p(\vm{x}_{t}^{0}) \dx \vm{x}_{t}^{0}  \prod_{a=1}^{\NR} 
	\Big(\e^{-\beta \Freg(\vm{x}_{t}^{a})}\dx \vm{x}_{t}^{a} \Big) 
	\IEEEeqnarraynumspace\IEEEnonumber\\ 
	&& \times
	\prod_{0 \leq a\leq b \leq \NR}
	\delta(\vm{x}_{t}^{a} \cdot \vm{x}_{t}^{b}-\M Q^{[a,b]}_{t}),
	\IEEEeqnarraynumspace
	\label{eq:PQmeas}
\end{IEEEeqnarray}
where $\Freg(\vm{x}_{t}^{a})$ is interpreted as in 
\eqref{eq:Fdecoupling} to be a sum of scalar regularization 
functions and the normalization constant
is given by $z_{\beta,\M} = z_{\beta}\M^{-(\NR+1)/2}$, where
$z_{\beta}$ is as in \eqref{eq:post_pdf_xi}.

Then we proceed as follows: 
\begin{enumerate}
	\item Fix the matrices $\{\vm{Q}_{t}\}_{t=1}^{T}$ so that the lengths 
	$\tilde{S}^{[a,a]}_{t}, a=1,\ldots,\NR$ in \eqref{eq:StildeAB} 
	are constant and, thus, $\{\Delta\tilde{\vm{x}}_{t}^{a}\}$ have
	fixed (squared) lengths.  Then, assuming $\M$ grows without bound,
	average over the joint distribution of $\vm{w}, \{\vm{O}_{t}\}$ and
	$\{\Delta\tilde{\vm{x}}_{t}^{a}\}$,
	given $\{\vm{Q}_{t}\}_{t=1}^{T}$.
	\item Average the obtained result w.r.t.\
	$p_{\beta,\M}(\vm{Q}_{t};\, \NR)$ as given in \eqref{eq:PQmeas} when
	$\beta\to\infty$.
\end{enumerate}
The first step may be achieved by separately averaging 
over the replicas $a = 1,\ldots,\NR$ using the following result.
Note this is always possible since we consider the setting of 
large $\M$ and hence $\NR \ll \M$.

\begin{lemma}
	\label{lemma:Ffunc_Tortho}
	
	Let $\{\vm{u}_{t}\}_{t=1}^{T}$ be a set of 
	length-$\M$  vectors that satisfy $\|\vm{u}_{t}\|^{2} = \M \nu_{t}$
	for some given non-negative reals $\{\nu_{t}\}$.  
	Let $\{\vm{O}_{t}\}$ a set of independent Haar matrices and define
	\begin{IEEEeqnarray}{rCl}
		\e^{\M G_{\beta,\lambda}(\sigma^{2},\{\nu_{t}\})}
		&=& \E_{\vm{w},\{\vm{O}_{t}\}}
		\e^{-\frac{\beta}{2\lambda} \| \sigma\vm{w} - \sum_{t=1}^{T}\vm{O}_{t} \vm{u}_{t}\|^{2}},
		\IEEEeqnarraynumspace
	\end{IEEEeqnarray}
	where $\vm{w}$ is a standard Gaussian random vector.  Then, for 
	large $\M$
	\begin{IEEEeqnarray}{l}
		G_{\beta,\lambda}(\sigma^{2},\{\nu_{t}\}) = - \frac{1}{2} 
		\bigg(T - \ln \lambda +\sum_{t=1}^{T}  \ln (\beta\nu_{t}) \bigg)
		\IEEEnonumber\\
		+ \frac{1}{2}\extr_{\{\Lambda_{t}\}}
		\bigg\{\! \sum_{t=1}^{T} \big[\Lambda_{t} (\beta \nu_{t}) -\ln \Lambda_{t}\big]
		\!-\! \ln \bigg(\lambda+ \beta\sigma^{2} 
		\! + \sum_{t=1}^{T} \frac{1}{\Lambda_{t}} \bigg)\!\bigg\},
		\IEEEnonumber\\
		\label{eq:Ffunc_final_lemma}
	\end{IEEEeqnarray}
	where we have omitted terms of the order $O(1/M)$.  
\end{lemma}

\begin{proof}
	Proof is given in Appendix~\ref{sec:Ffunc_Tortho}.
\end{proof}

Since $\{\vm{O}_t\Delta\tilde{\vm{x}}_{t}^{a}\}_{a=1}^{\NR}$ are uncorrelated,
we can apply the above result to \eqref{eq:Iu_2} 
separately for all replica indexes $a = 1,\ldots,\NR$ in order to 
evaluate the expectations w.r.t.\ $\vm{w}$ and $\{\vm{O}_{t}\}$.  
Thus, for $\NR \ll \M$, we get
\begin{figure*}
	\begin{IEEEeqnarray}{l}
		\frac{1}{\M} \ln \Xi_{\beta,\M}(\NR) \IEEEnonumber\\
		= \frac{1}{\M} \ln
		\int \prod_{t=1}^{T} \big[
		p_{\beta,\M}(\vm{Q}_{t};\, \NR) \dx \vm{Q}_{t} \big]
		\E_{\vm{w},\{\vm{O}_{t}\}} 
		\Big\{\e^{-\frac{\beta}{2\lambda} \|\sqrt{\NR\sigma^{2}} \vm{w} -
			\sum_{t=1}^{T} \vm{O}_{t} \Delta\tilde{\vm{x}}_{t}^{1}\|^{2}}\Big\}
		\bigg[\E_{\vm{w},\{\vm{O}_{t}\}}
		\Big\{\e^{-\frac{\beta}{2\lambda}
			\| \sum_{t=1}^{T} \vm{O}_{t} \Delta\tilde{\vm{x}}_{t}^{2}\|^{2}}\Big\}
		\bigg]^{\NR-1} 
		\IEEEeqnarraynumspace\IEEEnonumber\\
		= \frac{1}{\M} \ln \int \prod_{t=1}^{T} \big[
		p_{\beta,\M}(\vm{Q}_{t};\, \NR) \dx \vm{Q}_{t} \big]
		\e^{\M G_{\beta,\lambda} 
			(\NR\sigma^{2},\{\NR (r_{t} - 2 m_{t} + q_{t}) + Q_{t} - q_{t}\}) }
		\e^{\M (\NR-1) G_{\beta,\lambda} (0,\{Q_{t} - q_{t}\}) }
		\label{eq:I_Q_sigma}
	\end{IEEEeqnarray}
	\hrulefill
\end{figure*}
\eqref{eq:I_Q_sigma} at the top of the next page, 
where $\NR$ is now just a parameter in $G_{\beta,\lambda}$ and 
is not enforced to be an integer by the function itself.
The next step is to compute the integral over $\{\vm{Q}_{t}\}$.
With some abuse of notation, we start by
using the Dirac's delta identity \eqref{eq:dirac_fft} to write
\begin{IEEEeqnarray}{l}
	p_{\beta,\M}(\vm{Q}_{t};\, \NR)
	=  \frac{1}{z^{\NR}_{\beta,\M}}
	\int\bigg(\prod_{0 \leq a\leq b \leq \NR} \frac{\dx \tilde{Q}^{[a,b]}_{t}}{2 \pi \im}\bigg)
	\IEEEnonumber\\ \qquad\times
	\exp \bigg(\M \sum_{0 \leq a\leq b \leq \NR} \tilde{Q}^{[a,b]}_{t} Q^{[a,b]}_{t}\bigg)
	\mathcal{V}_{\beta,\M}(\tilde{\vm{Q}}_{t};\, \NR),
	\IEEEeqnarraynumspace
	\label{eq:PQmeas2}
\end{IEEEeqnarray}
where 
$\{\tilde{\vm{Q}}_{t}\}_{t=1}^{T}$ is a set of 
transform domain matrices whose elements are $\{\tilde{Q}^{[a,b]}_{t}\}$ and
\begin{IEEEeqnarray}{rCl}
	\mathcal{V}_{\beta,\M}(\tilde{\vm{Q}}_{t};\, \NR) &=&
	\int p(\vm{x}_{t}^{0}) \dx \vm{x}_{t}^{0}
	\prod_{a=1}^{\NR} 
	\Big(\e^{-\beta \Freg(\vm{x}_{t}^{a})} \dx \vm{x}_{t}^{a} \Big) 
	\IEEEnonumber\\ &&\times
	\exp \bigg(-\sum_{0 \leq a\leq b \leq \NR}
	\tilde{Q}^{[a,b]}_{t} \vm{x}_{t}^{a}\cdot\,\vm{x}_{t}^{b} \bigg).
	\IEEEeqnarraynumspace
	\label{eq:V_defn}
\end{IEEEeqnarray}
Note that $\NR$ has to be an integer in \eqref{eq:V_defn} and 
the goal is thus to write it in a form where $\NR$ can be 
regarded as a real valued non-negative variable.
One can then verify that since $\tilde{Q}_{t}^{[0,0]}$ 
is connected only to the zeroth replica,
$\tilde{Q}_{t}^{[0,0]} \to 0$ when $\NR\to 0$.  Therefore, 
it plays no role in the evaluation of the 
asymptotic free energy and consequently, the MSE.  
This is indeed a common feature of replica symmetric
solution and similar conclusion can be found, e.g., 
in \cite{Tulino-etal-cs2013}, \cite{Tanaka-2002} and \cite{Guo-Verdu-2005Jun}.  
To simplify notation, we
therefore omit $\tilde{Q}_{t}^{[0,0]}$ from further consideration.

With some foresight we now impose the RS assumption on $\{\tilde{\vm{Q}}_{t}\}$ 
via auxiliary parameters 
$\{\Qhat_{t},\mhat_{t},\chihat_{t}\}$ as
\begin{IEEEeqnarray}{rll}
	\tilde{Q}_{t}^{[\nr,0]} = 
	\tilde{Q}_{t}^{[0,\nrb]} &\;= -\beta \mhat_{t}, 
	\qquad &\forall a,b\geq 1,\\
	\tilde{Q}_{t}^{[\nr,\nr]} &\;= 
	\frac{\beta \Qhat_{t} - \beta^{2}\chihat_{t}}{2},
	\qquad &\forall a\geq 1, 
	\IEEEeqnarraynumspace \\
	\tilde{Q}_{t}^{[\nr,\nrb]} &\;= -\beta^{2}\chihat_{t},
	\qquad &\forall a \neq b\geq 1. \IEEEeqnarraynumspace
\end{IEEEeqnarray}
Recalling that the elements of $\vm{x}_{t}^{0}$ 
are IID according to \eqref{eq:true_sourcesym_pdf_k}
with $\pi(x) = \Gpdf{x}(0;\,1)$
simplifies the function $\mathcal{V}_{\beta,\M}$ 
under the RS assumption
to $\mathcal{V}_{\beta,\M}(\tilde{\vm{Q}}_{t};\, \NR)
= [\mathcal{V}_{\beta}(\hat{\vm{Q}}_{t};\, \NR)]^{\M}$
where 
\begin{IEEEeqnarray}{l}
	\mathcal{V}_{\beta}(\hat{\vm{Q}}_{t};\, \NR) =
	\int \bigg(\prod_{a=1}^{\NR}  \dx x_{t}^{a}\bigg) 
	\exp\bigg(-\beta \sum_{a=1}^{\NR} \Freg(x_{t}^{a})\bigg)
	\IEEEnonumber\\
	\times \Bigg\{(1-\rho_{t})
	\exp\bigg[-\frac{\beta\Qhat_{t}}{2} \sum_{a=1}^{\NR} (x_{t}^{a})^{2}
	+ \frac{1}{2} \bigg(\beta\sqrt{\chihat_{t}}\sum_{a=1}^{\NR}x_{t}^{a} \bigg)^{2}\bigg]  
	\IEEEnonumber\\
	+ \rho_{t} \exp\bigg[-\frac{\beta\Qhat_{t}}{2} \sum_{a=1}^{\NR} (x_{t}^{a})^{2}
	+ \frac{1}{2} 
	\bigg(\beta\sqrt{\chihat_{t}+\mhat_{t}^{2}}\sum_{a=1}^{\NR}x_{t}^{a} \bigg)^{2}\bigg]
	\Bigg\}, \IEEEnonumber\\
	\label{eq:V_defn2}
\end{IEEEeqnarray}
and $\hat{\vm{Q}}_{t}$ should be read as a shorthand for 
the set $\{\chihat_{t},\Qhat_{t},\mhat_{t}\}$.
To assess the integrals in \eqref{eq:V_defn2} w.r.t.\ 
the replicated variables we first 
decouple the quadratic terms that have summations inside by using
\eqref{eq:Gint}.  By the fact that 
all integrals for $a = 1,2,\ldots,\NR$ are identical 
we obtain 
\begin{IEEEeqnarray}{l}
	\mathcal{V}_{\beta}(\hat{\vm{Q}}_{t};\, \NR) \IEEEnonumber\\ 
	\quad =
	(1-\rho_{t}) \int \bigg\{
	\int \e^{-\beta [\Qhat_{t} x_{t}^{2} / 2 - z_{t}	\sqrt{\chihat_{t}} x_{t}
		+ \Freg(x_{t}) ]} \dx x_{t} \bigg\}^{\NR}\Dx z_{t}
	\IEEEnonumber\\
	\qquad  
	+ \rho_{t} \int \bigg\{\int
	\e^{-\beta [ \Qhat_{t} x_{t}^{2}/2 
		- z_{t}\sqrt{\chihat_{t}+\mhat^{2}_{t}} x_{t} + \Freg(x_{t})]} \dx x_{t}
	\bigg\}^{\NR} \Dx z_{t}, \IEEEnonumber\\  
	\label{eq:Vb_decoupled}  
\end{IEEEeqnarray}
where 
$\Dx z_{t} = \dx z_{t} \e^{-z_{t}^{2}/2}/\sqrt{2 \pi}$ is the standard Gaussian
measure. 
For large $\beta$ we may then employ the saddle-point integration
w.r.t.\ $x_{t}$.  If we now specialize to 
LASSO reconstruction \eqref{eq:CS_Standard_Problem_LASSO}
so that the per-element regularization function is $\Freg(x) = |x|$, 
we may define
\begin{IEEEeqnarray}{rCl}
	\phi(y;\, \Qhat)
	&=& \min_{x\in\R} \bigg\{
	\frac{\Qhat}{2} x^{2} - y x + |x| \bigg\} \IEEEnonumber\\ [2mm]
	&=& 
	\begin{cases}
		\displaystyle
		- \frac{( |y| - 1)^{2}}{2 \Qhat}, \qquad & |y| > 1, \\
		0 & \text{otherwise,}
	\end{cases} \IEEEeqnarraynumspace
	\label{eq:func_phi}
\end{IEEEeqnarray}
where the second equality follows 
by the fact that the $x$ that minimizes the cost in \eqref{eq:func_phi} is 
given by%
\footnote{
	Note that $\xhat (y;\,\Qhat)$ can be interpreted as soft thresholding 
	of observation $y$.  Compare the above also to \eqref{eq:xhat_asderivative}.
	For further discussion on the relevance and 
	interpretation of this function, see Section~\ref{sec:extensions_and_sketch}.}
\begin{equation}
	\xhat (y;\,\Qhat) = 
	\begin{cases}
		\displaystyle
		\frac{y-1}{\Qhat}, \qquad & \text{if } y > 1; \\
		0, & \text{if } |y| \leq 1;\\
		\displaystyle
		\frac{y+1}{\Qhat}, \qquad & \text{if } y < -1. \\
	\end{cases}
\end{equation}
The saddle-point method then provides the following expression
\begin{IEEEeqnarray}{rCl}
	\!\!\mathcal{V}_{\beta}(\hat{\vm{Q}}_{t};\, \NR) 
	&=& 
	(1-\rho_{t}) \int 
	\exp \Big[-\beta \NR \phi\big(z_{t}\sqrt{\chihat_{t}};\, \Qhat_{t}\big) \Big]\Dx z_{t} 
	\IEEEnonumber\\
	&&+ \rho_{t}
	\int \exp \Big[-\beta \NR 
	\phi\big(z_{t}\sqrt{\chihat_{t}+\mhat^{2}_{t}};\, \Qhat_{t}\big)\Big] \Dx z_{t}. 
	\IEEEnonumber\\
\end{IEEEeqnarray}
Note that the structure of the equations does not force 
$\NR$ to be an integer anymore, so we assume that analytical
continuation can be used to take the limit $\NR \to 0.$
This provides $\mathcal{V}_{\beta}(\hat{\vm{Q}}_{t};\, \NR) \to 1$ for the data
dependent part of the probability weight \eqref{eq:PQmeas2}, which is 
consistent with \eqref{eq:V_defn}.

Returning to \eqref{eq:PQmeas2} and 
denoting $\chi_{t} = \beta (Q_{t}-q_{t})$,
we have under RS ansatz 
\begin{IEEEeqnarray}{l}
	p_{\beta,\M}(\vm{Q}_{t};\, \NR) \IEEEnonumber\\
	=\frac{1}{z^{\NR}_{\beta,\M}}\int\exp\bigg[\beta\M\bigg(
	\NR
	\frac{\Qhat_{t} Q_{t} - \chihat_{t} \chi_{t}}{2}
	- \NR \mhat_{t} m_{t} 
	- \NR^{2}\beta\frac{\chihat_{t} q_{t}}{2}
	\IEEEnonumber\\
	\qquad \qquad  \qquad \qquad \qquad + \frac{1}{\beta}
	\log \mathcal{V}_{\beta}(\hat{\vm{Q}}_{t};\, \NR) \bigg)\bigg] \dx \hat{\vm{Q}}_{t},
	\IEEEeqnarraynumspace
	\label{eq:Tortho-p_bM}
\end{IEEEeqnarray}
where $\dx \hat{\vm{Q}}$ is a short-hand for $\dx\chihat_{t}\dx \Qhat_{t}\dx\mhat_{t}$.
It is important to recognize that we have now managed
to write the components of the free energy in a 
functional form of $\NR$ where the limit 
$\NR\to 0$ can be taken, at least in principle.
Applying the saddle-point method 
to integrate w.r.t.\ $\hat{\vm{Q}}$ and $\vm{Q}$ as
$\beta,\M\to\infty$ and changing the order of extremization and partial 
derivation, we get
\begin{figure*}
	\begin{IEEEeqnarray}{rCl}
		f = \frac{1}{T} \extr_{\vm{Q},\hat{\vm{Q}}} \bigg\{
		\sum_{t=1}^{T} \bigg(\mhat_{t} m_{t}  &-&  \frac{ \Qhat_{t} Q_{t}}{2}
		+\frac{\chihat_{t} \chi_{t}}{2} + (1-\rho_{t})
		\int \Dx z_{t} \phi\big(z_{t}\sqrt{\chihat_{t}};\, \Qhat_{t}\big)
		+\rho_{t} \int \Dx z_{t} \phi\big(z_{t}\sqrt{\chihat_{t}+\mhat^{2}_{t}};\, 
		\Qhat_{t}\big)\bigg)
		\IEEEeqnarraynumspace\IEEEnonumber\\
		&-& \lim_{\beta\to \infty} \lim_{\NR\to 0} 
		\frac{\partial}{\partial \NR} \frac{1}{\beta}
		G_{\beta,\lambda}(\NR\sigma^{2}, \{\NR (\rho_{t} - 2 m_{t} + q_{t}) + Q_{t} - q_{t}\}) 
		\bigg\}
		\label{eq:freeE_3}
	\end{IEEEeqnarray}
	\hrulefill
\end{figure*}
\eqref{eq:freeE_3} at the top of the next page.
Here we used the fact that
$\beta^{-1} G_{\beta,\lambda}(0,\{Q_{t} - q_{t}\}) \to 0$
as $\beta\to\infty$ for $\chi,\Lambda\in \R$ and
\begin{IEEEeqnarray}{l}
	- \frac{1}{\beta}\lim_{\NR\to 0} \frac{\partial}{\partial \NR}
	\log \mathcal{V}_{\beta}(\hat{\vm{Q}}_{t};\, \NR) \IEEEnonumber\\
	\quad =  (1-\rho_{t}) \int \Dx z_{t} \phi\big(z_{t}\sqrt{\chihat_{t}};\, \Qhat_{t}\big)
	\IEEEnonumber\\
	\qquad +\rho_{t} \int \Dx z_{t} \phi\big(z_{t}\sqrt{\chihat_{t}+\mhat^{2}_{t}};\, \Qhat_{t}\big).
	\IEEEeqnarraynumspace
\end{IEEEeqnarray}
We also used above the fact that $r_{t} \to \rho_{t}$ for large $\M$
and that the term $\frac{1}{z^{\NR}_{\beta,\M}}$ in \eqref{eq:PQmeas} is
irrelevant for the analysis because
\begin{equation}
\frac{1}{\M} \lim_{\NR\to 0} 
\frac{\partial}{\partial \NR}
\ln z^{\NR}_{\beta,\M}
\xrightarrow{\M\to\infty} 0.
\end{equation}
By the chain rule 
\begin{IEEEeqnarray}{l}
	\frac{1}{\beta}
	\lim_{\NR\to 0} 
	\frac{\partial}{\partial \NR}
	G_{\beta,\lambda}(\NR\sigma^{2},\{\nu_{t}(\NR)\}_{t=1}^{T}) 
	\IEEEnonumber\\
	= \frac{\sigma^{2}}{\beta} \lim_{\NR\to 0} 
	\frac{\partial G_{\beta,\lambda}(\NR\sigma^{2},\{\nu_{t}(\NR)\}_{t=1}^{T})}
		{\partial (\NR \sigma^{2})}
	\IEEEnonumber\\
	\quad + \frac{1}{\beta} \lim_{\NR\to 0} 
	\sum_{t=1}^{T} \bigg(\frac{\partial \nu_{t}(n)}{\partial \NR} \bigg)
	\frac{\partial G_{\beta,\lambda}(\NR\sigma^{2},\{\nu_{t}(\NR)\}_{t=1}^{T})}{\partial \nu_{t}(n)}, 
	\IEEEeqnarraynumspace
	\label{eq:chainrule}    
\end{IEEEeqnarray}
so that by plugging $\nu_{t}(\NR) = \NR(\rho_{t} - 2 m_{t} + q_{t}) + 
Q_{t}-q_{t}$ to \eqref{eq:chainrule} we have
\begin{IEEEeqnarray}{l}
	\frac{1}{\beta}
	\lim_{\NR\to 0} 
	\frac{\partial}{\partial \NR}
	G_{\beta,\lambda}(\NR\sigma^{2},\{\nu_{t}(\NR)\}_{t=1}^{T})
	\IEEEnonumber\\
	= -\frac{\sigma^{2}}{2} \bigg(\lambda+\sum_{t=1}^{T} 
	\frac{1}{\Lambda^{*}_{t}} \bigg)^{-1} +
	\sum_{t=1}^{T}\frac{\rho_{t} - 2 m_{t} + q_{t}}{2}
	\bigg(\Lambda_{t}^{*} -\frac{1}{\chi_{t}} \bigg), 
	\IEEEnonumber\\
	\label{eq:Dn_of_G}	
\end{IEEEeqnarray}
where $\{\Lambda_{t}^{*}\}$ denotes the
solution to the extremization problem in \eqref{eq:Ffunc_final_lemma},
given $\sigma^{2} = 0$.

To solve the last integrals in \eqref{eq:freeE_3}, let us denote
\begin{equation}
r(h) = \sqrt{\frac{h}{2\pi}} \e^{-\frac{1}{2 h}} -
(1+h)\Qfunc\bigg(\frac{1}{\sqrt{h}}\bigg).
\label{eq:rfunc}
\end{equation}
With some calculus one may verify that for $h > 0$
\begin{equation}
\label{eq:int_over_phi}
\int \phi\big(z_{t}\sqrt{h};\, \Qhat_{t}\big) \Dx z_{t}
= \frac{r (h)}{\Qhat_{t}}, 
\end{equation}
which implies that combining all of the above, the 
free energy has the form
\begin{figure*}
	\begin{IEEEeqnarray}{rCl}
		f = \frac{1}{T} 
		\extr_{\{m_{t}, Q_{t}, \chi_{t}, \mhat_{t}, \Qhat_{t}, \chihat_{t}\}}
		\bigg\{\sum_{t=1}^{T} \bigg[
		\mhat_{t} m_{t} &-& \frac{ \Qhat_{t} Q_{t}}{2}
		+\frac{\chihat_{t} \chi_{t}}{2} + \frac{1-\rho_{t}}{\Qhat_{t}} r(\chihat_{t})
		+ \frac{\rho_{t}}{\Qhat_{t}}r(\chihat_{t}+\mhat^{2}_{t})\bigg]
		\IEEEnonumber\\
		&+&
		\frac{\sigma^{2}}{2}
		\bigg( \lambda+\sum_{t=1}^{T} \frac{1}{\Lambda_{t}^{*}}\bigg)^{-1} 
		+ \sum_{t=1}^{T} \frac{\rho_{t} - 2 m_{t} + Q_{t}}{2}
		\bigg(\frac{1}{\chi_{t}} - \Lambda_{t}^{*} \bigg) \bigg\}
		\IEEEeqnarraynumspace
		\label{eq:freeE_4}
	\end{IEEEeqnarray}
	\hrulefill
\end{figure*}
\eqref{eq:freeE_4} given at the top of the next page.
To obtain this result, we used the fact that denoting
\begin{equation}
	\label{eq:Rt}
	R_{t}(\lambda, \{\Lambda_{t}\}) =
	\frac{1}{\Lambda_{t}}\bigg(\lambda+\sum_{s=1}^{T} \frac{1}{\Lambda_{s}}\bigg)^{-1},
\end{equation}
the extremization in \eqref{eq:Ffunc_final_lemma} implies
the condition
\begin{IEEEeqnarray}{l}
	\Lambda^{*}_{t} - \frac{1}{\chi_{t}}
	= - \frac{R_{t}(\lambda, \{\Lambda^{*}_{t}\})}{\chi_{t}},
	\label{eq:nu_condition_lemma}
\end{IEEEeqnarray}
between the variables $\chi_{t}$ and $\{\Lambda_{t}\}$.
Furthermore, in order to have a meaningful solution to \eqref{eq:freeE_4} and 
\eqref{eq:nu_condition_lemma} for $\sigma>0$, we also need to have 
$\chi_{t} = \beta(Q_{t}-q_{t})$ positive and finite%
\footnote{The case of 
	$\chi_t \to 0$ is in fact relevant for 
	the noise-free scenario $\sigma=0$ and corresponds to the perfect 
	recovery condition $\rho_t = m_t = Q_t 
	\implies \mse_t = \rho_t - 2 m_t + Q_t = 0$, which automatically 
	satisfies $Q_t = q_t$ as well.  Furthermore, for this 
	scenario $\Qhat = \mhat \to \infty$ as $\beta \to \infty$, 
	while in the noisy case they are always positive and finite
	parameters.}
for all values of
$\beta>0$, which means $\chi^{-1}_{t}(\rho_{t} - 2 m_{t} + q_{t})
\to \chi^{-1}_{t}(\rho_{t} - 2 m_{t} + Q_{t})$
as $\beta\to\infty$.

\subsection{Saddle-Point Conditions}

Using the short-hand notation 
\begin{equation}
	R_{t} = R_{t}(\lambda, \{\Lambda^{*}_{t}\}),
\end{equation}
the partial derivatives w.r.t. $\{m_{t},Q_{t}\}$ in 
\eqref{eq:freeE_4} provide the saddle-point 
conditions 
\begin{equation}
	\Qhat_{t} = \mhat_{t} = \frac{R_{t}}{\chi_{t}} = \frac{1}{\chi_{t}} - \Lambda^{*}_{t}.
	\label{eq:Qt_mt}
\end{equation}
By the fact that
\begin{equation}
\frac{\partial}{\partial x} r (h)
= - \bigg(\frac{\partial h}{\partial x}\bigg) \Qfunc\bigg(\frac{1}{\sqrt{h}}\bigg),
\end{equation}
we may assess the partial derivatives w.r.t. the variables 
$\{\Qhat_{t}, \mhat_{t}, \chihat_{t}\}$ as well to obtain
\begin{IEEEeqnarray}{rCl}
	\label{eq:sp_m}
	m_{t} &=& 2 \rho_{t}
	\Qfunc\bigg(\frac{1}{\sqrt{\chihat_{t}+\mhat^{2}_{t}}}\bigg),
	\\
	\label{eq:sp_Q}
	Q_{t} &=& -\frac{2(1-\rho_{t})}{\mhat_{t}^{2}} r(\chihat_{t})
	- \frac{2\rho_{t}}{\mhat_{t}^{2}} r(\chihat_{t}+\mhat^{2}_{t}),
	\\
	\label{eq:sp_chi}
	\chi_{t} &=& \frac{2(1-\rho_{t})}{\mhat_{t}}
	\Qfunc\bigg(\frac{1}{\sqrt{\chihat_{t}}}\bigg)
	+ \frac{2\rho_{t}}{\mhat_{t}} \Qfunc\bigg(\frac{1}{\sqrt{\chihat_{t}+\mhat^{2}_{t}}}\bigg),
	\IEEEeqnarraynumspace
\end{IEEEeqnarray}
where we used the identity $\Qhat_{t} = \mhat_{t}$ to simplify
the results.
The MSE of the reconstruction for $\vm{x}_{t}$ thus becomes
\begin{IEEEeqnarray}{rCl}
	\mse_{t} = 
	\rho_{t} &-& 2 m_{t} + Q_{t} \IEEEnonumber\\
	= \rho_{t} &-& 4 \rho_{t}
	\Qfunc\bigg(\frac{1}{\sqrt{\chihat_{t}+\mhat^{2}_{t}}}\bigg)
	\IEEEnonumber\\
	&-&\frac{2(1-\rho_{t})}{\mhat_{t}^{2}}
	r(\chihat_{t}) - \frac{2\rho_{t}}{\mhat_{t}^{2}} r(\chihat_{t}+\mhat^{2}_{t}). 
	\IEEEeqnarraynumspace
\end{IEEEeqnarray}
Finally, recalling that $\Lambda^{*}_{t}$ is a function of 
$\{\chi_{t}\}$, we obtain
from the partial derivative of $\chi_{t}$
\begin{equation}
	\chihat_{t} = \frac{\mse_{t}}{\chi_{t}^{2}} + \sum_{s=1}^{T} 
	(\mse_{s} - \sigma^{2} R^{2}_{s}) \Delta_{s,t}, 
\end{equation}
where we denoted $\Delta_{s,t} = \frac{\partial\Lambda_{s}}{\partial \chi_{t}}$ for 
the partial derivative of $\Lambda_{s}$ w.r.t.\ $\chi_{t}$.

To solve the equation for $\chihat_{t}$, 
we need an expression for $\Delta_{s,t}$.
We do this via the inverse function theorem that relates 
the Jacobian matrices as
\begin{equation}
\vm{\Delta} = \frac{\partial(\Lambda_{1},\ldots,\Lambda_{T})}
{\partial (\chi_{1},\ldots,\chi_{T})}
= \bigg[\frac{\partial (\chi_{1},\ldots,\chi_{T})}{\partial(\Lambda_{1},\ldots,\Lambda_{T})}\bigg]^{-1}.
\label{eq:jacob_transformation}
\end{equation}
Here the $(i,j)$th element of the Jacobian 
$\frac{\partial (\chi_{1},\ldots,\chi_{T})}{\partial(\Lambda_{1},\ldots,\Lambda_{T})}$ is given by
\begin{IEEEeqnarray}{rCl}
	\frac{\partial\chi_{i}}{\partial \Lambda_{j}}
	&=& \frac{\partial}{\partial \Lambda_{j}}\frac{(1-R_{i})}{\Lambda_{i}}
	\IEEEnonumber\IEEEeqnarraynumspace\\
	&=&
	-\frac{(1-R_{i})}{\Lambda_{i}^{2}} \delta_{ij}
	- \frac{1}{\Lambda_{i}} \frac{\partial R_{i}}{\partial \Lambda_{j}} 
	\IEEEnonumber\IEEEeqnarraynumspace\\
	&=& -\frac{(1-2 R_{i})}{\Lambda_{i}^{2}} \delta_{ij}
	- \frac{R_{i} R_{j}}{\Lambda_{i}\Lambda_{j}}.
\end{IEEEeqnarray}
In  other words,
denoting $\vm{b} = [R_{1} / \Lambda_{1} \;\, \cdots \;\, R_{T} / \Lambda_{T}]^{\trans}$
and defining
$\vm{C}$ to be diagonal matrix whose
$(t,t)$th entry is given by $(1-2 R_{t})\Lambda_{t}^{-2}$,
we obtain by \eqref{eq:jacob_transformation} and 
the matrix inversion lemma
the desired Jacobian as
\begin{IEEEeqnarray}{rCl}
	\vm{\Delta} &=& -(\vm{C}+\vm{b}\vm{b}^{\trans})^{-1}
	= -\vm{C}^{-1} + \frac{(\vm{C}^{-1}\vm{b})(\vm{C}^{-1}\vm{b})^{\trans}}
	{1+\vm{b}^{\trans}\vm{C}^{-1}\vm{b}},
	\IEEEeqnarraynumspace
\end{IEEEeqnarray}
which means that
\begin{IEEEeqnarray}{rCl}
	\Delta_{s,t} &=& 
	\frac{R_{s}R_{t}\Lambda_{s}\Lambda_{t}}{(1-2 R_{s})(1-2 R_{t})}
	\bigg(1+\sum_{k=1}^{T}\frac{R_{k}^{2}}{1-2 R_{k}}\bigg)^{-1}
	\IEEEnonumber\\
	&& -\frac{\Lambda_{t}^{2}}{1-2 R_{t}}
	\delta_{st}. \IEEEeqnarraynumspace
\end{IEEEeqnarray}
Combining all the results completes the derivation.


\section{Replica Analysis of Rotationally Invariant Setup}
\label{sec:eigen_replica}

The derivation in this Appendix provides an end result 
that is essentially the same as the 
HCIZ-formula based \cite{Harish-Chandra-1957, Itzykson-Zuber-1980}
approach used in Section~IV and Appendix~B in \cite{Tulino-etal-cs2013}.  
In our case the difference is that the source does not need to have
IID elements, but can have a block structure.  Furthermore,
our analytical approach is slightly different to the one in 
\cite{Tulino-etal-cs2013} since we do not seek to find first a decoupling 
result for finite $\beta$ and then use hardening arguments
as in \cite{Rangan-Fletcher-Goyal-2012} to obtain the final result
when  $\beta\to\infty$.  Both end results are equivalent
as shown in Section~\ref{sec:equivalence}.

Recall the rotationally invariant setup as given in 
Definition~\ref{defn:matrix_ensembles}.
Let $\PM(\vm{x}^{0};\, \{\rho_{t}\})$ be the 
distribution of the source vector 
$\vm{x}^{0}\in\R^{\N}$ and assume that each of the sub-vectors 
$\vm{x}^{0}_{t}$ has $\Mx$ elements drawn 
independently according to \eqref{eq:true_sourcesym_pdf_k}.
Clearly we have to have $\N = \Mx T$ but it is not necessary 
to have $\M = \Mx$ as in the case of $T$-orthogonal setup.
Define
\begin{IEEEeqnarray}{l}
	\label{eq:AVGpartFu_3}
	\Xi_{\beta,\N}(\NR) =
	\E_{\vm{w},\vm{A}}
	\int \PM(\vm{x}^{0};\, \{\rho_{\tx}\}) 
	\exp\bigg(\!\!-\beta\sum_{a=1}^{\NR}\Freg(\vm{x}^{a})\bigg)
	\IEEEnonumber\\
	\qquad \times \exp 
	\bigg( - \frac{\beta}{2\lambda} \sum_{a=1}^{\NR} 
		\| \sigma\vm{w} - \vm{A} \Delta\vm{x}^{a}\|^{2} \bigg)
	\prod_{a = 0}^{\NR} \dx \vm{x}^{a}, \IEEEeqnarraynumspace
\end{IEEEeqnarray}
where 
$\Delta\vm{x}^{a} = \vm{x}^{0}-\vm{x}^{a} \in \R^{\N}$ for 
$a=1,\ldots,n$, so that the counterpart of
\eqref{eq:freeE_replica_appendix_1} reads
\begin{equation}
	\label{eq:freeE_replica_appendix_3}
	f = - \lim_{\NR\to 0^{+}} \frac{\partial}{\partial \NR}
	\lim_{\beta,\N\to\infty} \frac{1}{\beta \N} \ln \Xi_{\beta,\N} (\NR).
\end{equation}
The goal is then to assess the normalized 
free energy \eqref{eq:freeE_replica_appendix_3} by following the 
same steps as given in Appendix~\ref{sec:Tortho_replica}.

Let us construct matrices $\vm{S}_{t} \in \R^{\NR\times \NR}$ and
$\vm{Q}_{t}$ for all $t=1,\ldots,T$ with elements 
as given in \eqref{eq:Smtx_1}~and~\eqref{eq:Smtx_2}.
Also define the ``empirical mean'' matrices
$\vm{S} = T^{-1}\sum_{t}\vm{S}_{t},$
$\vm{Q} = T^{-1}\sum_{t}\vm{Q}_{t},$ 
that have the respective elements $S^{[a,b]}$ and $Q^{[a,b]}$
and invoke the RS assumption
\eqref{eq:RS_ass_first}--\eqref{eq:RS_ass_last}.
We then make the transformation $\{\Delta \vm{x}^{a}_{t}\}
\to \{\Delta \tilde{\vm{x}}^{a}_{t}\}$ as with the $T$-orthogonal setup
so that the empirical correlations of 
$\{\Delta \tilde{\vm{x}}^{a}_{t}\}$
satisfy \eqref{eq:StildeAB}. 
Note that this means that given $\{\vm{Q}_{t}\}$, 
the transformed vectors 
$\Delta\tilde{\vm{x}}^{a} =\big[ 
(\Delta\tilde{\vm{x}}^{a}_{1})^{\trans} \;\, \cdots \; \,
(\Delta\tilde{\vm{x}}^{a}_{T})^{\trans}
\big]^{\trans}$
satisfy 
\begin{equation}
	\|\Delta\tilde{\vm{x}}^{a}\|^{2} = \Mx  \sum_{t=1}^{T}
	\tilde{S}^{[a,b]}_{t} = N \tilde{S}^{[a,b]},
\end{equation}
where 
$\tilde{S}^{[a,b]}  = T^{-1}\sum_{t}\tilde{S}^{[a,b]}_{t}$.
Combining the above provides the counterpart of \eqref{eq:Iu_2} as
\begin{IEEEeqnarray}{l}
	\Xi_{\beta,\N}(\NR) \IEEEnonumber\\ 
	= \E_{\vm{w},\{\vm{O}_{t}\}}
	\int \bigg(\prod_{a = 0}^{\NR} \dx \vm{x}^{a}\bigg)
	\PM(\vm{x}^{0};\, \{\rho_{\tx}\}) 
	\exp\bigg(\!\!-\beta\sum_{a=1}^{\NR}\Freg(\vm{x}^{a})\bigg)
	\IEEEnonumber\\ 
	\times
	\exp \bigg(-\frac{\beta}{2\lambda} \|\sqrt{\NR\sigma^{2}} \vm{w} 
	- \vm{A} \Delta\tilde{\vm{x}}^{1}\|^{2} 
	-\frac{\beta}{2\lambda} \sum_{a=2}^{\NR}
	\| \vm{A} \Delta\tilde{\vm{x}}^{a}\|^{2} 
	\bigg). \IEEEnonumber\\ 
	\label{eq:Iu_4} 
\end{IEEEeqnarray}
We then need the following small result to proceed.

\begin{lemma}
	\label{lemma:Ffunc_eigs}
	Consider the case where $\vm{A} \in \R^{\M \times \N}$ is sampled from 
	the rotationally invariant setup given in Definition~\ref{defn:matrix_ensembles}.
	Let $\{\vm{u}_{t}\}_{\tx=1}^{\Tx}$
	be a fixed set of length-$\Mx$ vectors satisfying
	$\|\vm{u}_{t}\|^{2} = \Mx \nu_{t}$ 
	for some given non-negative values $\{\nu_{t}\}$ 
	and $\N = \Tx \Mx$.
	Denote $\vm{u}\in\R^{\N}$ for the vector obtained by 
	stacking $\{\vm{u}_{t}\}$ and define
	\begin{IEEEeqnarray}{rCl}
		\e^{\N H_{\beta,\lambda}(\sigma^{2},\{\nu_{t}\})}
		&=&
		\E_{\vm{w},\vm{A}}
		\e^{-\frac{\beta}{2\lambda} \| \sigma\vm{w} - \vm{A} \vm{u}\|^{2}},
		\IEEEeqnarraynumspace
	\end{IEEEeqnarray}
	where $\vm{w}$ is a standard Gaussian random vector.  Then, for 
	large $\N$
	\begin{IEEEeqnarray}{l}
		\!\!\!\!\!\!H_{\beta,\lambda}(\sigma^{2},\nu) =
		H_{\beta,\lambda}(\sigma^{2},\{\nu_{t}\}) \IEEEnonumber\\
		\quad= \frac{1}{2}
		\extr_{\Lambda}\bigg\{\Lambda (\beta\nu) 
		- (1-\alpha) \ln \Lambda  \IEEEnonumber\\
		\qquad\qquad\qquad - \alpha \! \int \! \ln (\Lambda\beta\sigma^{2} +\Lambda\lambda+ x) 
		\dx F_{\vm{A}\vm{A}^{\trans}}(x)
		\bigg\}  \IEEEnonumber\\
		\qquad\quad - \frac{1 + \ln (\beta \nu)- \alpha\ln \lambda}{2},
		\IEEEeqnarraynumspace
		\label{eq:hatH}
	\end{IEEEeqnarray}
	where $\nu = T^{-1}\sum_{t=1}^{T} \nu_{t}$ and
	we omitted terms of the order $O(1/N)$.  
\end{lemma}

\begin{proof}
	Proof is given in Appendix~\ref{sec:Ffunc_rotational}.
\end{proof}

Notice that the $H$-function in \eqref{eq:hatH} depends on the parameters
$\{\nu_{t}\}$ only through the ``empirical mean'' 
$\nu = T^{-1}\sum_{t=1}^{T} \nu_{t}$. This will translate later 
to the fact that the performance of rotationally invariant setup 
depends on the sparsities $\{\rho_{t}\}$ only through 
$\rho = \Tx^{-1}\sum_{\tx} \rho_{\tx}$.
With the above in mind, we may obtain the 
probability weight $p_{\beta,\N}(\vm{Q};\, \NR)$ of $\vm{Q}$ 
by using \eqref{eq:identity_trivial} with suitable variable substitutions.
Applying then Lemma~\ref{lemma:Ffunc_eigs} to 
\eqref{eq:Iu_4} provides
\begin{IEEEeqnarray}{l}
	\frac{1}{\N} \ln \Xi_{\beta,\N}(\NR) 
	\IEEEnonumber\\
	\quad =
	\frac{1}{\N} \ln \int  p_{\beta,\N}(\vm{Q};\, \NR)
	\e^{\N H_{\beta,\lambda}(\NR\sigma^{2},\NR (r - 2 m + q) + Q - q) }
	\IEEEeqnarraynumspace\IEEEnonumber\\
	\qquad\qquad\qquad
	\times\e^{\N (\NR-1) H_{\beta,\lambda}(0,Q - q) } \dx \vm{Q}, \IEEEeqnarraynumspace
	\label{eq:I_Q_sigma2}
\end{IEEEeqnarray}
where 
$r = T^{-1} \sum_{t} r_{t},  
m = T^{-1} \sum_{t} m_{t},  
Q = T^{-1} \sum_{t} Q_{t},$
and $q = T^{-1} \sum_{t} q_{t}$ are the ``averaged'' versions of 
the RS variables $\{r_{t},m_{t},Q_{t},q_{t}\}$.
The probability weight of $\vm{Q}$ reads 
\begin{IEEEeqnarray}{l}
	p_{\beta,\N}(\vm{Q};\, \NR)
	=\N^{\NR(\NR+1)/2} 
	\int \bigg(\prod_{0 \leq a\leq b \leq \NR}  
	\frac{\dx \tilde{Q}^{[a,b]}}{2 \pi \im}\bigg) \IEEEnonumber\\
	\qquad \times
	\exp \bigg(\N \sum_{0 \leq a\leq b \leq \NR}  \tilde{Q}^{[a,b]} Q^{[a,b]}\bigg)
	\mathcal{V}_{\beta,\N}(\tilde{\vm{Q}};\, \NR),
	\IEEEeqnarraynumspace
	\label{eq:PQmeas3}
\end{IEEEeqnarray}
where $\tilde{\vm{Q}}$ is a $(\NR+1)\times(\NR+1)$
transform domain matrix
whose elements are $\{\tilde{Q}^{[a,b]}\}$ and
\begin{IEEEeqnarray}{rCl}
	\!\!\mathcal{V}_{\beta,\N}(\tilde{\vm{Q}};\, \NR) &=& 
	\int p(\vm{x}^{0}) \dx \vm{x}^{0}
	\prod_{a=1}^{\NR} 
	\Big(\e^{-\beta \|\vm{x}^{a}\|_{1}}\dx \vm{x}^{a} \Big) \IEEEnonumber\\
	&& \times\exp \bigg(-\sum_{0 \leq a\leq b \leq \NR}
	\tilde{Q}^{[a,b]} \vm{x}^{a}\cdot\,\vm{x}^{b} \bigg).
	\IEEEeqnarraynumspace
	\label{eq:V_defn3}
\end{IEEEeqnarray}
We then get directly using the arguments from Appendix~\ref{sec:Tortho_replica}
that \eqref{eq:V_defn3} becomes in the limit $\N \to \infty$ 
\begin{IEEEeqnarray}{rCl}
	\mathcal{V}_{\beta}(\hat{\vm{Q}};\, \NR) 
	&=& (1-\rho) \int \exp \Big[-\beta \NR \phi\big(z\sqrt{\chihat};\, \Qhat\big)\Big]\Dx z 
	\IEEEnonumber\\ &&
	+ \rho \int \exp \Big[-\beta \NR \phi\big(z\sqrt{\chihat+\mhat^{2}};\, \Qhat\big)\Big] \Dx z,
	\IEEEeqnarraynumspace
\end{IEEEeqnarray}
where $\rho = \Tx^{-1}\sum_{\tx} \rho_{\tx}$ is the expected 
sparsity of the entire source vector $\vm{x}$.  Therefore, the 
details of how the non-zero elements are distributed on different
sub-blocks $\{\vm{x}_{t}\}$ is irrelevant for the rotationally
invariant case.

Combining everything above and denoting 
$\chi = T^{-1} \sum_{t}\beta(Q_{t}-q_{t})$ 
implies that the free energy for the rotationally 
invariant case reads
\begin{IEEEeqnarray}{rCl}
	f &=& \extr_{\{m, Q, \chi, \mhat, \Qhat, \chihat\}}
	\bigg\{\mhat m  -  \frac{ \Qhat Q}{2}+\frac{\chihat \chi}{2} \IEEEnonumber\\
	&& +\frac{1}{\Qhat}
	\big[ (1-\rho) r (\chihat) - \rho  r (\chihat+\mhat^{2})\big]
	\IEEEnonumber\\
	&& +\frac{\alpha\sigma^{2}\Lambda^{*}}{2} G_{\vm{A}\vm{A}^{\trans}}(-\lambda \Lambda^{*})
	+\frac{\rho - 2 m + Q}{2} \bigg(\frac{1}{\chi} - \Lambda^{*} \bigg)\bigg\}, \IEEEnonumber\\
	\label{eq:freeE_5}
\end{IEEEeqnarray}
where we used the Stieltjes transform 
of $F_{\vm{A}\vm{A}^{\trans}}(x)$,
\begin{equation}
	G_{\vm{A}\vm{A}^{\trans}}(s) = \int \frac{1}{x - s} \dx F_{\vm{A}\vm{A}^{\trans}}(x),
	\label{eq:stieltjes}
\end{equation}
along with the chain rule 
\begin{IEEEeqnarray}{l}
	\lim_{\beta\to \infty} 
	\frac{1}{\beta}
	\lim_{\NR\to 0} 
	\frac{\partial}{\partial \NR}
	H_{\beta,\lambda}(\NR\sigma^{2},\nu(\NR)) \IEEEnonumber\\
	\; = 
	\lim_{\beta\to \infty} \frac{\sigma^{2}}{\beta} \lim_{\NR\to 0} 
	\frac{\partial 
		H_{\beta,\lambda}(\NR\sigma^{2},\nu(\NR))
	}{\partial (\NR \sigma^{2})}   \IEEEnonumber\\
	\qquad + \lim_{\beta\to \infty} 
	\frac{1}{\beta} \lim_{\NR\to 0} 
	\bigg(\frac{\partial \nu(\NR)}{\partial \NR} \bigg)
	\frac{\partial
		H_{\beta,\lambda}(\NR\sigma^{2},\nu(\NR))
	}{\partial \nu(\NR)}
	\IEEEnonumber\\
	\; = 
	- \frac{\alpha\sigma^{2}}{2}
	\int \frac{\Lambda^{*}}
	{\lambda \Lambda^{*} + x} \dx F_{\vm{A}\vm{A}^{\trans}}(x)
	+ \frac{\rho - 2 m + Q}{2}
	\bigg( \Lambda^{*} - \frac{1}{\chi} \bigg) 
	\IEEEnonumber\\
	\; = 
	- \frac{\alpha\sigma^{2}\Lambda^{*}}{2}
	G_{\vm{A}\vm{A}^{\trans}}(-\lambda \Lambda^{*})
	+ \frac{\rho - 2 m + Q}{2}
	\bigg( \Lambda^{*} - \frac{1}{\chi} \bigg),
	\IEEEnonumber\\
	\label{eq:chainrule2}
\end{IEEEeqnarray}
where $\nu(\NR) = \NR (r - 2 m + q) + Q - q$.
Here $\Lambda^{*}$ is the solution to the extremization in 
\eqref{eq:hatH}, given $\sigma^{2}=0$, and satisfies the condition
\begin{equation}
	\Lambda^{*}-\frac{1}{\chi} 
	= - \frac{\alpha}{\chi} 
	\big[1- (\lambda\Lambda^{*}) \cdot G_{\vm{A}\vm{A}^{\trans}}(-\lambda \Lambda^{*}) \big]
	= - \frac{\hat{R}(\Lambda^{*})}{\chi},
	\label{eq:Fextr_eigs2}
\end{equation}
where
\begin{equation}
	\hat{R}(\Lambda^{*}) = \alpha
	\big[1- (\lambda\Lambda^{*}) \cdot G_{\vm{A}\vm{A}^{\trans}}(-\lambda \Lambda^{*}) \big].
\end{equation}

Finally, we need to resolve the saddle point conditions
in \eqref{eq:freeE_5}.  The partial derivatives w.r.t. $\{m,Q\}$ provide 
\begin{equation}
\Qhat = \mhat = \frac{1}{\chi} - \Lambda^{*} = \frac{\hat{R}(\Lambda^{*} )}{\chi},
\label{eq:Q_m}
\end{equation}
while the partial derivatives w.r.t.\ $\{\Qhat, \mhat, \chihat\}$
are of the same format as in \eqref{eq:sp_m}--\eqref{eq:sp_chi} but 
without indexes $t$.
Finally, recalling that $\Lambda^{*}$ depends on $\chi$
\begin{IEEEeqnarray}{rCl}
	\frac{\partial}{\partial \chi} G_{\vm{A}\vm{A}^{\trans}}(-\lambda \Lambda^{*})
	&=&
	-\lambda\bigg(\frac{\partial \Lambda}{\partial \chi}\bigg)
	G'_{\vm{A}\vm{A}^{\trans}}(-\lambda \Lambda^{*}),
	\IEEEeqnarraynumspace
\end{IEEEeqnarray}
where
$G_{\vm{A}\vm{A}^{\trans}}'$ denotes the derivative of 
$G_{\vm{A}\vm{A}^{\trans}}$ w.r.t.\ the argument,
gives
\begin{IEEEeqnarray}{l}
	\chihat  = \mse
	\bigg(\frac{1}{\chi^{2}} 
	+ \frac{\partial\Lambda^{*}}{\partial \chi} \bigg)  \IEEEnonumber\\
	\;   - \alpha \sigma^{2}
	\big[
	G_{\vm{A}\vm{A}^{\trans}}(-\lambda \Lambda^{*})  - (\lambda \Lambda^{*}) \cdot 
	G'_{\vm{A}\vm{A}^{\trans}}(-\lambda \Lambda^{*})
	\big] \frac{\partial\Lambda^{*}}{\partial \chi},
	\IEEEeqnarraynumspace
\end{IEEEeqnarray}
in which
\begin{equation}
\frac{\partial\Lambda^{*}}{\partial \chi}
= -\bigg[\frac{1-\alpha}{(\Lambda^{*})^{2}} + (\alpha\lambda^{2})
\cdot G'_{\vm{A}\vm{A}^{\trans}}(-\lambda \Lambda^{*})\bigg]^{-1}.
\end{equation}
To obtain the last formula we used the fact that
\begin{IEEEeqnarray}{rCl}
	\frac{\partial \chi}{\partial \Lambda^{*}} 
	&=& -\frac{1}{(\Lambda^{*})^{2}}(1 - \hat{R}) 
	-\frac{1}{\Lambda^{*}} \frac{\partial \hat{R}}{\partial \Lambda^{*}} 
	\IEEEnonumber\\
	&=&  - \bigg[\frac{1-\alpha}{(\Lambda^{*})^{2}}
	+ (\alpha\lambda^{2}) \cdot G'_{\vm{A}\vm{A}^{\trans}}(-\lambda \Lambda^{*})\bigg].
	\IEEEeqnarraynumspace
\end{IEEEeqnarray}

\begin{remark}
	\label{rem:Tortho_vs_rowortho}
	Consider the row-orthogonal setup where
	\begin{IEEEeqnarray}{rCl}
		\label{eq:GR_rowortho_app}
		G_{\vm{A}\vm{A}^{\trans}}(s) &=& \frac{1}{\alpha^{-1}-s}. \IEEEeqnarraynumspace
	\end{IEEEeqnarray}
	For this case, the extremization in \eqref{eq:hatH}
	can also be written in the form
	\begin{IEEEeqnarray}{l}
		\Lambda - \frac{1}{\beta \nuAvg} = 
		-\frac{1}{\beta \nuAvg}\bigg(\frac{1}{\alpha^{-1}+ \Lambda (\lambda + \beta \sigma^{2})}\bigg)
		\IEEEnonumber\\
		\qquad \qquad \xrightarrow{\sigma=0}
		-\frac{1}{\beta \nuAvg}\bigg(\frac{1}{\alpha^{-1}+ \Lambda \lambda}\bigg).
		\label{eq:Lambda_nu_relation2}	
	\end{IEEEeqnarray}
	We may then plug \eqref{eq:GR_rowortho_app} and \eqref{eq:Lambda_nu_relation2} to
	\eqref{eq:freeE_5} and compare the end result with 
	\eqref{eq:freeE_4}--\eqref{eq:nu_condition_lemma}.  It is clear that 
	the two free energies are exactly the same 
	if we set $\alpha = 1/T$ and
	$\rho = \rho_{t}$ so that
	$\nuAvg = \nuAvg_{t}$ and $\Lambda = \Lambda_{t}$ for all $t = 1,\ldots,T$.
	Therefore, also the saddle point solutions of row-orthogonal and $T$-orthogonal
	setups match for this special case and the MSE is the same.
\end{remark}


\section{Useful Matrix Integrals}
\label{sec:matrix_integrals}

\subsection{$\T$-Orthogonal Setup}
\label{sec:Ffunc_Tortho}

Let $\{\vm{O}_{t}\}_{t=1}^{T}$ be a set of 
independent $\M \times \M$ Haar matrices and
$\{\Delta\vm{x}_{t}\}_{t=1}^{T}$ 
a set of (fixed) length-$\M$ vectors that satisfy 
$\|\Delta\vm{x}_{t}\|^{2} = \M \nu_{t}$ 
for some given non-negative values $\{\nu_{t}\}$.
Given $\{\Delta\vm{x}_{t}\}$ and $\{\nu_{t}\}$, the vector
$\vm{u}_{t} = \vm{O}_{t}\Delta\vm{x}_{t}$
is uniformly distributed on a surface 
of a sphere that has a fixed radius $\sqrt{M \nu_{t}}$ for
each $t=1,\ldots,T$. Thus, the joint PDF
of $\{\vm{u}_{t}\}$ reads
\begin{IEEEeqnarray}{l}
	p^{(\M)}(\{\vm{u}_{t}\};\,\{\nu_{t}\})  
	\IEEEnonumber\\
	\label{eq:jointPDF_ut_1}
	\quad = \frac{1}{Z(\{\nu_{t}\})} \prod_{t=1}^{T}\delta(\|\vm{u}_{t}\|^{2} - \M \nu_{t})  
	\IEEEeqnarraynumspace\\
	\quad =\frac{(4 \pi \im)^{-T}}{Z(\{\nu_{t}\})}
	\int \prod_{t=1}^{T} \Big(
	\e^{-\frac{\Lambda_{t}}{2}(\|\vm{u}_{t}\|^{2}-\M \nu_{t})} \dx \Lambda_{t}\Big),
	\IEEEeqnarraynumspace
	\label{eq:jointPDF_ut}
\end{IEEEeqnarray}
where $Z(\{\nu_{t}\})$ is the normalization factor, 
$\{\Lambda_{t}\}$ is a set of complex numbers and we used the identity
\begin{IEEEeqnarray}{rCl}
	\delta(t-a) &=& 
	\frac{1}{4 \pi \im} \int_{c-\im\infty}^{c+\im\infty} \e^{-\frac{\Lambda}{2}(t-a)} \dx \Lambda, 
	\label{eq:dirac_fft}
\end{IEEEeqnarray}
where $a,c,t\in\R, \Lambda\in\C.$
Using the Gaussian integration formula
\begin{equation}
	\label{eq:Gint}
	\frac{1}{(2\pi)^{N/2}}\int \e^{-\frac{1}{2}\vm{z}^{\trans} \vm{M} 
		\vm{z}+ \vm{b}^{\trans} \vm{z}} \dx \vm{z} 
	= \frac{1}{\sqrt{\det(\vm{M})}}\e^{\frac{1}{2}
		\vm{b}^{\trans}\vm{M}^{-1}\vm{b}}, 
\end{equation}
where $\vm{b},\vm{z} \in \R^{N}$ and $\vm{M}$ is symmetric positive definite, 
the normalization factor becomes
\begin{IEEEeqnarray}{rCl}
	Z(\{\nu_{t}\})
	&=&  \frac{1}{(4 \pi \im)^{T}}
	\int \prod_{t=1}^{T}
	\Big(\e^{\frac{\Lambda_{t}}{2}\M \nu_{t}} 
	\e^{-\frac{1}{2} \Lambda_{t}\|\vm{u}_{t}\|^{2}} \dx \vm{u}_{t} 
	\dx \Lambda_{t} \Big) \IEEEnonumber\\
	&=&  \bigg(\frac{(2\pi)^{\M / 2}}{4 \pi \im}\bigg)^{T}
	\int \prod_{t=1}^{T} \Big(\e^{
	\frac{\M}{2}(\Lambda_{t} \nu_{t} -\ln \Lambda_{t} )} \dx\Lambda_{t} \Big).
	\IEEEeqnarraynumspace
	\label{eq:ZnuT}
\end{IEEEeqnarray}
Since the argument of the exponent in \eqref{eq:ZnuT}
is a complex analytic function of $\{\Lambda_{t}\}$
and we are interested in the large-$\M$ asymptotic,
the saddle-point method further simplifies the normalization
factor to the form
\begin{IEEEeqnarray}{rCl}
	\frac{1}{\M} \ln Z (\{\nu_{t}\}) &=&
	\frac{1}{2} \sum_{t=1}^{T} \extr_{\Lambda_{t}}
	\big\{\Lambda_{t} \nu_{t} -\ln \Lambda_{t}\big\} + O(M^{-1})
	\IEEEeqnarraynumspace\IEEEnonumber\\
	&=& \sum_{t=1}^{T}  \frac{1 + \ln \nu_{t}}{2} + O(M^{-1}),
	\label{eq:lnZ_PDF_of_u}
\end{IEEEeqnarray}
where the second equality is obtained by solving the extremization problem.  
Substituting \eqref{eq:lnZ_PDF_of_u}
back to \eqref{eq:jointPDF_ut} provides an expression for 
$p^{(\M)}(\{\vm{u}_{t}\};\,\{\nu_{t}\})$.

Recall the $T$-orthogonal setup given in Definition~\ref{defn:matrix_ensembles}.
Fix the parameters $\M,\beta,\lambda$ and define
\begin{IEEEeqnarray}{l}
	\label{eq:Ffunc_0}
	G^{(\M)}_{\beta,\lambda}(\sigma^{2},\{\nu_{t}\})
	\IEEEnonumber \\
	\quad =
	\frac{1}{\M}
	\ln \E_{\vm{w},\{\vm{O}_{t}\}}
	\e^{-\frac{\beta}{2\lambda} \| \sigma\vm{w} -\sum_{t=1}^{T}\vm{O}_{t} \Delta\vm{x}_{t}\|^{2}}
	\IEEEeqnarraynumspace\IEEEnonumber \\
	\quad =
	\frac{1}{\M}
	\ln \E_{\vm{w}} \int
	p^{(\M)}(\{\vm{u}_{t}\};\,\{\nu_{t}\}) \IEEEnonumber\\ 
	\qquad\qquad\qquad\times 
	\e^{-\frac{1}{2\lambda} \| \sqrt{\beta \sigma^{2}}\vm{w} 
		- \sqrt{\beta} \sum_{t=1}^{T}\vm{u}_{t}\|^{2}}
	\prod_{t=1}^{T}\dx\vm{u}_{t},
	\IEEEeqnarraynumspace 
	\label{eq:Ffunc_1}
\end{IEEEeqnarray}
where $\{\vm{O}_{t}\}$, $\{\Delta\vm{x}_{t}\}$, $\{\vm{u}_{t}\}$ and $\{\nu_{t}\}$
are as before. Applying the Gaussian integration formula \eqref{eq:Gint} from 
right-to-left along with the expressions
\eqref{eq:jointPDF_ut}
and \eqref{eq:lnZ_PDF_of_u} provides 
\begin{IEEEeqnarray}{l}
	G^{(\M)}_{\beta,\lambda}(\sigma^{2},\{\nu_{t}\}) \IEEEnonumber\\
	= \frac{1}{M} \ln \E_{\vm{w}}  \int
	\prod_{t=1}^{T}
	\Big( \dx \Lambda_{t} \e^{\frac{\M}{2}\Lambda_{t} \nu_{t}} \Big)
	\int\! \dx \vm{k} \e^{\tmpfuncA(\vm{k},\vm{w})}  
	\IEEEnonumber\\
	\qquad\qquad\qquad
	\times\int \prod_{t=1}^{T} \Big(
	\e^{-\frac{1}{2}\Lambda_{t}\|\vm{u}_{t}\|^{2}
		- \im \sqrt{\beta}\vm{k}^{\trans} \vm{u}_{t}} 
	\dx\vm{u}_{t} \Big) 
	\IEEEnonumber\\
	\qquad + \frac{1}{2} \ln \frac{\lambda}{2\pi} 
	- \frac{1}{\M} \ln Z (\{\nu_{t}\}),
	\IEEEeqnarraynumspace
	\label{eq:F_after_HS}
\end{IEEEeqnarray}
where $\vm{k} \in \R^{\M}$, the normalization factor is given in
\eqref{eq:lnZ_PDF_of_u} and we denoted
\begin{IEEEeqnarray}{rCl}
	\label{eq:tmpfuncA}
	\tmpfuncA(\vm{k},\vm{w}) &=& -\frac{\lambda}{2} \|\vm{k}\|^{2}
	+ \im \sqrt{\beta\sigma^{2}}\vm{k}^{\trans} \vm{w}.
\end{IEEEeqnarray}
Using next Gaussian integration repeatedly to assess the 
expectations w.r.t.\ $\{\vm{u}_{t}\}$, $\vm{k}$ and $\vm{w}$ yields
\begin{figure*}
	\begin{IEEEeqnarray}{rCl}
		G^{(\M)}_{\beta,\lambda}(\sigma^{2},\{\nu_{t}\}) 
		&=&
		\frac{1}{M}\ln  \E_{\vm{w}} 
		\int \prod_{t=1}^{T}
		\Big( \dx \Lambda_{t} \e^{\frac{\M}{2}\Lambda_{t} \nu_{t} -\frac{\M}{2} \ln \Lambda_{t}} 
		\Big) 
		\int \exp\bigg[
		-\frac{1}{2} \bigg(\lambda+\sum_{t=1}^{T} 
		\frac{\beta}{\Lambda_{t}}\bigg)\|\vm{k}\|^{2}
		+ \im \sqrt{\beta\sigma^{2}}\vm{w}^{\trans}\vm{k}\bigg]\dx \vm{k}
		\IEEEeqnarraynumspace\IEEEnonumber\\
		&&
		+ \frac{1}{2} \ln \frac{\lambda}{2\pi} 
		- \frac{1}{\M} \ln Z (\{\nu_{t}\})
		\IEEEeqnarraynumspace\IEEEnonumber\\
		&=& \frac{1}{M} \ln 
		\int \bigg(\prod_{t=1}^{T} \dx \Lambda_{t}\bigg)
		\exp\bigg\{\frac{\M}{2}\bigg[\sum_{t=1}^{T}
		\Lambda_{t}\nu_{t} -\sum_{t=1}^{T}\ln \Lambda_{t} -  \ln \bigg(\lambda+\sum_{t=1}^{T} 
		\frac{\beta}{\Lambda_{t}}\bigg)  \bigg]\bigg\}
		\IEEEeqnarraynumspace\IEEEnonumber\\
		&& \times \frac{1}{(2 \pi)^{\M/2}} \int 
		\exp\bigg\{-\frac{1}{2}
		\bigg[1 + \beta\sigma^{2} \bigg(\lambda+\sum_{t=1}^{T} 
		\frac{\beta}{\Lambda_{t}}\bigg)^{-1}
		\bigg] \|\vm{w} \|^{2}\bigg\} \dx \vm{w}
		+ \frac{1}{2} \ln \lambda - \frac{1}{\M} \ln Z (\{\nu_{t}\})
		\IEEEeqnarraynumspace\IEEEnonumber\\
		&=& \frac{1}{M} \ln \int \exp\bigg\{\frac{\M}{2}\bigg[
		\sum_{t=1}^{T} \Lambda_{t}\nu_{t}
		- \sum_{t=1}^{T}\ln \Lambda_{t}
		- \ln \bigg(\lambda+ \beta\sigma^{2} +\sum_{t=1}^{T} 
		\frac{\beta}{\Lambda_{t}} \bigg) \bigg]\bigg\} \prod_{t=1}^{T} \dx \Lambda_{t}
		\IEEEeqnarraynumspace\IEEEnonumber\\
		&& + \frac{1}{2} \ln \lambda  - \frac{1}{\M} \ln Z (\{\nu_{t}\})
		\IEEEeqnarraynumspace 
		\label{eq:Favg_1}
	\end{IEEEeqnarray}
	\hrulefill
\end{figure*}
\eqref{eq:Favg_1} at the top of the next page.
We then change the integration variables as
$\Lambda_{t} \to \beta \Lambda_{t}$, take the limit 
$\M\to\infty$ and employ saddle-point integration. Omitting
all terms that vanish in the large-$\M$ limit provides
the final expression
\begin{IEEEeqnarray}{l}
	G_{\beta,\lambda}(\sigma^{2},\{\nu_{t}\})  
	= - \frac{1}{2} \bigg(T - \ln \lambda +\sum_{t=1}^{T}  \ln (\beta\nu_{t}) \bigg) 
	\IEEEnonumber\\
	+\frac{1}{2}\extr_{\{\Lambda_{t}\}}
	\bigg\{\! \sum_{t=1}^{T} \big[\Lambda_{t} (\beta \nu_{t}) \!-\! \ln \Lambda_{t}\big]
	\!-\! \ln \bigg(\lambda+ \beta\sigma^{2} +\sum_{t=1}^{T} 
	\frac{1}{\Lambda_{t}}\bigg)\! \bigg\}.
	\IEEEnonumber\\
	\label{eq:Ffunc_final}
\end{IEEEeqnarray}

Finally, we remark that the extremization in 
$G_{\beta,\lambda}(\sigma^{2},\{\nu_{t}\})$
as given above enforces the condition 
\begin{IEEEeqnarray}{l} 
	\beta\nu_{t}(\sigma^{2}, \beta, \lambda)
	= \frac{1}{\Lambda_{t}} 
	\Bigg(1- \frac{\Lambda_{t}^{-1}}{\lambda+\beta\sigma^{2}+\sum_{t=1}^{T}\Lambda_{t}^{-1}}
	\Bigg),
	\IEEEeqnarraynumspace
	\label{eq:nu_condition}
\end{IEEEeqnarray}
implying $\Lambda_{t}\in\R\setminus\{0\}$ for all 
$\{\beta, \lambda, \sigma^{2}\}$ and $t = 1,\ldots,T$.
Thus, the expression \eqref{eq:Ffunc_final}
together with the condition \eqref{eq:nu_condition}
provides the solution to the integration problem 
defined in \eqref{eq:Ffunc_1}.  Furthermore, for the special case of
$\sigma = 0$ we have
\begin{IEEEeqnarray}{rCl}
	\beta \nu_{t}(\sigma^{2} = 0, \beta, \lambda)
	&=& \frac{1}{\Lambda_{t}}\Bigg(
	1 - \frac{\Lambda_{t}^{-1}}{\lambda + \sum_{k=1}^{T}\Lambda_{k}^{-1}}\Bigg), 
	\IEEEeqnarraynumspace
	\label{eq:sp_cond_tortho}
\end{IEEEeqnarray}
so that $\nu_{t}(\sigma^{2} = 0, \beta\to\infty, \lambda) \to 0$
and $\beta^{-1} G_{\beta,\lambda}(\sigma^{2}=0,\{\nu_{t}\}) \xrightarrow{\beta\to \infty} 0$. 
This is fully compatible with the earlier result obtained 
in \cite{Kabashima-Vehkapera-Chatterjee-2012}, as expected.

\subsection{Rotationally Invariant Setup}
\label{sec:Ffunc_rotational}

Let us consider the case where $\vm{A} \in \R^{\M \times \N}$ is sampled from 
an ensemble that allows the decomposition
$\vm{R} = \vm{A}^{\trans}\vm{A}= \vm{O}^{\trans} \vm{D} \vm{O}$ where
$\vm{O}$ is an $\N \times \N$ Haar matrix and 
$\vm{D} = \diag(d_{1}, \ldots, d_{\N})$ contains the eigenvalues 
of $\vm{R}$.  This is the case of rotationally invariant setup 
given in Definition~\ref{defn:matrix_ensembles}.
Furthermore, let $\{\Delta\vm{x}_{t}\}_{\tx=1}^{\Tx}$
be a set of (fixed) length-$\Mx$ vectors satisfying
$\|\Delta\vm{x}_{t}\|^{2} = \Mx \nu_{t}$ 
for some given non-negative values $\{\nu_{t}\}$ 
and $\N = \Tx \Mx$.
For notational convenience, we write 
$\Delta\vm{x}\in\R^{\N}$ for the vector obtained by 
stacking $\{\Delta\vm{x}_{t}\}$.
The counterpart of \eqref{eq:Ffunc_0} reads then
\begin{IEEEeqnarray}{l}
	H^{(\N)}_{\beta,\lambda}(\sigma^{2},\{\nu_{t}\}) \IEEEnonumber\\
	\quad = \frac{1}{\N} \ln \E_{\vm{w},\vm{A}}
	\e^{-\frac{\beta}{2\lambda} \| \sigma\vm{w} -  \vm{A}\Delta\vm{x}\|^{2}},
	\IEEEnonumber \\
	\quad =
	-\frac{\alpha}{2}\ln \bigg(1+\frac{\beta\sigma^{2}}{\lambda}\bigg)
	\IEEEnonumber\\
	\qquad + \frac{1}{\N} \ln \E_{\vm{R}}\exp\bigg[
	- \frac{1}{2}\bigg(\frac{\beta}{\lambda+\beta\sigma^{2}}\bigg)
	\Delta\vm{x}^{\trans}\vm{R}\Delta\vm{x}\bigg],
	\IEEEeqnarraynumspace
	\label{eq:Hfunc_1_row}
\end{IEEEeqnarray}
where the second equality follows by using Gaussian integration 
formula \eqref{eq:Gint} to 
average over the additive noise term $\vm{w}$.  
Recall next the fact that $\vm{R}
=\vm{O}^{\trans} \vm{D} \vm{O}$ and 
denote $\vm{u} = \vm{O}\Delta\vm{x}$.  Since $\vm{O}$
are Haar matrices and
\begin{equation}
	\|\vm{O}\Delta\vm{x}\|^{2} 
	= T \Mx \sum_{t=1}^{T} \frac{\nu_{t}}{T}
	= \N \nuAvg,
\end{equation}
where $\nuAvg$ is the ``empirical average'' over 
$\{\nu_{t}\}$, we get 
by the same arguments as in Appendix~\ref{sec:Ffunc_Tortho}
\begin{figure*}
	\begin{IEEEeqnarray}{rCl}
		H^{(\N)}_{\beta,\lambda}(\sigma^{2},\nu) &=&
		\frac{1}{\N}
		\ln \E_{\vm{D}}
		\int\dx \Lambda \e^{\frac{\Lambda}{2} \N \nuAvg}
		\int \exp\bigg[-\frac{1}{2}\vm{u}^{\trans}\bigg(\Lambda \I_{\N} + 
		\frac{\beta}{\lambda+\beta\sigma^{2}}\vm{D} \bigg)
		\vm{u} \bigg] \dx \vm{u} \IEEEeqnarraynumspace\IEEEnonumber\\
		&& -\frac{\alpha}{2}\ln \bigg(1+\frac{\beta\sigma^{2}}{\lambda}\bigg)
		- \frac{1 + \ln \nuAvg}{2}
		+ O(N^{-1}) \IEEEnonumber\IEEEeqnarraynumspace\\
		&=&
		\frac{1}{\N} \ln \E_{\vm{D}}
		\int \exp\bigg\{\frac{\N}{2} \bigg[ \Lambda (\beta\nuAvg) - \frac{1}{\N} 
		\sum_{n=1}^{\N} \ln \bigg(\Lambda + \frac{1}{\lambda+\beta\sigma^{2}}d_{n} \bigg) 
		\bigg]\bigg\} \dx \Lambda \IEEEeqnarraynumspace\IEEEnonumber\\
		&& -\frac{\alpha}{2}\ln \bigg(1+\frac{\beta\sigma^{2}}{\lambda}\bigg)
		- \frac{1 + \ln (\beta \nuAvg)}{2}  
		+ O(N^{-1}) \IEEEeqnarraynumspace
		\label{eq:Hrotationally_invariant}
	\end{IEEEeqnarray}
	\hrulefill
\end{figure*}
an expression for $H^{(\N)}_{\beta,\lambda}(\sigma^{2},\nu)$ as given in
\eqref{eq:Hrotationally_invariant} at the top of the next page.
Considering next the limit of large $M$ and $N$, we 
replace the summation in \eqref{eq:Hrotationally_invariant} 
by an integral over the 
empirical distribution of the eigenvalues
\eqref{eq:eed}, so that the outer expectation w.r.t.\ $\vm{D}$
becomes an expectation over all empirical eigenvalue distributions
of $\vm{R}$.  But when
$\M,\N\to\infty$ with a finite and fixed ratio
$\alpha = \M / \N$, this expectation is by assumption w.r.t.\ a probability 
measure that has a single non-zero point corresponding to the limiting deterministic 
eigenvalue distribution $F_{\vm{A}^{\trans}\vm{A}}$.  Finally, using 
saddle point method to integrate over $\Lambda$, we obtain
\begin{IEEEeqnarray}{l}
	H_{\beta,\lambda}(\sigma^{2},\nu) 
	\IEEEnonumber\\
	\quad =  \frac{1}{2}
	\extr_{\Lambda}\bigg\{\Lambda (\beta\nuAvg) 
	- \int \ln  \bigg(\Lambda + 
	\frac{1}{\lambda+\beta\sigma^{2}}x
	\bigg) 
	\dx F_{\vm{A}^{\trans}\vm{A}}(x)
	\bigg\} \IEEEnonumber\\
	\qquad -\frac{\alpha}{2}\ln \bigg(1+\frac{\beta\sigma^{2}}{\lambda}\bigg)
	- \frac{1 + \ln (\beta \nuAvg)}{2} 
	\label{eq:H_func_appendix1}
	\IEEEeqnarraynumspace\\
	\quad =
	\frac{1}{2}
	\extr_{\Lambda}\bigg\{\Lambda (\beta\nuAvg) 
	- (1-\alpha) \ln \Lambda \IEEEnonumber\\
	\qquad \qquad \qquad- \alpha \int \ln (\Lambda\beta\sigma^{2} +\Lambda\lambda+ x) 
	\dx F_{\vm{A}\vm{A}^{\trans}}(x)
	\bigg\} \IEEEnonumber\\
	\qquad - \frac{1 + \ln (\beta \nuAvg)- \alpha\ln \lambda}{2},
	\label{eq:H_func_appendix}
	\IEEEeqnarraynumspace
\end{IEEEeqnarray}
where the second equality is obtained by changing the integral
measure and simplifying.  
For the case $\sigma^{2} = 0$, the extremization 
then provides the condition
\begin{IEEEeqnarray}{l}
	\Lambda-\frac{1}{\beta\nu} 
	= - \frac{\alpha}{\beta\nu}\bigg(1 
	- \int 
	\frac{\Lambda\beta\sigma^{2}+\Lambda\lambda}{ \Lambda \beta\sigma^{2} +  \Lambda \lambda +x}
	\dx F_{\vm{A}\vm{A}^{\trans}}(x) \bigg) \IEEEnonumber\\
	\quad \xrightarrow{\sigma^{2}=0}
	\Lambda-\frac{1}{\beta\nu} 
	= - \frac{\alpha}{\beta\nu} 
	\big[1- (\Lambda\lambda) G_{\vm{A}\vm{A}^{\trans}}(-\Lambda\lambda) \big],
	\IEEEeqnarraynumspace
	\label{eq:Fextr_eigs}
\end{IEEEeqnarray}
where we used again the Stieltjes transformation 
\eqref{eq:stieltjes} of $F_{\vm{A}\vm{A}^{\trans}}(x)$.


\section{Geometric Ensemble}
\label{sec:geometric_ensemble}

Recall that the geometric singular value ensemble is generated as
$\vm{A} = \vm{U} \vm{\Sigma} \vm{V}^{\trans}$ where $\vm{U}$ and $\vm{V}$ are independent 
Haar matrices. The diagonal elements of $\vm{\Sigma}$ are the singular values 
$\sigma_{m}=\sqrt{\sigmanorm}\tau^{m-1}, m = 1,\ldots,M$ of $\vm{A}$ with $\tau \in (0,1]$
and $\sigmanorm > 0$ such that $N^{-1}\sum_{m=1}^{M}\lambda_{m} = 1$ 
where $\lambda_{m}=\sigma_{m}^{2}$ are the eigenvalues of $\vm{A}\vm{A}^{\trans}$. 
Alternatively, we may write $\lambda_{i+1} = \sigmanorm\e^{-\gamma_{M} (i/M)}, i = 0,1,\ldots,M-1$ where
$\gamma_{M} = -2 M \ln \tau \geq 0$. Letting $\M\to\infty$ provides the continuous 
limit function for the eigenvalues
\begin{equation}
	\label{eq:lambda_continuous}
	\lambda(t) = \sigmanormC e^{-\gamma t}, \qquad t \in [0,1),
\end{equation}
where $\gamma>0$ satisfies
\begin{equation}
	\label{eq:kappa_continuous}
	\kappa = \frac{\lambda(0)}{\int_{0}^{1} \lambda(t) d t} 
	= \frac{\gamma}{1-e^{-\gamma}},
\end{equation}
for the given peak-to-average ratio $\kappa$.  The normalization condition 
$N^{-1}\sum_{m=1}^{M}\lambda_{m} = 1$ becomes now
\begin{equation}
	\label{eq:scaling_of_eigs_cont}
	\alpha \int_{0}^{1} \sigmanormC e^{-\gamma t} d t = 1 \iff
	\sigmanormC = \frac{\kappa}{\alpha},
\end{equation}
which means that 
$\lambda(t) \in [\frac{\kappa}{\alpha} e^{-\gamma},\frac{\kappa}{\alpha}]$.

The function \eqref{eq:lambda_continuous} describes the eigenvalues
of $\vm{A}\vm{A}^{\trans}$ in the large system limit.  Since the order 
of the eigenvalues and associated eigenvectors does not affect the performance 
of the reconstruction, we may also consider sampling randomly and uniformly 
$t \in [0,1)$ and assigning the corresponding eigenvalues according to 
\eqref{eq:lambda_continuous}.  Then, by 
construction the limit of \eqref{eq:eed} for this ensemble is 
given by $F_{\vm{A}\vm{A}^{T}}(A e^{-\gamma t}) = 1 - t,\, t \in [0,1)$
or more conveniently
\begin{equation}
	F_{\vm{A}\vm{A}^{T}}(x) = 
	\Bigg\{
	\begin{array}{ll}
	1 + \gamma^{-1} \ln x - \gamma^{-1} \ln A, & \textrm{if } x \in (A e^{-\gamma},A],\\
	0, &  \textrm{otherwise,}
	\end{array}
\end{equation}
where we wrote for simplicity $A = \sigmanormC$.  This is also called
the reciprocal distribution whose density reads
\begin{equation}
	\label{eq:pdf_continuous}
	f_{\vm{A}\vm{A}^{T}}(x) = 
	\Bigg\{
	\begin{array}{ll}
	\frac{1}{\gamma x}, & \quad \textrm{if } x \in (A e^{-\gamma},A],\\
	0, & \quad \textrm{otherwise.}
	\end{array}
\end{equation}
For the analysis, one can obtain the Stieltjes transform of
\eqref{eq:pdf_continuous} directly from the definition \eqref{eq:stieltjes},
as given in Example~\ref{example:geometric}.  The sensing matrices for 
the geometric setup in finite size simulations, on the other hand,
can be constructed as follows:
\begin{enumerate}
	\item 
	Generate $\M \times \N$ matrix $\vm{X}$ with IID standard normal elements and calculate the 
	singular value decomposition $\vm{X} = \vm{U} \vm{S} \vm{V}^{\trans}$. 
	For the Gaussian ensemble, $\vm{U}$ and $\vm{V}$ are independent Haar matrices.
	\item Find numerically the value of $\tau$ that meets the peak-to-average constraint
	\eqref{eq:peaktoaverage} and set 
	\begin{equation}
	\sigmanorm = \frac{1}{N^{-1}\sum_{m=1}^{M}\tau^{2(m-1)}},
	\end{equation}
	so that the average power constraint is satisfied.
	\item
	Replace $\vm{S}$ by $\vm{\Sigma}$ to create a sensing matrix 
	$\vm{X} = \vm{U} \vm{\Sigma} \vm{V}^{\trans}$.  Note that permutations 
	of the diagonal elements in $\vm{\Sigma}$ has no impact on the reconstruction
	performance.
\end{enumerate}


\bibliography{./IEEEabrv,./biblio_saikat_CS}

\begin{thebibliography}{10}
\providecommand{\url}[1]{#1}
\csname url@samestyle\endcsname
\providecommand{\newblock}{\relax}
\providecommand{\bibinfo}[2]{#2}
\providecommand{\BIBentrySTDinterwordspacing}{\spaceskip=0pt\relax}
\providecommand{\BIBentryALTinterwordstretchfactor}{4}
\providecommand{\BIBentryALTinterwordspacing}{\spaceskip=\fontdimen2\font plus
\BIBentryALTinterwordstretchfactor\fontdimen3\font minus
  \fontdimen4\font\relax}
\providecommand{\BIBforeignlanguage}[2]{{%
\expandafter\ifx\csname l@#1\endcsname\relax
\typeout{** WARNING: IEEEtran.bst: No hyphenation pattern has been}%
\typeout{** loaded for the language `#1'. Using the pattern for}%
\typeout{** the default language instead.}%
\else
\language=\csname l@#1\endcsname
\fi
#2}}
\providecommand{\BIBdecl}{\relax}
\BIBdecl

\bibitem{Donoho_2006_Compressed_sensing}
D.~L. Donoho, ``Compressed sensing,'' \emph{IEEE Trans. Inform. Theory},
  vol.~52, no.~4, pp. 1289--1306, Apr. 2006.

\bibitem{Candes_Romberg_Tao_2006}
E.~J. Candes, J.~Romberg, and T.~Tao, ``Robust uncertainty principles: Exact
  signal reconstruction from highly incomplete frequency information,''
  \emph{{IEEE} Trans. Inf. Theory}, vol.~52, no.~2, pp. 489--509, Feb. 2006.

\bibitem{Candes_Tao_NearOptimal_2006}
E.~J. Candes and T.~Tao, ``Near-optimal signal recovery from random
  projections: Universal encoding strategies?'' \emph{{IEEE} Trans. Inf.
  Theory}, vol.~52, no.~12, pp. 5406--5425, Dec. 2006.

\bibitem{Tibshirani_1996_lasso}
R.~Tibshirani, ``Regression shrinkage and selection via the lasso,'' \emph{J.
  Royal. Statist. Soc., Ser. B}, vol.~58, no.~1, pp. 267--288, 1996.

\bibitem{chen1998atomic}
S.~S. Chen, D.~L. Donoho, and M.~A. Saunders, ``Atomic decomposition by basis
  pursuit,'' \emph{SIAM J. Sci Comp.}, vol.~20, no.~1, pp. 33--61, 1998.

\bibitem{cvx}
{CVX Research, Inc.}, ``{CVX}: Matlab software for disciplined convex
  programming,'' \url{http://cvxr.com/cvx}.

\bibitem{Donoho_Huo_2001}
D.~Donoho and X.~Huo, ``Uncertainty principles and ideal atomic
  decomposition,'' \emph{IEEE Trans. Inform. Theory}, vol.~47, no.~7, pp.
  2845--2862, Nov. 2001.

\bibitem{Elad_Bruckstein_2002}
M.~Elad and A.~Bruckstein, ``A generalized uncertainty principle and sparse
  representation in pairs of bases,'' \emph{IEEE Trans. Inform. Theory},
  vol.~48, no.~9, pp. 2558--2567, Sep. 2002.

\bibitem{Donoho_Elad_2003}
D.~Donoho and M.~Elad, ``Optimally sparse representation in general
  (non-orthogonal) dictionaries via $l_1$ minimization,'' \emph{Proc. Nat.
  Acad. Sci.}, vol. 100, no.~5, pp. 2197--2202, Nov. 2003.

\bibitem{Candes_Tao_2005}
E.~Candes and T.~Tao, ``Decoding by linear programming,'' \emph{IEEE Trans.
  Inform. Theory}, vol.~51, no.~12, pp. 4203 -- 4215, dec. 2005.

\bibitem{Baraniuk_asimple_proof_RIP_2007}
R.~Baraniuk, M.~Davenport, R.~Devore, and M.~Wakin, ``A simple proof of the
  restricted isometry property for random matrices,'' \emph{Constr Approx},
  vol.~28, no.~3, pp. 253--263, 2008.

\bibitem{Tropp_2007_OMP}
J.~Tropp and A.~Gilbert, ``Signal recovery from random measurements via
  orthogonal matching pursuit,'' \emph{{IEEE} Trans. Inf. Theory}, vol.~53,
  no.~12, pp. 4655--4666, Dec. 2007.

\bibitem{Davenport_2010_Orthogonal_Matching_pursuit}
M.~A. Davenport and W.~B. Wakin, ``Analysis of orthogonal matching pursuit
  using the restricted isometry property,'' \emph{{IEEE} Trans. Inf. Theory},
  vol.~56, no.~9, pp. 4395--4401, Sep. 2010.

\bibitem{Cai_OMP_TIT_2011}
T.~T. Cai and L.~Wang, ``Orthogonal matching pursuit for sparse signal recovery
  with noise,'' \emph{{IEEE} Trans. Inf. Theory}, vol.~57, no.~7, pp.
  4680--4688, Jul. 2011.

\bibitem{Dai_2009_Subspace_pursuit}
W.~Dai and O.~Milenkovic, ``Subspace pursuit for compressive sensing signal
  reconstruction,'' \emph{{IEEE} Trans. Inf. Theory}, vol.~55, no.~5, pp.
  2230--2249, May 2009.

\bibitem{Needell_Tropp_2009_CoSaMP}
D.~Needell and J.~A. Tropp, ``{CoSaMP}: Iterative signal recovery from
  incomplete and inaccurate samples,'' \emph{Applied and Computational Harmonic
  Analysis}, vol.~26, no.~3, pp. 301--321, 2009.

\bibitem{Lv_Bi_Wan_2011_Group_Lasso}
X.~Lv, G.~Bi, and C.~Wan, ``The group lasso for stable recovery of block-sparse
  signal representations,'' \emph{{IEEE} Trans. Signal Process.}, vol.~59,
  no.~4, pp. 1371--1382, Apr. 2011.

\bibitem{Ambat_Chatterjee_Hari_FACS_2013}
S.~K. Ambat, S.~Chatterjee, and K.~V.~S. Hari, ``Fusion of algorithms for
  compressed sensing,'' \emph{{IEEE} Trans. Signal Process.}, vol.~61, no.~14,
  pp. 3699--3704, Jul. 2013.

\bibitem{Donoho-Tanner-2009}
D.~Donoho and J.~Tanner, ``Counting faces of randomly projected polytopes when
  the projection radically lowers dimension,'' \emph{Journal of the American
  Mathematical Society}, vol.~22, no.~1, pp. 1--53, 2009.

\bibitem{Donoho-Tanner-2010}
D.~L. Donoho and J.~Tanner, ``Counting the faces of randomly-projected
  hypercubes and orthants, with applications,'' \emph{Discrete \& Computational
  Geometry}, vol.~43, no.~3, pp. 522--541, 2010.

\bibitem{Donoho-Tanner-PrIEEE2010}
------, ``Precise undersampling theorems,'' \emph{Proc. {IEEE}}, vol.~98,
  no.~6, pp. 913--924, Jun. 2010.

\bibitem{Donoho-Maleki-Montanari-2009}
D.~L. Donoho, A.~Maleki, and A.~Montanari, ``Message passing algorithms for
  compressed sensing,'' \emph{Proc. Nat. Acad. Sci.}, vol. 106, pp.
  18\,914--18\,919, 2009.

\bibitem{Montanari-graphical-2012}
A.~Montanari, ``Graphical models concepts in compressed sensing,'' in
  \emph{Compressed sensing: Theory and Applications}, Y.~Eldar and G.~Kutyniok,
  Eds.\hskip 1em plus 0.5em minus 0.4em\relax Cambridge University Press, 2012,
  pp. 394--438.

\bibitem{Bayati-Montanari-2011}
M.~Bayati and A.~Montanari, ``The dynamics of message passing on dense graphs,
  with applications to compressed sensing,'' \emph{{IEEE} Trans. Inf. Theory},
  vol.~57, no.~2, pp. 764--785, Feb. 2011.

\bibitem{mezard1987spin}
M.~M{\'e}zard, G.~Parisi, and M.~A. Virasoro, \emph{Spin Glass Theory and
  Beyond}.\hskip 1em plus 0.5em minus 0.4em\relax Singapore: World Scientific,
  1987.

\bibitem{Dotsenko-2001}
V.~Dotsenko, \emph{Introduction to the Replica Theory of Disordered Statistical
  Systems}.\hskip 1em plus 0.5em minus 0.4em\relax New York: Cambridge
  University Press, 2001.

\bibitem{Nishimori-2001}
H.~Nishimori, \emph{Statistical Physics of Spin Glasses and Information
  Processing}.\hskip 1em plus 0.5em minus 0.4em\relax New York: Oxford
  University Press, 2001.

\bibitem{Rangan-Fletcher-Goyal-2012}
S.~Rangan, A.~K. Fletcher, and V.~K. Goyal, ``Asymptotic analysis of {MAP}
  estimation via the replica method and applications to compressed sensing,''
  \emph{IEEE Trans. Inform. Theory}, vol.~58, no.~3, pp. 1902--1923, Mar. 2012.

\bibitem{Guo-Baron-Shamai-2009}
D.~Guo, D.~Baron, and S.~Shamai, ``A single-letter characterization of optimal
  noisy compressed sensing,'' in \emph{Proc. Annual Allerton Conf. Commun.,
  Contr., Computing}, Sep. 30 - Oct. 2 2009, pp. 52--59.

\bibitem{Tulino-etal-cs2013}
A.~M. Tulino, G.~Caire, S.~Verd{\'u}, and S.~Shamai, ``Support recovery with
  sparsely sampled free random matrices,'' \emph{{IEEE} Trans. Inf. Theory},
  vol.~59, no.~7, pp. 4243--4271, Jul. 2013.

\bibitem{Vehkapera_Kabashima_Chatterjee_et_all_2012_ITW}
M.~Vehkaper\"{a}, Y.~Kabashima, S.~Chatterjee, E.~Aurell, M.~Skoglund, and
  L.~Rasmussen, ``Analysis of sparse representations using bi-orthogonal
  dictionaries,'' in \emph{Proc. IEEE Inform. Theory Workshop}, Sep. 3--7 2012.

\bibitem{Kabashima-Vehkapera-Chatterjee-2012}
Y.~Kabashima, M.~Vehkaper{\"a}, and S.~Chatterjee, ``Typical $l_{1}$-recovery
  limit of sparse vectors represented by concatenations of random orthogonal
  matrices,'' \emph{J. Stat. Mech.}, vol. 2012, no.~12, p. P12003, 2012.

\bibitem{Tanaka_Raymond_2010}
T.~Tanaka and J.~Raymond, ``Optimal incorporation of sparsity information by
  weighted $l_1$-optimization,'' in \emph{Proc. IEEE Int. Symp. Inform.
  Theory}, Jun. 2010, pp. 1598--1602.

\bibitem{Kabashima-Wadayama-Tanaka-2009}
Y.~Kabashima, T.~Wadayama, and T.~Tanaka, ``A typical reconstruction limit for
  compressed sensing based on $l_{p}$-norm minimization,'' \emph{J. Stat.
  Mech.}, vol. 2009, no.~9, p. L09003, 2009.

\bibitem{Talagrand-2003}
M.~Talagrand, \emph{Spin Glasses: A Challenge for Mathematicians, Cavity and
  Mean Field Models}.\hskip 1em plus 0.5em minus 0.4em\relax Berlin Heidelberg:
  Springer-Verlag, 2003.

\bibitem{Guerra-Toninelli-2002}
F.~Guerra and F.~L. Toninelli, ``Quadratic replica coupling in the
  {S}herrington-{K}irkpatrick mean field spin glass model,'' \emph{J. Math.
  Phys.}, vol.~43, no.~7, pp. 3704--3716, 2002.

\bibitem{Guerra-Toninelli-2002-2}
------, ``The thermodynamic limit in mean field spin glass models,''
  \emph{Commun. Math. Phys.}, vol. 230, no.~1, pp. 71--79, 2002.

\bibitem{Guerra-2003}
F.~Guerra, ``Broken replica symmetry bounds in the mean field spin glass
  model,'' \emph{Commun. Math. Phys.}, vol. 233, no.~1, pp. 1--12, 2003.

\bibitem{Talagrand-2006}
M.~Talagrand, ``The {P}arisi formula,'' \emph{Annals of Math}, vol. 163, no.~1,
  pp. 221--263, 2006.

\bibitem{Korada-Macris-2010}
S.~B. Korada and N.~Macris, ``Tight bounds on the capacity of binary input
  random {CDMA} systems,'' \emph{{IEEE} Trans. Inf. Theory}, vol.~56, no.~11,
  pp. 5590--5613, Nov. 2010.

\bibitem{Montanari-2005}
A.~Montanari, ``Tight bounds for {LDPC} and {LDGM} codes under {MAP}
  decoding,'' \emph{{IEEE} Trans. Inf. Theory}, vol.~51, no.~9, pp. 3221--3246,
  Sep. 2005.

\bibitem{Kudekar-Macris-2009}
S.~Kudekar and N.~Macris, ``Sharp bounds for optimal decoding of low-density
  parity-check codes,'' \emph{{IEEE} Trans. Inf. Theory}, vol.~55, no.~10, pp.
  4635--4650, Oct. 2009.

\bibitem{Tanaka-2002}
T.~Tanaka, ``A statistical-mechanics approach to large-system analysis of
  {CDMA} multiuser detectors,'' \emph{IEEE Trans. Inform. Theory}, vol.~48,
  no.~11, pp. 2888--2910, Nov. 2002.

\bibitem{Guo-Verdu-2005Jun}
D.~Guo and S.~Verd{\'u}, ``Randomly spread {CDMA}: Asymptotics via statistical
  physics,'' \emph{{IEEE} Trans. Inf. Theory}, vol.~51, no.~6, pp. 1983--2010,
  Jun. 2005.

\bibitem{Harish-Chandra-1957}
Harish-Chandra, ``Differential operators on a semisimple lie algebra,''
  \emph{Amer. J. Math.}, vol.~79, no.~1, pp. 87--120, 1957.

\bibitem{Itzykson-Zuber-1980}
C.~Itzykson and J.~B. Zuber, ``Planar approximation 2,'' \emph{J. Math. Phys.},
  vol.~21, no.~3, pp. 411--421, 1980.

\bibitem{Wright-Ma-2010}
J.~Wright and Y.~Ma, ``Dense error correction via $\ell_{1}$-minimization,''
  \emph{{IEEE} Trans. Inf. Theory}, vol.~56, no.~7, pp. 3540--3560, Jul. 2010.

\bibitem{Vehkapera-Kabashima-Chatterjee-2013}
M.~Vehkaper{\"a}, Y.~Kabashima, and S.~Chatterjee, ``Statistical mechanics
  approach to sparse noise denoising,'' in \emph{Proc. European Sign. Proc.
  Conf.}, Sep. 9--13 2013.

\bibitem{Wu-Verdu-2012}
Y.~Wu and S.~Verd{\'u}, ``Optimal phase transitions in compressed sensing,''
  \emph{{IEEE} Trans. Inf. Theory}, vol.~58, no.~10, pp. 6241--6263, Oct. 2012.

\bibitem{Korada-Montanari-2011}
S.~B. Korada and A.~Montanari, ``Applications of the {L}indeberg principle in
  communications and statistical learning,'' \emph{{IEEE} Trans. Inf. Theory},
  vol.~57, no.~4, pp. 2440--2450, Apr. 2011.

\bibitem{viswanath1999optimal}
P.~Viswanath, V.~Anantharam, and D.~N.~C. Tse, ``Optimal sequences, power
  control, and user capacity of synchronous {CDMA} systems with linear {MMSE}
  multiuser receivers,'' \emph{{IEEE} Trans. Inf. Theory}, vol.~45, no.~6, pp.
  1968--1983, Sep. 1999.

\bibitem{Kitagawa-Tanaka-2010}
K.~Kitagawa and T.~Tanaka, ``Optimization of sequences in {CDMA} systems: A
  statistical-mechanics approach,'' \emph{Computer Networks}, vol.~54, no.~6,
  pp. 917--924, 2010.

\bibitem{Oymak-Hassibi-isit2014}
S.~Oymak and B.~Hassibi, ``A case for orthogonal measurements in linear inverse
  problems,'' in \emph{Proc. IEEE Int. Symp. Inform. Theory}, 2014, pp.
  3175--3179.

\bibitem{Thrampoulidis-Hassibi-isit2015}
C.~Thrampoulidis and B.~Hassibi, ``Isotropically random orthogonal matrices:
  Performance of {LASSO} and minimum conic singular values,'' in \emph{Proc.
  IEEE Int. Symp. Inform. Theory}, 2015, pp. 556--560.

\bibitem{Wen-etal-arxiv2014}
C.-K. Wen, J.~Zhang, K.-K. Wong, J.-C. Chen, and C.~Yuen, ``On sparse vector
  recovery performance in structurally orthogonal matrices via {LASSO},''
  arXiv:1410.7295 [cs.IT], Oct. 2014.

\bibitem{Kabashima-Vehkapera-isit2014}
Y.~Kabashima and M.~Vehkaper{\"a}, ``Signal recovery using expectation
  consistent approximation for linear observations,'' in \emph{Proc. IEEE Int.
  Symp. Inform. Theory}, 2014, pp. 226--230.

\bibitem{Cakmak-Winther-Fleury-itw2014}
B.~Cakmak, O.~Winther, and B.~H. Fleury, ``{S-AMP}: Approximate message passing
  for general matrix ensembles,'' in \emph{Proc. IEEE Inform. Theory Workshop},
  2014, pp. 192--196.

\bibitem{Ma-Yuan-Ping-2015}
J.~Ma, X.~Yuan, and L.~Ping, ``On the performance of turbo signal recovery with
  partial {DFT} sensing matrices,'' \emph{{IEEE} Trans. Signal Process.},
  vol.~22, no.~10, pp. 1580--1584, Oct. 2015.

\bibitem{Opper-Cakmak-Winther-arxiv2015}
M.~Opper, B.~Cakmak, and O.~Winther, ``A theory of solving {TAP} equations for
  {I}sing models with general invariant random matrices,'' arXiv:1509.01229
  [cond-mat.dis-nn], Sep. 2015.

\bibitem{Kudekar-Pfister-Allerton-2010}
S.~Kudekar and H.~D. Pfister, ``The effect of spatial coupling on compressive
  sensing,'' in \emph{Proc. Annual Allerton Conf. Commun., Contr., Computing},
  Sep. 29 - Oct. 1 2010, pp. 347--353.

\bibitem{Donoho-Javanmard-Montanari-isit2012}
D.~L. Donoho, A.~Javanmard, and A.~Montanari, ``Information-theoretically
  optimal compressed sensing via spatial coupling and approximate message
  passing,'' in \emph{Proc. IEEE Int. Symp. Inform. Theory}, Jul. 1 - 6 2012,
  pp. 1231--1235.

\bibitem{Krzakala-etal-CS-2012-2}
F.~Krzakala, M.~M{\'e}zard, F.~Sausset, Y.~Sun, and L.~Zdeborov{\'a},
  ``Probabilistic reconstruction in compressed sensing: algorithms, phase
  diagrams, and threshold achieving matrices,'' \emph{J. Stat. Mech.}, vol.
  2012, no.~08, p. P08009, 2012.

\bibitem{Krzakala-etal-CS-2012}
F.~Krzakala, M.~M{\'e}zard, F.~Sausset, Y.~F. Sun, and L.~Zdeborov{\'a},
  ``Statistical-physics-based reconstruction in compressed sensing,''
  \emph{Phys. Rev. X}, vol.~2, p. 021005, May 2012.

\bibitem{Takeda-Uda-Kabashima-2006}
K.~Takeda, S.~Uda, and Y.~Kabashima, ``Analysis of {CDMA} systems that are
  characterized by eigenvalue spectrum,'' \emph{Europhys. Lett.}, vol.~76,
  no.~6, p. 1193, 2006.

\bibitem{Kabashima-confser2008}
Y.~Kabashima, ``Inference from correlated patterns: A unified theory for
  perceptron learning and linear vector channels,'' \emph{J. Phys. Conf. Ser.},
  vol.~95, no.~1, p. 012001, 2008.

\bibitem{Rangan-Fletcher-Schniter-Kamilov-arxiv2015}
S.~Rangan, A.~K. Fletcher, P.~Schniter, and U.~Kamilov, ``Inference for
  generalized linear models via alternating directions and {B}ethe free energy
  minimization,'' arXiv:1501.01797 [cs.IT], Jan. 2015.

\bibitem{Mezard-Montanari-2009}
M.~M{\'e}zard and A.~Montanari, \emph{Information, Physics, and
  Computation}.\hskip 1em plus 0.5em minus 0.4em\relax New York: Oxford
  University Press, 2009.

\bibitem{spasm}
K.~Sj{\"o}strand and B.~Ersb{\o}ll, ``{SpaSM}: A {M}atlab toolbox for sparse
  statistical modeling,'' \url{http://www2.imm.dtu.dk/projects/spasm/}.

\bibitem{Efron-Hastie-Johnstone-Tibshirani-2004}
B.~Efron, T.~Hastie, I.~Johnstone, and R.~Tibshirani, ``Least angle
  regression,'' \emph{Ann. Statist.}, vol.~32, no.~2, pp. 407--499, 04 2004.

\bibitem{Orszag-Bender-1978}
S.~A. Orszag and C.~M. Bender, \emph{Advanced Mathematical Methods for
  Scientists and Engineers}.\hskip 1em plus 0.5em minus 0.4em\relax
  McGraw-Hill, 1978.

\bibitem{Arfken-Weber-Harris-2013}
G.~B. Arfken, H.~J. Weber, and F.~E. Harris, \emph{Mathematical Methods for
  Physicists}, 7th~ed.\hskip 1em plus 0.5em minus 0.4em\relax Elsevier, 2013.

\bibitem{Goutis-Casella-1999}
C.~Goutis and G.~Casella, ``Explaining the saddlepoint approximation,''
  \emph{The American Statistician}, vol.~53, no.~3, pp. 216--224, 1999.

\end{thebibliography}

\bibliographystyle{IEEEtran}

\newpage

\begin{IEEEbiographynophoto}{Mikko~Vehkaper{\"a}}
	received the Ph.D.\ degree from Norwegian University of Science
	and Technology (NTNU), Trondheim, Norway, in 2010. 
	Between 2010--2013 he was a post-doctoral researcher at School of 
	Electrical Engineering, and the ACCESS Linnaeus Center, KTH Royal Institute 
	of Technology, Sweden, and 2013--2015 an Academy of Finland Postdoctoral 
	Researcher at Aalto University School of Electrical Engineering, Finland.  
	He is now a lecturer (assistant professor) at 
	University of Sheffield, Department of  Electronic and Electrical 
	Engineering, United Kingdom.
	He held visiting appointments at Massachusetts Institute of Technology (MIT), US, 
	Kyoto University and Tokyo Institute of Technology, Japan, and University of 
	Erlangen-Nuremberg, Germany.  His research interests are in the field of wireless 
	communications, information theory and signal processing. 
	Dr.\ Vehkaper{\"a} was a co-recipient for the Best Student Paper 
	Award at IEEE International Conference on Networks (ICON2011) and
	IEEE Sweden Joint VT-COM-IT Chapter Best Student Conference Paper Award 
	2015.
\end{IEEEbiographynophoto}
\vspace*{-3cm}
\begin{IEEEbiographynophoto}{Yoshiyuki~Kabashima}
	received the B.Sci., M.Sci., and Ph.D.\ degrees in physics from Kyoto University, 
	Japan, in 1989, 1991, and 1994, respectively. 
	From 1993 until 1996, he was with the Department of Physics, 
	Nara Women's University, Japan. In 1996, he moved to the Department 
	of Computational Intelligence and Systems Science, Tokyo Institute of Technology, 
	Japan, where he is currently a professor. His research interests 
	include statistical mechanics, information theory, and machine learning.  
	Dr.\ Kabashima received the 14th Japan IBM Science Prize in 2000, 
	the Young Scientist Award from the Ministry of 
	Education, Culture, Sports, Science, and Technology, Japan, in 2006, and 
	the 11th Ryogo Kubo Memorial Prize in 2007. 
\end{IEEEbiographynophoto}
\vspace*{-3cm}
\begin{IEEEbiographynophoto}{Saikat~Chatterjee}
	is an assistant professor and docent in the Dept of Communication Theory, KTH-Royal Institute of Technology, Sweden. He is also part of Dept of Signal Processing, KTH. Before moving to Sweden, he received Ph.D. degree in 2009 from Indian Institute of Science, India. He has published more than 80 papers in international journals and conferences. He was a co-author of the paper that won the best student paper award at ICASSP 2010. His current research interests are signal processing, machine learning, coding, speech and audio processing, and computational biology.
\end{IEEEbiographynophoto}
\enlargethispage{-11cm}

\end{document}